\documentclass[11pt]{article}
\usepackage[utf8]{inputenc}

\usepackage{amsmath,amsfonts,amssymb,amsthm}
\usepackage{fullpage}
\usepackage{latexsym}
\usepackage{xspace}
\usepackage{paralist}
\usepackage[bookmarks=true]{hyperref}
\usepackage{xy}
\usepackage{url}
\usepackage{tabularx}
\usepackage{graphicx}
\usepackage[usenames,dvipsnames]{color}

\input xy
\xyoption{all}

\DeclareGraphicsExtensions{.jpg,.png}

\newtheorem{theorem}{Theorem}[section]
\newtheorem{lemma}[theorem]{Lemma}
\newtheorem{corollary}[theorem]{Corollary}

\newtheorem{proposition}[theorem]{Proposition}

\newtheorem{observation}[theorem]{Observation}

\theoremstyle{definition}
\newtheorem{definition}[theorem]{Definition}

\newtheorem{openproblem}[theorem]{Open Problem}

\newtheorem{fact}[theorem]{Fact}
\newtheorem{remark}[theorem]{Remark}
\newtheorem{example}[theorem]{Example}
\newtheorem*{definition*}{Definition}
\newtheorem*{exercise*}{Exercise}

\theoremstyle{plain}

\newcommand{\rmkqed}{\nobreak\hfill$\lhd$}

\newcommand{\gpcls}{\mathfrak{C}}
\newcommand{\N}{\mathbb{N}}
\newcommand{\F}{\mathbb{F}}

\newcommand{\Z}{\mathbb{Z}}
\newcommand{\R}{\mathbb{R}}

\newcommand{\ie}{i.\,e.\xspace}
\newcommand{\eg}{e.\,g.}

\newcommand{\etal}{\emph{et al.}\xspace}
\newcommand{\Sec}[1]{Section~\ref{sec:#1}}

\newcommand{\cc}[1]{\mathsf{#1}\xspace} % operator for complexity class

\newcommand{\algprob}[1]{{\sc #1}\xspace}
\newcommand{\algprobm}[1]{\mbox{\sc #1}\xspace}
\newcommand{\GrI}{\algprobm{GraphI}}
\newcommand{\GpI}{\algprobm{GpI}}
\newcommand{\CGpI}{\algprobm{CayleyGpI}}

\newcommand{\cohiso}{{\sc Cohomology Class Isomorphism}\xspace}
\newcommand{\actcomp}{{\sc Action Compatibility}\xspace}
\newcommand{\cosetint}{\algprob{Coset Intersection}}
\newcommand{\lincode}{\algprob{Linear Code Equivalence}}
\newcommand{\modcyc}{\algprob{Module Cyclicity}}
\newcommand{\modiso}{\algprob{Module Isomorphism}}
\newcommand{\outactcomp}{{\sc Outer Action Compatibility}\xspace}
\newcommand{\edpc}{{\sc Extension Data Pseudo-congruence}\xspace}

\newcommand{\aut}{\mathrm{Aut}} % operator as AUTmorphism
\newcommand{\id}{\mathrm{id}} % the IDentity element of a group
\newcommand{\inn}{\mathrm{Inn}}

\newcommand{\Iso}{\mathrm{Iso}}
\newcommand{\Aut}{\mathrm{Aut}}
\newcommand{\Inn}{\mathrm{Inn}}
\DeclareMathOperator{\Out}{\mathrm{Out}}

\DeclareMathOperator{\im}{Im}
\renewcommand{\ker}{\text{Ker}}

\newcommand{\GL}{\mathrm{GL}}

\newcommand{\extension}[3]{#1 \hookrightarrow #2 \twoheadrightarrow #3}
\newcommand{\extensionl}[5]{#1 \stackrel{#2}{\hookrightarrow} #3 \stackrel{#4}{\twoheadrightarrow} #5}

\newcommand{\ISO}{\mathrm{Iso}}

\newcommand{\poly}{\mathrm{poly}}
\newcommand{\sym}{\mathrm{Sym}}

\DeclareMathOperator{\rad}{Rad}

\newcommand{\soc}{\mathrm{Soc}}

\newcommand{\defeq}{:=}

\newcommand{\charfn}{\mathcal{S}}

\newcommand{\pc}{\cong}

\newcommand{\pker}{\mathrm{Pker}}
\newcommand{\tb}[1]{M_{#1}}
\newcommand{\Tor}{\mathrm{Tor}}
\newcommand{\Hom}{\mathrm{Hom}}
\newcommand{\End}{\mathrm{End}}
\newcommand{\Ext}{\mathrm{Ext}}

\newcommand{\cocycleprod}{Z^2_{\text{prod}}(Q, A)}
\newcommand{\coboundaryprod}{B^2_{\text{prod}}(Q, A)}
\newcommand{\cochainprod}{C^2_{\text{prod}}(Q,A)}
\newcommand{\CodeEq}{\mathrm{CodeEq}}

 \newcommand{\jnote}[1]{\textcolor{red}{\small$\bullet$\footnote{\textcolor{red}{Josh:
 #1}}}}

\newtheorem*{quotlistthm}{Theorem~\hyperref[cor:quotient_list]{A} (=Corollary~\ref{cor:quotient_list})}
\newtheorem*{quotlistcor}{Corollary~\ref{cor:quotient_list2}}
\newtheorem*{quotlistgen}{Theorem~\hyperref[cor:quotient_list_general]{B} (=Corollaries~\ref{cor:quotient_list_general} and \ref{cor:quotient_list2_general})}
\newtheorem*{prodthm}{Theorem~\hyperref[thm:ecentradwsoc]{C} (=\ref{thm:ecentradwsoc})}
\newtheorem*{splitprop}{Facts~\ref{fact:nec_actcomp}, \ref{fact:nec_cohiso}, and Lemmas~\ref{lem:main}, \ref{lem:main_general}}
\newtheorem*{lemma*}{Lemma}

\newcommand{\cosetComment}[1]{#1}

\title{Algorithms for group isomorphism via\\
group extensions and cohomology
\footnote{The introduction may serve as an extended abstract: \Sec{approach} 
contains an informal exposition of \Sec{prel_framework}, \Sec{framework} and 
\Sec{strategy}; \Sec{outline} gives a brief overview of \Sec{prel_concrete}, 
\Sec{autq} and \Sec{centrad_bb}. An extended abstract based on the introduction 
appeared as \cite{GQccc}.} \\
}
\author{Joshua A. Grochow
\footnote{Department of Computer Science, University of Colorado, Boulder and the Santa Fe Institute. \texttt{joshua.grochow@cs.colorado.edu}}
\and
Youming Qiao
\footnote{Centre for Quantum Software and Information, University of 
Technology Sydney. \tt{youming.qiao@uts.edu.au}}
}
\date{\today}

\begin{document}
\hypersetup{pageanchor=false}

\maketitle

\pagenumbering{roman}
% !TEX root = main.tex

\begin{abstract}

The isomorphism problem for finite groups of order $n$ (\GpI) has long been known 
to be solvable in $n^{\log n+O(1)}$ time, but only recently were polynomial-time 
algorithms designed for several interesting group classes. %including among 
%others, groups with no abelian normal subgroups (Babai \etal, ICALP 2012), groups 
%with abelian Sylow towers (Babai--Qiao, STACS 2012), and quotients of 
%generalized Heisenberg groups (Lewis--Wilson, Groups - Complex. - Cryptol.
%2012).
Inspired by recent progress, we revisit the strategy for \GpI via the 
extension theory of groups. 
%\jnote{Deleted references, moved the preceding sentence to 1st para. It seemed 
%like citing everyone here felt more like an introduction than an abstract, and we 
%also left out the genus 2 paper.}

The extension theory describes how a normal subgroup 
$N$ is related to $G/N$ via $G$, and this naturally leads to a divide-and-conquer 
strategy that ``splits'' \GpI into two subproblems: one regarding group actions on 
other groups, and one regarding group cohomology. When the normal subgroup $N$ is 
abelian, this strategy is well-known. Our first contribution is to extend this 
strategy to handle the case when $N$ is not necessarily abelian. This allows 
us to provide a unified explanation of all recent polynomial-time algorithms for 
special group classes. 

Guided by this strategy, to make further progress on \GpI, we consider 
\emph{central-radical} groups, proposed in Babai \etal (SODA 2011): the class of 
groups such that $G$ mod its center has no abelian normal subgroups. This class is 
a natural extension of the group class considered by Babai \etal (ICALP 2012), 
namely those groups with no abelian normal subgroups. Following the above 
strategy, we solve \GpI in $n^{O(\log \log n)}$ time for central-radical groups, 
and in polynomial time for several prominent subclasses of central-radical groups. 
We also solve \GpI in $n^{O(\log\log n)}$ time for groups whose solvable normal 
subgroups are elementary abelian but not necessarily central. As far as we are 
aware, this is the first time there have been worst-case guarantees on a 
$n^{o(\log n)}$-time algorithm that tackles both aspects of \GpI---actions and 
cohomology---simultaneously.

Prior to this work, the best proven upper bounds on algorithms for groups with 
central radicals were $n^{O(\log n)}$, even for groups with a central radical of 
constant size, such as $\rad(G) = Z(G)=\Z_2$.
To develop our new algorithms we utilize several mathematical results on the 
detailed structure of cohomology classes, as well as algorithmic results for code 
equivalence, coset intersection and cyclicity testing of modules over 
finite-dimensional associative algebras. 
We also suggest several promising directions for future work. 
%\jnote{Moved old abstract into journal/old/abstract-old.tex}
\end{abstract}

%\thispagestyle{empty}
%\clearpage
%\setcounter{page}{1}
\newpage

\setcounter{tocdepth}{2}
\tableofcontents

\clearpage
\hypersetup{pageanchor=true}

\setcounter{page}{0}
\pagenumbering{arabic}

%\addtocontents{toc}{\setcounter{tocdepth}{1}}
% !TEX root = main.tex
\section{Introduction}

The group isomorphism problem (\GpI) is to determine whether two finite groups %, given by their multiplication tables (``Cayley tables''), 
are isomorphic. For groups of order $n$, the easy $n^{\log n + O(1)}$-time 
algorithm \cite{FN,Mil78}\footnote{Miller \cite{Mil78} attributes this algorithm 
to Tarjan.} 
for the general case of \GpI has barely seen any asymptotic 
improvement over the past four decades; it was improved recently to $n^{1/4 \log n 
+ O(1)}$ by Rosenbaum \cite{Rosen2} (see \cite[Sec. 2.2]{GR16}), but even the extensive 
body of work on practical 
algorithms %by the  Computational Group Theory community---
led by Eick, Holt, Leedham-Green and O'Brien (\eg, \cite{BEO02, ELGO02, BE99, CH03})---resulting in most of the 
functional algorithms in use today---was recently found \cite{wilsonConf} 
%\footnote{\label{fn:CGT}This was found at  the 2014 conference on \emph{Groups, Computation, and Geometry} at Colorado State University, co-organized by P. Brooksbank, A. Hulpke, T. Penttila, J. Wilson, and W. Kantor; personal communication from J. Wilson.} 
to only improve the constant in 
the exponent, still resulting in a $n^{\Theta(\log n)}$-time algorithm for the 
general case. The past few years have witnessed a resurgence of activity on 
worst-case guaranteed algorithms for this problem \cite{Gal09,BCGQ,QST11,Wag11,LW12,BQ,BCQ,Rosen,Rosen2,BMW15,GQ15}.

Before introducing these works and our results, we recall why group isomorphism is 
an intriguing problem from the complexity-theoretic perspective, even when the 
groups are given by their multiplication tables. %\jnote{Added:} 
(See Remark~\ref{rmk:CGT} below.) We call this version of the problem \CGpI. %\footnote{Our results are stated in terms of 
%the runtime as a function of $n=|G|$, but the algorithms themselves are fairly 
%agnostic to the method of input. That is, when the runtime of an algorithm is at 
%least quadratic in the order of the group, even if the group is input as a black 
%box, the algorithm may first enumerate the entire multiplication table, and then 
%proceed from there. For the few cases in this paper when the distinction of input 
%type does matter, we denote the version of \GpI when the multiplication tables 
%(``Cayley tables'') are given as input by \CGpI.}
As \CGpI reduces to \algprob{Graph Isomorphism} (\GrI) (see, \eg, the book \cite{struct}), \CGpI currently has an intermediate status: It is not $\cc{NP}$-complete unless $\cc{PH}$ collapses \cite{BHZ,babai-moran}, and is not known to be in $\cc{P}$.  In addition to its intrinsic interest, resolving the exact complexity of \GpI is a tantalizing question. Further, there is a surprising connection between \CGpI and the Geometric Complexity Theory program (see, \eg, \cite{gct_acm} and references therein): Techniques from \CGpI were used to solve cases of \algprob{Lie Algebra Isomorphism} that have applications in Geometric Complexity Theory \cite{grochowLie}.

In a survey article \cite{Babaisurvey} in 1995, after enumerating several 
isomorphism-type problems including \GrI and \GpI, Babai expressed the belief that 
\CGpI might be the only one expected to be in $\cc{P}$.
%\jnote{Removed exact Babai quotation in effort to reduce footnotes}
%\footnote{The exact quotation from Babai's 1995 survey \cite{Babaisurvey} is: 
%``None of the problems mentioned in this section, with the possible exception of 
%isomorphism of groups given by a Cayley table, is expected to have polynomial 
%time 
%solution.''} 
Indeed, 
in many ways \CGpI seems easier than \GrI: There is a simple $n^{\log 
n+O(1)}$-time algorithm for \GpI, whereas the best known algorithm for \GrI takes 
time $2^{(\log n)^c}$ for some $c \geq 3$ \cite{Bab16} and is quite complicated.\footnote{At the time of writing, the paper \cite{Bab16} is still under peer 
review. The previous-best algorithm took time $2^{O(\sqrt{n \log n})}$ (see 
\cite{BL83}), and was also quite complicated.}
%\jnote{Added footnote to address Comment D.}
There is a polynomial-time reduction from \CGpI to \GrI, yet there is provably no 
$\cc{AC}^0$ reduction in the opposite direction \cite{CTW10}. 
The reduction $\CGpI \leq \GrI$ means that \CGpI stands as an obstacle to putting 
\GrI into $\cc{P}$; in light of the recent quasi-polynomial-time algorithm for 
\GrI \cite{Bab16}, this obstacle has become much more salient (the previous best 
algorithm for \GrI was so far from quasi-polynomial that there were clearly 
obstacles to be overcome before \CGpI became a serious obstacle, but many of those 
obstacles have now been overcome).
Further, \GrI is as hard as its counting version, whereas no such counting-to-decision reduction is known for \GpI. Finally, whereas the smallest standard complexity class known to contain \GrI is $\cc{NP} \cap \cc{coAM}$, Arvind and Tor\'{a}n \cite{AT11} showed that \CGpI for solvable groups is in $\cc{NP} \cap \cc{coNP}$ under a plausible assumption, weaker than that needed to show $\GrI \in \cc{coNP}$ (recall that a group is solvable if all its composition factors are abelian, or equivalently if the derived series $G_0 = G$, $G_{i+1} = [G_i, G_i]$ terminates in the identity).

Despite this situation and considerable attention to \GpI, prior to 2009 the 
actual developments towards algorithms %\jnote{Rephrased in a more CGT-friendly 
%way} 
with worst-case guaranteed running time polynomial in $|G|$ %worst-case guaranteed 
%polynomial-time algorithms for \CGpI 
essentially stopped at abelian groups, although there have been impressive 
practical advances %\jnote{Deleted: in Computational Group Theory} 
(see, \eg, the theses 
\cite{smith,howden} and references therein for nice overviews). 
%\jnote{Updated as you suggested:} 
Isomorphism of abelian groups has long been known to be solvable in polynomial 
time \cite{Sav80, Ilio85, Vik96, Kav07}. 
%For abelian groups, Kavitha exhibited an $O(n)$-time algorithm \cite{Kav07}, improving on  previous polynomial-time algorithms \cite{Sav80, Vik96, Ilio85}\footnote{Iliopoulous \cite{Ilio85} gave an algorithm that decides isomorphism of abelian groups given by generating sets of size at most $\tilde{O}(\sqrt{|G|})$ in time $\tilde{O}(\sqrt{|G|})$, where the $\tilde{O}$ hides poly-logarithmic factors. For groups given by generating sets of this size, Iliopoulos's algorithm is better than Kavitha's. However, when given a 
%Cayley table as input, the best algorithm we know of for computing such a small 
%generating set takes time $O(|G| \log |G|)$, so to even apply Iliopoulos's 
%algorithm in the Cayley table model seems to require pre-processing that takes 
%time $O(|G|\log|G|)$. If one merely takes the entire group as the generating set, 
%in order to avoid this pre-processing, the first step of Iliopoulos's algorithm is 
%to compute the orders of every element of generating set; computing the orders of 
%every element of a group in $O(|G|)$ time was not known until Kavitha's paper. In 
%contrast, Kavitha's algorithm \cite{Kav07} works directly from the Cayley table 
%and decides isomorphism in $O(|G|)$ time.}\ynote{How about we just say that 
%abelian groups are long known to be polynomial-time solvable, and list these 
%references as such? Readers probably would not get into these papers any way and 
%we are not interested in faster abelian algorithms either...}
The next natural group class after abelian groups---class 2 nilpotent groups, 
whose quotient by their center is abelian---turns out to be 
formidable %\jnote{Replaced footnote with:} 
(see \cite{GZ, LW12, BMW15} for efficient algorithms in some restricted cases).

Beginning in 2009 there were several advances on worst-case guaranteed algorithms, 
starting with Le Gall \cite{Gal09}. In \cite{BCQ}, following \cite{BCGQ}, Babai 
\etal developed a polynomial-time algorithm for groups with no abelian normal 
subgroups. This suggests the presence of abelian normal subgroups as a bottleneck. 
%\jnote{Removed footnote in attempt to reduce footnotes.} 
%%%\footnote{\label{fn:indecomp} Abelian direct factors, as in $H \times \Z_n$ are 
%%not a bottleneck however: Kayal and Nezhmetdinov \cite{KN09} and Wilson 
%%\cite{Wil10} gave polynomial-time algorithms to decompose a direct product into 
%%its direct factors. For polynomial-time algorithms, one may thus assume that the 
%%groups under consideration are \emph{directly indecomposable}: they cannot be 
%%written as a direct product of two nontrivial groups.} 
With this in mind, Babai and Qiao \cite{BQ} developed a polynomial-time algorithm for a special class %\footnote{The group class is referred to as groups with abelian Sylow towers in \cite{BQ}, but we will see that in our context, it corresponds to coprime extensions of abelian groups.}
of non-nilpotent solvable groups, building on \cite{Gal09, QST11}; this was 
recently extended by the present authors to the so-called groups of tame 
extensions \cite{GQ15}. In \cite{LW12}, Lewis and Wilson made intriguing progress 
on $p$-groups: They gave a polynomial-time algorithm for quotients of generalized 
Heisenberg groups, a decently large subclass of $p$-groups of class 2. 
%In 2013, 
%Rosenbaum \cite{Rosen, Rosen2} exhibited a deterministic $n^{0.25\log n+o(\log 
%n)}$-time algorithm for solvable groups, developing ideas of Wagner \cite{Wag11}. 
%Furthermore, Rosenbaum \cite{Rosen2} developed a general algorithmic technique 
%that brings the time complexity of \GpI to $n^{0.5\log n+O(1)}$. 
Rosenbaum's recent works 
\cite{Rosen,Rosen2} (some of them building on ideas of 
Wagner \cite{Wag11}) lead to an $n^{1/4\log n+O(1)}$-time algorithm for \GpI 
(see \cite[Sec. 2.2]{GR16}).
To summarize, at 
present it is crucial to understand 
indecomposable %\jnote{Removed footnote} %\textsuperscript{\ref{fn:indecomp}} 
groups with abelian normal  subgroups to develop $n^{o(\log n)}$-time algorithms. 

Given these developments, we are at an interesting crossroads: First, as several 
nontrivial algorithms with worst-case guarantees have recently been developed 
%\jnote{Added parenthetical CGT-friendly remark} 
(or worst-case guarantees proven for previously known algorithms), it is 
reasonable to reflect back to see if there is some common pattern or structure to 
these results. Second, of course, we should continue to improve the state of the 
art by developing more $n^{o(\log n)}$-time algorithms for special group classes. 
Finally, class 2 nilpotent groups seem to remain the bottleneck, but despite 
heuristic evidence, it is still desirable to formalize a reduction from the 
general case to this seeming bottleneck.

In this paper we contribute to all three of the preceding aspects. %Specifically, we propose a general strategy for group isomorphism, which helps to explain in a unified way the recent successes on special group classes. We then follow this strategy to develop $n^{O(\log \log n)}$-time algorithm for a group class proposed in \cite{BCGQ}, and polynomial-time algorithms for some prominent subclasses. 
Our contributions are twofold: %(1) we propose a general strategy for group isomorphism; 
(1) we show how a general strategy for group isomorphism from the mathematics literature can be used to bound the worst-case complexity; 
and (2) using that strategy, we develop an $n^{O(\log \log n)}$-time algorithm for 
a group class proposed in \cite{BCGQ}, and polynomial-time algorithms for some 
prominent subclasses thereof. The worst-case analysis of this strategy also helps 
to explain in a unified way the recent successes on other group classes 
\cite{Gal09,QST11,BQ,BCGQ,BCQ,LW12,GQ15}, which can be viewed as adding 
class-specific tactics to the general strategy. % outlined here. 
We also explain how these results may help to reduce general \GpI to the class 2 nilpotent case.

%\ynote{I am somewhat concerned that this remark, as written in its 
%current form, may bring us more vulnerable to attacks from the CGT side. This 
%remark seems to suggest we take certain responsibility for the potential practical 
%implementation of the methods in this paper. However we probably should only state 
%this after some efforts are made or some evidence appears. For me, I would just 
%list the following points for CGT readers to bear in mind while reading our paper: 
%(1) the current inability to improve to $|G|^{o(\log |G|)}$ suggests that this is 
%a potentially deep research topic; (2) aiming at bringing to 
%$|G|^{O(1)}$ can be seen a first step towards understanding the difficulties; (3) 
%interaction between worst-case analysis and improving practical algorithms has 
%been witnessed in e.g. Seress's work and it is our sincere hope that this great 
%tradition may continue in this topic as well.  }
%\jnote{Agreed. I added the paragraph from email as we had discussed, and updated 
%the final paragraph of this remark to be in line with your suggestion.}
\begin{remark}[On efficient implementation versus computational complexity] \label{rmk:CGT}
There are naturally two audiences for this paper, who might view it quite 
differently: (A) computational complexity theorists / algorithms theory 
researchers, and (B) computational group theorists. %\ynote{Deleted: (by which we mean researchers focused on practical implementations of algorithms for groups as in the software {\sf Magma}~\cite{magma} and {\sf GAP}~\cite{GAP}, as well as the surrounding math)}. 
To some in the latter group, much of the content of this paper is 
surely well-known, and to them perhaps not even worth writing down at this point 
in history. To them, we would like to highlight our results on nonabelian 
cohomology (Section~\ref{sec:framework}) and the use of 
Guralnick--Kantor--Kassabov--Lubotzky \cite{GKKL}, which we believe may be new 
even in light of the large body of work in that community. We would also like to 
point out that the remainder of the paper will likely seem new to computational 
complexity theorists, despite seeming trivial to computational group theorists. As 
so often happens when ideas from one area B (in this case, classical and 
computational group theory) are imported to another area A (computational 
complexity), area B may view the results as trivial while area A may view them as 
a nontrivial advance. However, while much of this paper is targeted at 
computational complexity theorists and so has this flavor, we believe that some of 
the results (as mentioned above) will be new to both communities.

As discussed above, there are good reasons (e.\,g., related to \GrI) to be 
interested in the getting the asymptotic complexity of \GpI down to polynomial in 
the order of the group, independently of its (ir)relevance to practical 
implementation. Combined with the fact that all of the practical algorithms put 
together still leave the worst-case bound at $|G|^{\Theta(\log |G|)}$ 
\cite{wilsonConf}, we thus state our runtime bounds as a function of $n=|G|$. 
Nonetheless, our algorithms are fairly agnostic to the method of input, and in 
particular \emph{do not depend on the input being a Cayley table}. That is, when 
the runtime of an algorithm is at least quadratic in the order of the group, even 
if the group is input as a black box, the algorithm may first enumerate the entire 
multiplication table, and then proceed from there. 
%\jnote{Removed this sentence, as the only place we use \CGpI is in the 
%introduction, above here!} %If the distinction of the input type does matter, we 
%%denote the version of \GpI when the multiplication tables (``Cayley tables'') 
%%are 
%%given as input by \CGpI.

Except in a 
few cases, we do not bother to estimate multiplicative 
constants, nor even constant exponents, in most of our algorithms. The current 
inability to get a general algorithm with run-time $|G|^{o(\log |G|)}$ 
\cite{wilsonConf} suggests that this is a potentially deep research topic, and 
getting algorithms with guaranteed run-time $|G|^{O(1)}$ can be seen as a first 
step in understanding the difficulties involved (not to mention its 
complexity-theoretic interest). Furthermore, there is non-trivial precedent for 
synergistic interaction between worst-case analyses and practical algorithms, as 
reflected in, e.\,g., A. Seress's pioneering implementations in {\sf GAP} \cite{GAP} of 
theoretically-proven fast algorithms for permutation groups (see, e.\,g., \cite{seressbook}). Therefore, although we do 
not deal with issues of practical algorithms in this paper, it is our sincere hope 
that this great tradition will continue in future work.
%We believe that in most cases their estimation is routine, the numbers are relatively small, and could probably be improved by using the standard methods that are used today in implementations in {\sf Magma}~\cite{magma} and/or {\sf GAP}~\cite{GAP}, perhaps even to the point of yielding new, practically efficient algorithms in many cases.
\end{remark}

\subsection{Main results} \label{sec:mainresults}
The classes of groups we consider are natural extensions of the class of groups 
considered in \cite{BCGQ, BCQ}, and are additionally motivated by the Babai--Beals 
filtration \cite{BB99} (resurrecting certain ideas 
that go back to Fitting \cite{Fit33}), and the Cannon--Holt 
approach to group isomorphism in 
the practical setting \cite{CH03}. We go into the details of the Babai--Beals 
filtration and the Cannon--Holt approach in \Sec{future}. Here we merely give 
enough of a flavor to help motivate the classes of groups we consider.

Important in both the Babai--Beals filtration and the Cannon--Holt approach is the 
solvable radical. 
%\jnote{Deleted, because we say this just a few pages earlier: Recall that a group 
%is solvable if it has a series of subgroups $1 = G_0 \unlhd G_1 \unlhd \dotsb 
%\unlhd G_k = G$ such that each $G_i$ is normal in the next and $G_i / G_{i-1}$ is 
%abelian for all $i$.} 
The solvable radical $\rad(G)$ of a group $G$ is the unique maximum solvable 
normal subgroup of $G$. Note that the center $Z(G)$, as an abelian normal 
subgroup, is contained in $\rad(G)$. $G/\rad(G)$ contains no solvable normal 
subgroups, side-stepping the currently intractable obstacle of solvable groups. 
Babai \etal \cite{BCQ} give a polynomial-time algorithm for isomorphism of groups 
with no solvable normal subgroups (equivalently, no abelian normal 
subgroups); %\jnote{Added parenthetical remark and removed subsequent footnote}; 
following them, we call such groups ``semisimple.''%\footnote{Semisimple groups 
%can be characterized either as having no solvable normal subgroups or as having 
%no 
%abelian normal subgroups, for if there is a solvable normal subgroup $S \unlhd 
%G$, 
%there is an abelian normal subgroup of $G$, namely the last term in the derived 
%series of $S$.}

We mainly consider the class of groups whose solvable radical coincides with its 
center, that is, $\rad(G) = Z(G)$; in \Sec{abelian} we also consider groups whose 
solvable radical is abelian, but need not be contained in the center. The former 
class, which we refer to as \emph{groups with central radicals} or 
\emph{central-radical groups} for short, is a natural 
extension of the class of 
semisimple groups and a natural stepping stone towards general groups (see 
Figure~\ref{fig:groups} below). Note that for such groups the solvable radical is 
necessarily abelian. Besides the motivations mentioned above, central-radical 
groups also cover a class of groups that is well-studied in finite group theory 
(see %\jnote{Deleted: \Sec{mainresults}} and 
Appendix~\ref{app:fitting}). In the theory of Lie 
groups, central-radical groups correspond to the well-studied and important class 
of reductive Lie groups, which are important throughout mathematics and physics, 
often because of their nice representation-theoretic properties.

We use the strategy outlined in \Sec{approach} below to achieve the following results. 
For groups with central radicals, we give an $n^{O(\log \log n)}$-time algorithm 
in general, and for several subclasses of groups with central radicals we give 
polynomial-time algorithms. We also give similarly efficient algorithms for groups 
with elementary abelian, but not necessarily central, radicals. Prior to this 
work, the best proven upper bounds on algorithms for groups with central radicals 
were $n^{O(\log n)}$, even for groups with a central radical of constant size, 
such as $\rad(G) = Z(G)=\Z_2$.
%\ynote{Removed the sentence: ``Prior to this work, nothing better than an 
%$n^{O(\log n)}$-time algorithm was known, even for groups with a central radical 
%of constant size, such as $Z(G)= \rad(G) = \Z_2$.''} \jnote{See my comment on 
%this in the abstract.}

\cosetComment{Recall that for any groups $G$ and $H$, the set of isomorphisms between them is either empty or a coset of the automorphism group $\aut(G)$ in the group of permutations of the disjoint union $G \sqcup H$. We say that the coset of isomorphisms can be found if one isomorphism $G \stackrel{\cong}{\to} H$  and a generating set for $\aut(G)$ can be found. Finding the full coset of isomorphisms---rather than just deciding \GpI or finding a single isomorphism---is often useful in recursively building algorithms for larger group classes from those for smaller classes.}

\begin{quotlistthm}
Isomorphism of central-radical groups of order $n$ can be decided in time $n^{c \log \log n + O(1)}$, for $c = 1/\log_2(60) \approx 0.169$. \cosetComment{Furthermore, if the radical is elementary abelian, the coset of isomorphisms can be found in the same time bound.}
\end{quotlistthm}

The algorithm in the above theorem in fact runs in polynomial time when the order or structure of the semisimple quotient $G/\rad(G)$ is bounded as follows. Recall that a normal subgroup of $G$ is \emph{minimal} if it is nontrivial and does not contain any smaller normal subgroups of $G$. The number of minimal normal subgroups of $G/\rad(G)$ is always at most $\log_{60} n$; if it happens to be just slightly smaller, then we have:

\begin{quotlistcor}
Let $G$ and $H$ be central-radical groups of order $n$. If $G/\rad(G)$ has $O(\frac{\log n}{\log \log n})$ minimal normal subgroups, isomorphism between $G$ and $H$ can be decided in $\poly(n)$ time. \cosetComment{Furthermore, if the radical is elementary abelian, the coset of isomorphisms can be found in the same time bound.} 
\end{quotlistcor}
In particular, this includes groups $G$ satisfying $|G / \rad(G)| \leq n^{O(1 / \log \log n)}$, but also many groups where $G/\rad(G)$ is much larger. Both of these theorems are in fact corollaries of our more general Theorem~\ref{thm:quotient_list} together with previous results on semisimple groups \cite{BCGQ}, but we defer the statement of Theorem~\ref{thm:quotient_list} until \Sec{autq}, as the above results make its significance clearer.

For groups with elementary abelian radicals---even if they are not central---we get the same conclusions. This requires us to simultaneously solve \actcomp and \cohiso. We combine the above techniques with a novel reduction to known representation-theoretic algorithms \cite{CIK97} to get:

\begin{quotlistgen}
Isomorphism of groups of order $n$ with elementary abelian radicals can be decided\cosetComment{, and the coset of isomorphisms found,} in time $n^{c \log \log n + O(1)}$, for $c = 1/\log_2(60) \approx 0.169$.

If furthermore $G/\rad(G)$ has $O(\frac{\log n}{\log \log n})$ normal subgroups, isomorphism can be decided\cosetComment{, and the coset of isomorphisms found,} in $\poly(n)$ time.
\end{quotlistgen}

We then consider central-radical groups with $G/\rad(G)$ a direct product of 
nonabelian simple groups. Although this may seem restrictive, this class of groups 
is quite natural. In group theory, this class is closely related to the 
generalized Fitting subgroups (see, \eg, \cite[Ch.~6, \S6]{Suzuki2} and 
\cite[Ch.~11]{aschbacher}, as well as Appendix~\ref{app:fitting}). Also, within 
central-radical groups, this class has two characterizations: (1) the last two of 
the four levels of the Babai--Beals filtration are trivial (see \Sec{babaibeals}); 
or (2) those groups that are equal to their generalized Fitting subgroup (see 
Appendix~\ref{app:fitting}). Our next result gives
polynomial-time algorithms for this group class, which includes, %\jnote{Made this 
%accurate}
% when certain parameters are fixed. This includes, 
for example, central extensions of $\Z_p^{\Theta(\log n)}$ by $A_k^{\Theta(\log 
n)}$, which do not satisfy the conditions of the results above.
% \jnote{Deleted: Corollaries~\ref{cor:quotient_list2} nor 
%\ref{cor:quotient_list2_general}} 

\begin{prodthm}
Isomorphism between two groups $G_1, G_2$ with central, elementary abelian radicals can be decided, and the coset of isomorphisms found, in $\poly(|G_i|)$ time if either:
\begin{enumerate}
\item $G_1/\rad(G_1)$ is a direct product of simple groups%
; or
\item $G_1/\rad(G_1)$ is a direct product of perfect groups, each of order $O(1)$.
\end{enumerate}
\end{prodthm}

More importantly, we believe the techniques that go into proving this theorem 
%\jnote{Deleted: Theorem~\ref{thm:ecentradwsoc}} 
are 
worth noting: we rely on a detailed analysis of the structure of the cohomology 
classes specific to this group class (see \Sec{prep_coho}, and the use of the powerful results of \cite{GKKL}) %\jnote{Removed Suzuki 
%reference since we also now use GKKL} %; these can also be found in Suzuki's book 
%%\cite{Suzuki2}), 
to allow for the application of known algorithmic techniques, including singly-exponential-time algorithms for \lincode \cite{BabCode} (see \cite[Thm.~7.1]{BCGQ}) and \cosetint \cite{Bab83, Luk99} (see also \cite{Bab10, BKL83}). 

%\jnote{Deleted repetitive paragraph}
%\ynote{Some repetition with the paragraph before Theorem 
%6.2.}\cosetUncomment{Recall that for any groups $G$ and $H$, if we assume (without 
%loss of generality) that they are both group structures on the same underlying 
%set, the set of isomorphisms between them is either empty or a coset of the 
%automorphism group $\aut(G)$ in the group of permutations of the underlying set. 
%We say that the coset of isomorphisms is known---or can be found---if one 
%isomorphism $G \stackrel{\cong}{\to} H$  and a generating set for $\aut(G)$ is 
%known. Finding the full coset of isomorphisms---rather than just deciding \GpI or 
%finding a single isomorphism---is often useful in building algorithms for larger 
%group classes from those for smaller classes.}
%

%\jnote{deleted: In Proposition~\ref{prop:warmup}, we show a similar result, with 
%``elementary abelian'' replaced by ``$|\aut(\rad(G))| \leq \poly(|G|)$.''}

\subsection{Motivation for the classes of groups considered} \label{sec:motivation}
Aside from the motivations already mentioned above, Figure~\ref{fig:groups} gives the general idea of where this paper fits in the 
picture of a larger approach towards putting \GpI into $\cc{P}$. The figure is 
neither complete nor 100\% accurate in terms of the landscape of groups and 
algorithms for \GpI, but is more or less correct %\jnote{Added "for algorithms w/ 
%worst-case guarantees"} 
for algorithms with worst-case guarantees at a large scale. %Terms in the figure are explained below.

\begin{figure}[!htbp]
\hspace{-0.25in}\includegraphics[width=7in,keepaspectratio=true]{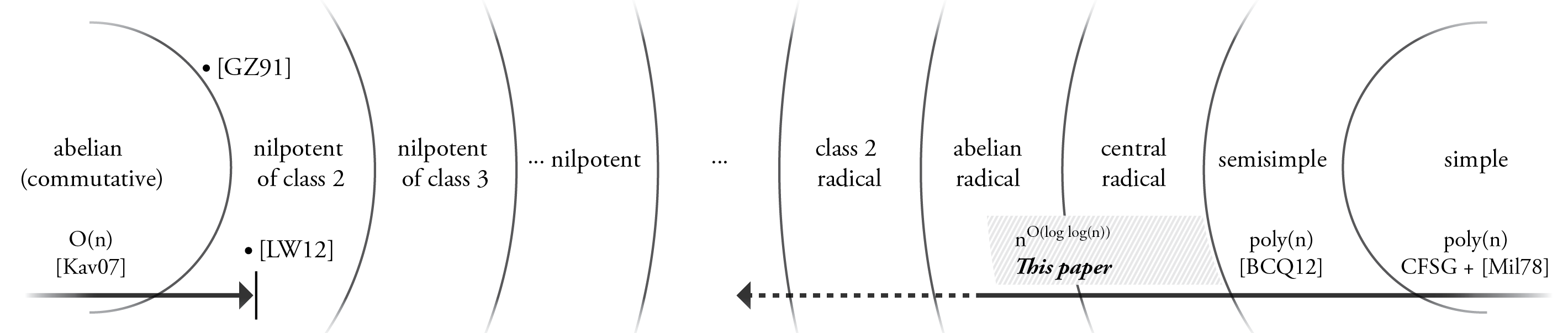}
\caption{Some progress on group isomorphism. The arrow on the left indicates that the techniques for abelian groups hit a wall before class 2 nilpotent groups. The arrow on the right indicates the recent progress at the ``opposite end of the spectrum,'' including this paper and prospects for future progress.}
\label{fig:groups}
\end{figure}

While central-radical groups may seem to be only a slight extension of semisimple 
groups, they in fact differ significantly from previous classes of groups with 
$n^{o(\log n)}$-time isomorphism algorithms. In particular, most previous $n^{o(\log n)}$-time 
algorithms for \GpI of special group classes only consider one of the two main 
aspects of \GpI, namely actions (in \Sec{previous}, we briefly indicate how 
actions are used in \cite{Gal09,QST11,BQ,BCGQ,BCQ}). On the other hand, to work 
with groups with central radicals, we need to focus on the other main aspect of 
the problem, namely cohomology (see \Sec{approach}). 

Our results also suggest one more step towards a formal reduction from the general 
case to nilpotent groups of class 2. In particular, in 
Proposition~\ref{prop:action_triv} and 
Remark~\ref{remark:action_triv}, we show that for groups $G$ where the action of 
$G$ on $\rad(G)$ by conjugation is 
essentially trivial (technically: the action is by inner automorphisms of 
$\rad(G)$) and $\Aut(Z(\rad(G)))$ is small, \GpI reduces to isomorphism of central-radical groups and isomorphism 
of solvable groups, separately. Thus, if isomorphism of central-radical groups 
could be decided in polynomial time, then \GpI for this class of groups would 
reduce to \GpI for solvable groups. Note that central-radical groups arise here 
naturally, by considering cases in which the relationship between $G/\rad(G)$ and 
$\rad(G)$ is simple. This is just one example of how we believe our ideas point 
towards a possible reduction from the general case to nilpotent groups; some other 
ideas in this paper may also be useful in this regard (see \Sec{cannonHolt}).
\subsection{A strategy via group extensions and cohomology} \label{sec:approach}
In this paper we use the theory of group extensions (which we describe briefly in 
this section, and in more detail in \Sec{prel_framework}, %\jnote{Added pointer to 
%appendix.} 
and give expository preliminaries in Appendix~\ref{app:gentle}; for textbook 
treatments, see, \eg, \cite[Chapter 11]{Rob} and \cite[Chapter 7]{Rot}) to show 
that the group isomorphism problem ``splits'' into two subproblems---one coming 
from actions of groups on other groups (\actcomp), and the other coming from group 
extensions and cohomology (\cohiso), which we explain below. We note that Besche 
and Eick have proposed this splitting in a slightly different setting, under the 
name ``strong isomorphism'' \cite{BE99}. In the abstract theory of finite groups 
this splitting is standard material; the contribution here is the observation that 
this standard material can %\jnote{Rephrased in a more CGT-friendly way} 
be used to 
prove worst-case algorithmic guarantees
%be made algorithmically effective, 
and that doing so is useful and even formally necessary to resolve the complexity 
of $\GpI$. %\jnote{Added:} 
We also extend this approach to the setting where the normal subgroup can be 
non-abelian. For the converse direction, we observe that special cases of these 
subproblems reduce to $\GpI$
under polynomial-time reductions (\Sec{necessary}). We summarize these results in:

\begin{splitprop}[``Splitting'' \GpI into actions and cohomology]\label{prop:splitting}
\quad
\begin{itemize}
\item For coprime extensions  \actcomp $ \equiv_m^p \GpI$.
\item For $p$-groups of class 2, when $p 
> 2$, \cohiso $ 
\equiv_m^p \GpI$.
\item \GpI for groups with a normal subgroup from one class of groups 
$\mathcal{N}$ and a quotient from a class of groups $\mathcal{Q}$ reduces to 
solving \GpI in $\mathcal{N}$, solving it in $\mathcal{Q}$, and simultaneously 
solving %\jnote{Moved footnote to main text.}%\footnote{Both problems have the 
%%same space of potential solutions, namely certain automorphisms of certain 
%%groups; 
%%a simultaneous solution is an automorphism that is simultaneously a solution for 
%%both problems. See \Sec{outline} and Lemma~\ref{lem:main} for more details.} 
\actcomp and \cohiso.
\end{itemize}
\end{splitprop}

A ``simultaneous solution'' to these two problems is possible because they have the same space of potential solutions, namely certain automorphisms of certain groups; a simultaneous solution is an automorphism that is simultaneously a solution for both problems. See \Sec{outline} and Lemma~\ref{lem:main} for more details.

Most previous complexity-theoretic results on \GpI have focused on some combination of algorithmic techniques and \actcomp. In this paper, for the first time from the worst-case complexity perspective, we make progress on \cohiso.

We now explain this ``splitting'' and the problems mentioned above informally. 
Consider the following natural strategy for testing whether $G$ is isomorphic to 
$H$. If $G$ is simple, then isomorphism can be tested in polynomial time as $G$ is 
generated by at most two elements (Fact~\ref{fact:genby2}). If $G$ is not simple, 
then it has some normal subgroup $N \unlhd G$, and we may try to use a 
divide-and-conquer strategy by first solving the isomorphism problem for $N$ and 
$G/N$. However, even if we find $M \unlhd H$ such that $N \cong M$ and $G/N \cong 
H/M$, this is typically not sufficient to conclude that $G \cong H$: \eg, $\Z_4$ 
and $\Z_2\times \Z_2$ %\jnote{Changed to be more CC-friendly, as we haven't yet 
%introduced the notion of extension} 
both have $\Z_2$ as a normal subgroup, with corresponding quotient $\Z_2$. %are 
%both extensions of $\Z_2$ by $\Z_2$. 
We must then understand how the groups $N$ and $G/N$ ``glue'' back together to get 
$G$. $G$ is called an \emph{extension} of $N$ by $G/N$ (some authors use the 
opposite order of terminology); %\jnote{Moved footnote to main 
%text.}%\footnote{Some authors use the opposite nomenclature and call this an 
%%extension of $G/N$ by $N$.} 
given $N$ and $Q$, understanding the collection of groups $G$ which are extensions of $N$ by $Q$---that is, where $N \unlhd G$ and $G/N \cong Q$---is known as the \emph{extension problem}. The extension problem is considered quite difficult in general, but the theory of group cohomology exactly captures this problem and provides useful tools for its study, including connections with other cohomology theories such as in algebraic topology. One of the main technical achievements of the present paper is to make some aspects of group cohomology effective in the setting of worst-case complexity.

When $N$ is abelian the extension theory is conceptually easier and technically cleaner. Coincidentally, due to the polynomial-time algorithm for semisimple groups \cite{BCQ}, abelian normal subgroups are exactly the subject of interest at present. So for the rest of this subsection, \emph{we assume $N$ is abelian}; the theory for the general case is similar but more complicated, and is covered in \Sec{main_lemma_nonab}.

The extensions of $N$ by $Q$ are governed by two pieces of data: (1) an action of $Q$ on $N$ and (2) a cohomology class. We explain each of these in turn.

\paragraph{The action.} If $G$ is an extension of $N$ by $Q$, then $N \unlhd G$, so $G$ acts on $N$ by conjugation, giving a homomorphism $\theta' \colon G \to \Aut(N)$. As we have assumed $N$ is abelian, $N$ lies in the kernel of $\theta'$, so the conjugation action of $G$ on $N$ induces an action $\theta$ of $G/N \cong Q$ on $N$. Two such actions are \emph{compatible} if they become equal after applying some element of $\Aut(N) \times \Aut(Q)$, giving rise to the first problem:

\begin{definition}[\actcomp]
Given two actions $\varphi_i\colon Q \to \Aut(N)$ of a group $Q$ on a group $N$---specified by giving, for each $q \in Q$, $\varphi_i(q)$ as a permutation on the set $N$---decide whether the actions are compatible, that is, whether there is an element of $\Aut(N) \times \Aut(Q)$ whose application to $\varphi_1$ makes it equal to $\varphi_2$.
\end{definition}

\paragraph{The cohomology class.} Informally speaking, the simplest examples of 
extensions are when $Q$ can be ``lifted'' to a subgroup of $G$ that is compatible 
with the isomorphism $G/N \cong Q$ (the extension is said to be \emph{split}). 
However, it is possible to have an extension $G$ of $N$ by $Q$ in which this 
cannot happen. For example, consider the additive group of real numbers $\R$, and 
its normal subgroup $2\pi\Z$. (There are similar examples in finite groups, but we 
believe this example has more intuitive appeal. For readers familiar with group 
extensions, the goal here is to exhibit a nonsplit extension; $\Z_2 \unlhd \Z_4$ 
is a familiar example.) %\jnote{Turned footnote into parenthetical remark.} 
The quotient $\R/2\pi\Z$ is isomorphic to the ``circle group'' $S^1$ of unit 
complex numbers under multiplication, yet $S^1$ is not even a subgroup of $\R$, 
let alone ``liftable to $\R$.'' Contrast with the group $G = 2\pi\Z \times S^1$, 
which also has $2\pi\Z \unlhd G$ and $G/2\pi\Z \cong S^1$, yet $S^1$ is a subgroup 
of $G$. Note that as both $\R$ and $G$ are abelian the conjugation action of $\R$ 
or $G$ on any normal subgroup is trivial. So the actions cannot explain the fact 
that $S^1$ is not a subgroup of $\R$; instead, it is group cohomology that exactly 
captures this phenomenon.

Specifically, if $G$ is an extension of $N$ by $Q$, the failure of $Q$ to be 
``liftable'' to $G$ (a split extension) is measured by a \emph{cohomology class} as follows. Consider 
any function $s\colon Q \to G$ such that $s(q)$ is in the coset of $N$ 
corresponding to $q$ under the identification $G/N \cong Q$. $Q$ is ``liftable'' 
if and only if there is some such $s$ which is also a group homomorphism. The failure of any given $s$ to be a 
homomorphism is measured by the function 
\[
f_s(q, p) \defeq s(q)s(p)s(qp)^{-1}.
\]
Then $s$ is a homomorphism if and only if $f_s(q, p) = 1$ for all $p,q \in Q$. The 
cohomology class corresponding to $G$, viewed as an extension of $N$ by $Q$, is 
then the set $\{f_s | s\colon Q \to G \text{ as above}\}$. Two cohomology classes 
are \emph{isomorphic} if they become equal after applying some element of $\Aut(N) 
\times \Aut(Q)$, giving rise to the second problem:

\begin{definition}[\cohiso]
%\jnote{Re-worded to be consistent with terminology introduced so far}
Given two functions $f_i \colon Q \times Q \to N$ as above, decide whether there 
is an element $\alpha$ of $\Aut(N) \times \Aut(Q)$ whose application to $f_1$ 
makes $[f_1^\alpha] = [f_2]$.
%Given two cohomology classes $[f_i]$ ($i=1,2$)---each specified by a 2-cocycle $f_i\colon Q \times Q \to N$---decide whether there is an element of $\Aut(N) \times \Aut(Q)$ whose application to $f_1$ makes it cohomologous to $f_2$.
\end{definition}

\paragraph{Towards a formal strategy.} Let us now see how the action and cohomology class just introduced can be used in isomorphism testing. We refer to the pair $(\theta, f)$ of the corresponding action and (a representative of) a cohomology class as the \emph{extension data} of the extension. Suppose we are given two groups $G_1$ and $G_2$. We cleverly choose some $N_1\unlhd G_1$ and $N_2\unlhd G_2$, and (somehow we are lucky to find that) $N_1\cong N_2$ and $G_1/N_1\cong G_2/N_2$. Viewing $G_i$ as extensions of $N_i$ by $G_i/N_i$, we extract the action $\theta_i$  and cohomology classes $f_i$, for $i=1,2$. If there is a simultaneous solution (one single $(\alpha, \beta)\in\aut(N)\times\aut(G / N)$) to \actcomp for $\theta_1, \theta_2$ and \cohiso for $f_1, f_2$, we say the extension data are \emph{pseudo-congruent}.\footnote{\label{fn:pc} We take this terminology from Naik \cite{naik}, who gives a different definition of pseudo-congruence of extensions that is more standard from the group-theoretic point of view, but less well-adapted to the computational setting. We give the other definition and show that the two are formally equivalent in \S\ref{subsec:pc}. Robinson \cite{Rob} uses the term ``isomorphism'' for this notion; we prefer ``pseudo-congruence'' to avoid confusion with the several other notions of isomorphism floating around. Theoretical investigations of some aspects of this concept can be found in Robinson \cite[Sec.~4]{rob_paper}. }

\begin{definition}[\edpc]
Given two extension data $(\theta_i, f_i)$ for extensions of $N$ by $Q$---that is, $\theta_i \colon Q \to \Aut(N)$ and $f_i: Q \times Q \to N$---decide whether they are pseudo-congruent (see preceding paragraph).
\end{definition}

If the extension data are pseudo-congruent, then $G_1 \cong G_2$ and we are done. 
However, it is possible that $G_1 \cong G_2$ but the extension data are not 
pseudo-congruent (we thank Naik \cite{naik2} for providing 
Example~\ref{ex:isonotpc}). The difficulty is that $G$ may contain two normal 
subgroups $M, M' \unlhd G$ such that $M \cong M'$ and $G/M \cong G/M'$, but no 
automorphism of $G$ sends $M$ to $M'$. To resolve this problem, the Main 
Lemma~\ref{lem:main} shows that it is enough to take $N_1$ and $N_2$ to be the 
center or the radical, or more generally any characteristic subgroups that are 
preserved under isomorphisms. (We note that in Besche and Eick \cite{BE99}, they get around the pitfall by introducing the related concept of ``strong isomorphism,'' which is more natural for their purpose, namely the construction of finite groups.)

%\jnote{Moved to main text} %\footnote{In the practical setting, Besche and Eick 
%%\cite{BE99} got around the pitfall by introducing the related concept of 
%%``strong 
%%isomorphism,'' which is more natural for their purpose, namely the construction 
%%of 
%%finite groups.} 
Now we state the Main Lemma informally. Let us 
remark that, since in this section we mostly discuss the case of abelian normal 
subgroups, the Main Lemma is presented in the abelian case here, which is 
well-known (see e.g. \cite[Sec. 2.7.4]{HEO05}). We shall develop a general Main 
Lemma~\ref{lem:main_general} (including the case of \emph{non-abelian} normal subgroups) in 
Section~\ref{sec:main_lemma_nonab}. 

\begin{lemma*}[{Main Lemma~\ref{lem:main}, abelian case, informal}]
Given two groups $G_1$ and $G_2$, let $A_i$ be the abelian characteristic subgroup of $G_i$ of a given type (\eg, the center), $\theta_i$ the action of $G_i/A_i$ on $A_i$, and $f_i$ the cohomology class of the extension of $A_i$ by $G_i/A_i$ given by $G_i$. Suppose $A_1\cong A_2$ (identified as $A$) and $G_1/A_1\cong G_2/A_2$ (identified as $Q$). 

Then $G_1\cong G_2$ if and only if $\theta_1 \equiv \theta_2$, and $f_1 \equiv f_2$ up to the action of $\aut(A)\times\aut(Q)$.
\end{lemma*}
As evidence of the usefulness of the Main Lemma beyond this paper, we note that 
the polynomial-time algorithms for a special class of solvable groups in 
\cite{Gal09,QST11,BQ} follow this strategy: they use a theorem of 
Taunt~\cite{Tau55} to reduce isomorphism testing to a problem about linear 
representations of finite groups (see Problem 1 in \cite{QST11}), and solve that 
problem with additional tactics. In retrospect, Taunt's Theorem is a special case 
of the Main Lemma%
%\jnote{Moved footnote to main text and added the words ``and 
%extends''}%\footnote{Taunt's Theorem applies regardless whether the normal 
%%subgroup $N$ is abelian or not. The works \cite{Gal09,QST11,BQ} only used the 
%%case 
%%when $N$ is abelian. The general Main Lemma~\ref{lem:main_general} additionally 
%%covers the nonabelian case of Taunt's Theorem.}
, and Problem~1 in \cite{QST11} is essentially \actcomp. Taunt's Theorem applies regardless whether the normal subgroup $N$ is abelian or not, though the works \cite{Gal09,QST11,BQ} only used the case when $N$ is abelian. The general Main Lemma~\ref{lem:main_general} additionally covers and extends the nonabelian case of Taunt's Theorem.

Similarly, in retrospect the polynomial-time algorithm for semisimple groups \cite{BCGQ,BCQ} can be viewed as taking advantage of the nonabelian Main Lemma~\ref{lem:main_general}. We cover these examples in more detail in \Sec{previous}.

Due to the structure of the group classes considered in \cite{Gal09,QST11,BQ}, 
 \cohiso does not appear in these works. On the other hand, for $p$-groups of 
class $2$---currently believed the bottleneck---\cohiso is well-known to be 
necessary (see Fact~\ref{fact:nec_cohiso}). We thus turn to study the \cohiso 
problem in the following. As far as we know, this is the first time group 
cohomology has been used %\jnote{Re-worded to avoid Cayley table model} 
to get worst-case bounds for \GpI. % for \GpI in the Cayley table model. 

\subsubsection{The general Main Lemma}

As mentioned, the Main Lemma in the abelian case is well-known, and one 
contribution of this paper is to extend it to the case when the normal subgroup 
can be non-abelian. The extension theory in the non-abelian case is classical; a 
nice introduction can be found in Suzuki's book \cite[Section 2.7]{Suzuki}. In 
this paper, we shall adapt this theory explicitly to the setting of isomorphism 
testing. The general strategy is the same as the abelian case, but several 
technical 
details need to be taken care of. Consider the extension $\extension{N}{G}{Q}$ 
where $N$ need not be abelian. An obvious difference with the abelian case is that 
the conjugation 
action of $G$ on $N$ no longer induces an action of $Q$ on $N$, so one needs to 
consider the homomorphism $Q\to \Aut(N)/\Inn(N)$ instead. As another example, it 
is important 
to note is that the set of 2-cocycles is no longer a group, let alone an abelian 
group. Due to all these complications, a careful treatment is needed for the 
formulation and the proof for the general Main Lemma~\ref{lem:main_general}, and 
we refer the interested readers to Section~\ref{sec:main_lemma_nonab}.

\subsection{Overview of our algorithms} \label{sec:outline}
Here we give an overview of the structure of our algorithms, as well as some of the more salient details. We first consider the case when the solvable radical is abelian, to see how the strategy in the above section is applied. We then focus on central-radical groups to outline some key steps in the algorithms.

Given groups $G_1, G_2$, we first compute their solvable radicals $A_i=\rad(G_i)$ and the corresponding semisimple quotients $Q_i = G_i / \rad(G_i)$. Then apply the algorithm from \cite{Kav07} to $A_1$ and $A_2$, and the algorithm from \cite{BCQ} to $Q_1$ and $Q_2$. If either of them returns non-isomorphic, then $G_1\not\cong G_2$. If both algorithms return isomorphic, they also yield isomorphisms. Thus, without loss of generality, for $i=1, 2$, we use $A$ to denote $\rad(G_i)$ and $Q$ to denote $G_i/\rad(G_i)$, identifying $G_i$ as an extension of $A$ by $Q$.

Next, we compute the corresponding actions $\theta_1, \theta_2$ and 
representatives $f_1, f_2$ of the corresponding cohomology classes. As mentioned 
in \Sec{approach}, %\jnote{Removed ref to main lemma here, since it's only the 
%abelian case} %our Main Lemma~\ref{lem:main} says that 
$G_1$ and $G_2$ are isomorphic if and only if there is an element of $\Aut(A) \times \Aut(Q)$ which simultaneously turns $\theta_1$ into $\theta_2$, and $[f_1]$ into $[f_2]$ (as cohomology classes).

\paragraph{For groups with elementary abelian radicals (Theorem~\hyperref[cor:quotient_list_general]{B} = Corollaries~\ref{cor:quotient_list_general} and \ref{cor:quotient_list2_general})} %Corollary~\ref{cor:quotient_list_general}, stated above as the first half of ``Theorems~\ref{cor:quotient_list_general} and \ref{cor:quotient_list2_general}'').} 
Babai \etal \cite{BCGQ} showed that all automorphisms of a semisimple group can be enumerated in time $n^{O(\log\log n )}$. So if $n^{O(\log \log n)}$ time is allowed, we can use that algorithm to enumerate all $\beta\in\aut(Q)$. Then for each such $\beta$, search for some $\alpha\in\aut(A)$ such that $\theta_1=\theta_2^{(\alpha, \beta)}$, and $[f_1]=[f_2^{(\alpha, \beta)}]$. When $A\cong\Z_p^k$ is elementary abelian, this task can be reduced to \algprob{Module Cyclicity Testing} over finite-dimensional algebras, in almost the same way as the reduction from \algprob{Module Isomorphism} to \algprob{Module Cyclicity Testing} \cite{CIK97}. Here we only mention that the algebra is $U=\{\gamma\in M(k, p)\mid \forall q\in Q, \gamma\theta_2(q)=\theta_2(q)\gamma, \text{ and } \exists a\in \Z_p, [\gamma f_2]=[af_2]\}$ and we consider the $U$-module $V=\{\alpha\in M(k, p)\mid \forall q\in Q, \alpha\theta_1(q)=\theta_2(q)\alpha, \text{ and } \exists a\in\Z_p, [\alpha f_1]=[af_2]\}$. What is left is to verify that $V$ is a cyclic $U$-module if and only if there exists some desired $\alpha\in\GL(k, p)$. See \Sec{abelian} for the details.

\paragraph{For general central-radical groups (Theorem~\hyperref[cor:quotient_list]{A} = Corollary~\ref{cor:quotient_list}).} 
For groups with central radicals, $A=Z(G_i)$, so the actions $\theta_i$ are 
trivial, and we only need to solve \cohiso. As before, since $n^{O(\log \log n)}$ 
time is allowed, we can use the algorithm of Babai \etal \cite{BCGQ} to enumerate 
all $\beta\in\aut(Q)$. Then for each such $\beta$, we need to search for some 
$\alpha\in\aut(A)$ such that $[f_1]=[f_2^{(\alpha, \beta)}]$. We solve this 
problem using linear algebra over abelian groups, as follows. To ease the 
exposition let us assume $A=\Z_p^k$. Then we shall view any map $f:Q\times Q\to A$ 
as a $k\times |Q|^2$-size matrix over $\Z_p$, with $\aut(A)$ acting on the rows, 
$\aut(Q)$ inducing an action on the columns. The main difficulty at this point has 
to do with identifying which cohomology class $f$ is in,  in a way that is 
$\aut(A)$-invariant. Viewing $f$ as a $\Z_p$-linear vector (of dimension $k\times 
|Q|^2$), by Proposition~\ref{prop:cob} we can compute a projection $\pi$ in this 
vector space such that $\pi(f)$ identifies the cohomology class of $f$---that is, 
$\pi(f) = \pi(f')$ if and only if $f$ and $f'$ are in the same cohomology 
class---and such that $\pi$ commutes with every $\alpha \in \aut(A)$ (\ie, $\pi$ 
is $\aut(A)$-invariant). With fixed $\beta$, this allows us to compute $\pi(f_1)$ 
and $\pi(f_2^{(\id, \beta)})$, and then determine whether, as $k\times |Q|^2$-size 
matrices, their row spans are the same, which is a standard task in linear 
algebra. %\jnote{Made next sentence more accurate} 
For central-radical groups with elementary abelian radicals, this approach allows 
us to compute the coset of isomorphisms. We also give an alternative proof 
(\Sec{elem}) that allows us to decide isomorphism for general central-radical 
groups (where the radical need not be elementary abelian), but the alternative 
approach does not yield the coset of isomorphisms.
%Finally, to move from $A = \Z_p^k$ to general abelian $A$, we must consider the automorphism group of an arbitrary abelian group in some detail, which we do in \Sec{elem}.

\paragraph{For central-radical groups with $G/\rad(G)$ a direct product of nonabelian simple groups (Theorem~\hyperref[thm:ecentradwsoc]{C}=\ref{thm:ecentradwsoc}).} In this case $Q=\prod_{i\in[\ell]} T_i$, $T_i$ nonabelian simple. To ease the exposition let us assume $T_i$'s are all mutually isomorphic to some $T$, and $A=\Z_p^k$. For a function $f:Q\times Q\to A$, a key fact is that the cohomology class of $f$ is completely determined by the \emph{restrictions} of $f$ to the direct factors $T_i$ (Lemma~\ref{lem:prod}). Several group-theoretic facts lead to this cohomological proposition, including: (1) the direct product decomposition of $Q$ into nonabelian simple factors is unique (not just up to isomorphism); (2) if $U_i$ is the preimage of $T_i$ under the projection $G\to G/Z(G)$, then $u_i u_j = u_j u_i$ whenever $u_i\in U_i$, $u_j \in U_j$, and $i \neq j$ (\cite[Chapter 6, Proposition 6.5]{Suzuki2}, see Proposition~\ref{prop:central}). Another useful fact is the well-known description of $\aut(Q)$ as $\aut(T)\wr S_\ell$.

%\jnote{Updated next two paragraphs}
%When $A$ is small enough that we can enumerate $\aut(A)$ in polynomial time, this allows us, up to a polynomial overhead, to focus on the multiset of cohomology classes of the $U_i$'s. This idea leads to the algorithm for Theorem~\ref{thm:ecentradwsoc} (1).

%When $A$ is not small enough for the above tactic, 
Instead of considering $f\colon Q\times Q\to A$, we can thus consider $f_i\colon T_i\times T_i\to A$, $i\in[\ell]$; and instead of working with a $k\times |Q|^2$-size matrix, we can work with a $k\times (\sum_{i\in[\ell]}|T_i|^2)$-size matrix. This difference between $|Q|^2=\prod_{i\in[\ell]}|T_i|^2$ and $\sum_{i\in[\ell]}|T_i|^2$ leads to major savings. When $T$ is a nonabelian simple group, we use the powerful theorem of Guralnick, Kantor, Kassabov, and Lubotzky \cite{GKKL} (reproduced as Theorem~\ref{thm:GKKL} below); when $T$ is more generally only centerless and perfect, we restrict our attention to the setting where $|T|=O(1)$. In either case, we combine algorithms for \lincode and \cosetint. We need several technical ingredients (including Lemma~\ref{lem:direct_factor}) to make the above procedure work though. As in the previous setting, we give two proofs, one of which handles the general abelian case, and the other of which allows us to compute the coset of isomorphisms in the elementary abelian case.

\subsection{Organization of the paper}

%\jnote{Updated} 
In \Sec{prel_framework} we collect basic concepts from extension theory. 
%\jnote{Added sentence:} 
Appendix~\ref{app:gentle} contains a gentle introduction to extensions and 
cohomology, designed to be digestible without first preparing one's gut with half 
of a textbook. In \Sec{framework} we prove our (nonabelian) Main Lemma. We develop 
our strategy in \Sec{strategy}, which uses the Main Lemma to expand the ideas in 
\Sec{approach} into a formal framework. \Sec{prel_concrete} contains preliminaries 
and previous algorithmic results to prepare for the algorithms for central-radical 
groups. In \Sec{autq} we describe the $n^{O(\log \log n)}$-time algorithm for 
general central-radical groups (Theorem~\hyperref[cor:quotient_list]{A}=Corollary~\ref{cor:quotient_list}); this is also the 
algorithm for Corollary~\ref{cor:quotient_list2}. We also give the algorithms for 
groups with elementary abelian radicals that need not be central (Theorem~\hyperref[cor:quotient_list_general]{B}). In 
\Sec{centrad_bb}, we prove Theorem~\hyperref[thm:ecentradwsoc]{C}, giving the polynomial-time algorithms for central-radical 
groups with $G/\rad(G)$ a direct product of nonabelian simple groups (or 
centerless perfect groups of constant size). Finally, \Sec{future} contains future 
directions, some of which are motivated by the Cannon--Holt approach and the 
Babai--Beals filtration.

%\jnote{Added} 
For any of our results which currently depend on an abelian group being 
\emph{elementary} abelian, we discuss what obstacles towards this generalization 
we know how to overcome, and which remain to be overcome: see the end of 
Section~\ref{sec:coset}, Section~\ref{sec:towards_genab}, and 
Section~\ref{sec:future}. 

\section{Preliminaries on abelian cohomology} \label{sec:prel_framework}
The material in this section is standard group theory; for group theorists, this section serves primarily to fix notation. For computer scientists, we provide a gentle introduction to this material, with motivation and proofs, in Appendix~\ref{app:gentle}.

\paragraph{General notations.} For $n\in \N$, $[n]=\{1, \dots, n\}$. In this 
paper, all groups are finite. We use $\id$ to denote the identity element, or the 
group of order $1$. For a group $G$, $|G|$ denotes the order of $G$. We write 
$H\leq G$ if $H$ is a subgroup of $G$. The (right) coset of $H$ in $G$ containing 
$g\in G$ is $Hg=\{hg\mid h\in H\}$. Given two groups $G_1$ and $G_2$, $\ISO(G_1, 
G_2)$ denotes the set of $G_1\to G_2$ isomorphisms.  $\aut(G) = \ISO(G,G)$ is the 
group of automorphisms of $G$. The set $\ISO(G_1, G_2)$ is either empty or a coset 
of $\aut(G_1)$. For $g\in G$, conjugation by $g$ is the automorphism 
$\theta_g:G\to G$ defined by $\theta_g(x) \defeq gxg^{-1}$. For $g\in G$, the maps 
$\theta_g$ are the inner automorphisms of $G$, and they form a subgroup $\inn(G) 
\leq \aut(G)$. A subgroup $N\le G$ is normal if it is invariant under all inner 
automorphisms, and we write $N \unlhd G$. $N \leq G$ is a characteristic subgroup 
of $G$ if it is invariant under all automorphisms of $G$. $Z(G)$ denotes the 
center of $G$. For $K, L\leq  G$, $[K, L]$ denotes the subgroup generated by all 
elements of the form %\ynote{Changed the order} 
$[x,y]:=xyx^{-1}y^{-1}$, $x\in K$ 
and $y\in L$. $[G, G]$ is 
called the commutator subgroup of $G$.

\paragraph{Group extension data.} Given a finite group $G$ and an abelian normal subgroup $A \unlhd G$, when we consider $G$ as an extension of $A$ by $Q := G/A$, we denote this by $\extensionl{A}{\iota}{G}{\pi}{Q}$, 
%$A \stackrel{\iota}{\hookrightarrow} G \stackrel{\pi}{\twoheadrightarrow} Q$, 
where $\iota$ is an injective homomorphism and $\pi$ a surjective homomorphism, such that $\ker(\pi)=\im(\iota)$. In this paper, we mostly use the ``inner'' perspective, by identifying $A$ with its image $\iota(A) \unlhd G$. We sometimes refer to $G$ as the ``total group'' of the extension.

When $A \leq Z(G)$, the action of $G$ on $A$ by conjugation induces an action of $Q$ on $A$ by conjugation, which is the action associated to the extension $\extension{A}{G}{Q}$.

As $A$ is abelian, we write the group operation in $A$ additively, despite the fact that when considering general elements of $G$ we write the group operation in $G$ multiplicatively (this mixed notation is fairly standard in this setting). Even though $A$ is a subgroup of $G$, we tend to only use these notations in separate contexts and it should not cause confusion.

Let $\pi\colon G \to G/A \cong Q$ be the natural projection; then any function $s\colon Q \to G$ such that $\pi(s(q)) = q$ for all $q \in Q$ is called a \emph{section} of $\pi$. Any such section $s$ gives rise to a function $f_s\colon Q \times Q \to A$ defined by $f_s(p,q) \defeq s(p)s(q)s(pq)^{-1}$ (by applying $\pi$, it is readily verified that the image of $f_s$ is in fact contained in $A$). We are free to choose $s(1) = \id_G$, and then $f_s(1,q) = f_s(q,1) = 0$ for all $q \in Q$. Such $f$ are called \emph{normalized}. In the following all sections are normalized unless stated otherwise.

The fact that the group operation in $G$ is associative implies that for all $p,q,r \in Q$,
\[
f_s(p,q) + f_s(pq, r) = \theta_p(f_s(q,r)) + f_s(p,qr) \qquad \text{(the 2-cocycle identity)}.
\]
Any function $f\colon Q \times Q \to A$ is called a \emph{2-cochain}; any 2-cochain satisfying the 2-cocycle identity with respect to $\theta$ is a \emph{2-cocycle} (with respect to $\theta$). Given any homomorphism $\theta\colon Q \to \aut(A)$, every 2-cocycle with respect to $\theta$ arises as $f_s$ for some section $s$ of some extension $A \hookrightarrow G \twoheadrightarrow Q$ with action $\theta$.

Given a function $u \colon Q \to A$, the \emph{2-coboundary} associated to $u$ is the function $b_u\colon Q \times Q \to A$ defined by $b_u(p,q) \defeq u(p) + \theta_p(u(q)) - u(pq)$. Any two 2-cocycles associated to the same extension differ by a coboundary. 

%When do two 2-cocycles $f_s$, $f_{s'}$ correspond to the same extension? Suppose we know the two sections $s,s'\colon Q \to G$. As $s(q), s'(q)$ lie in the same coset of $A$, there is a function $u\colon Q \to A$ such that $s(q) = u(q)s'(q)$ for all $q \in Q$. Then $f_s(p,q) = f_{s'}(p,q) + \left(u(p) + \theta_p(u(q)) - u(pq)\right)$. A \emph{2-coboundary} is a function of the form $b_u(p,q) \defeq u(p) + \theta_p(u(q)) - u(pq)$ for any function $u\colon Q \to A$. Hence, if two 2-cocycles come from the same extension, they differ by a 2-coboundary. Eilenberg and Mac Lane \cite{em2} proved the converse, for a suitable notion of two extensions being ``the same,'' which we discuss in \S\ref{subsec:pc}. Two 2-cocycles that differ by a 2-coboundary are said to be \emph{cohomologous}.

The 2-cochains form an abelian group $C^2(Q, A)$ defined by pointwise addition: $(f+g)(p,q) \defeq f(p,q) + g(p,q)$. It is readily visible that the 2-cocycle identity is $\Z$-linear, and hence the 2-cocycles form a subgroup of the 2-cochains, denoted by $Z^2(Q, A, \theta)$. It is similarly verified that the 2-coboundaries form a subgroup of the 2-cocycles, denoted $B^2(Q, A, \theta)$.

A \emph{2-cohomology class} is a coset of $B^2(Q, A, \theta)$ in $Z^2(Q, A, \theta)$, and any element of this coset is a representative of the cohomology class. If $f \in Z^2(Q, A, \theta)$, we denote the corresponding cohomology class by $[f]$. The group of 2-cohomology classes is denoted $H^2(Q, A, \theta) \defeq Z^2(Q, A, \theta) / B^2(Q, A, \theta)$. By the above discussion, each extension $A \hookrightarrow G \twoheadrightarrow Q$ determines a single cohomology class $[f] \in H^2(Q, A, \theta)$.

We thus arrive at one of the central notions in this paper:

\begin{definition}\label{def:ext_data}
For $A$ an abelian group and $Q$ any group, a pair $(\theta, f)$ of an action 
$\theta\colon Q \to \aut(A)$ and a 2-cocycle $f \colon Q \times Q \to A$, $f \in 
Z^2(Q, A, \theta)$ is \emph{extension data}. Two extension data for the pair $(Q, 
A)$ are \emph{equivalent} if they have the exact same action and if the two 
2-cocycles are cohomologous (differ by a coboundary).
\end{definition}

Given an extension $A \hookrightarrow G \twoheadrightarrow Q$, the extension data associated to this extension are the action $\theta$ as defined above, and any 
2-cocycle $f_s$ for any section $s\colon Q \to G$.
Note that extension data are non-unique, as we may choose any representative of 
the corresponding 2-cohomology class. Furthermore, if the action is trivial then 
this 
extension is called \emph{central}. If the 2-cohomology class is trivial then this 
extension is called \emph{split}; in this case $G$ is a semi-direct product of $A$ 
by $P$ for some subgroup $P$.

%Two important special cases of extension data $(\theta, f)$ are as follows.
%\begin{description}
%\item[$f$ is trivial (as 2-cohomology class).] This implies that there exists 
%$P\leq G$ such that $AP=G$ and $P\cap A=\id$, \ie, that $G$ is the semidirect 
%product $A \rtimes P$. Such $P$ is called a \emph{complement} of $A$ in $G$, and 
%the extension is called a \emph{split} extension. In this case, an isomorphism 
%test need only focus on one of the two aspects of \GpI: \edpc simplifies to 
%\actcomp.
%\item[$\theta$ is trivial.] This implies that $A \leq Z(G)$, and the extension is 
%called central. In this case, an isomorphism test need only focus on the other 
%aspect of \GpI: \edpc simplifies to \cohiso.
%\end{description}

%\begin{remark}
%It is not difficult to test whether an input satisfies one of the above 
%conditions: it is trivial to test whether an extension is central, and for 
%completeness we include an algorithm to test whether an extension is split in 
%Appendix~\ref{app:split}.
%\end{remark}

\subsection{Pseudo-congruent extensions versus isomorphic total 
groups}\label{subsec:main}
Recall that a characteristic subgroup is a subgroup invariant under all automorphisms. The analogous notion for isomorphisms (rather than automorphisms) is a function $\mathcal{S}$ that assigns to each group $G$ a subgroup $\mathcal{S}(G) \leq G$ such that any isomorphism $\varphi\colon G_1 \to G_2$ restricts to an isomorphism $\varphi|_{\mathcal{S}(G_1)}\colon \mathcal{S}(G_1) \to \mathcal{S}(G_2)$. We call such a function a \emph{characteristic subgroup function}.
%In line with other works in group theory, we call such a function a \emph{characteristic subgroup function}. 
Note that if $G_1=G_2$, this says that $\mathcal{S}(G_1)$ is sent to itself by every automorphism of $G_1$, that is, $\mathcal{S}(G_1)$ is a characteristic subgroup of $G_1$. Most natural characteristic subgroups encountered are characteristic subgroup functions, for example the center $Z(G)$, the commutator subgroup $[G, G]$, or the radical $\rad(G)$.

\begin{definition}\label{def:pc}
Let $A$ be an abelian group and $Q$ any group, and let $(\theta_1, f_1)$ and 
$(\theta_2, f_2)$ be two extension data for $A$-by-$Q$. Then the extension data 
are \emph{pseudo-congruent}%
%\ynote{Removed: the footnote ``See Footnote~\ref{fn:pc} 
%on page~\pageref{fn:pc}.''}
\footnote{See Footnote~\ref{fn:pc} on page~\pageref{fn:pc}.} 
if there exists $(\alpha, \beta)\in \aut(A)\times\aut(Q)$, such that 
\begin{equation}\label{eqn:action_pc}
\theta_1(q)(a) = \alpha^{-1}(\theta_{2} (\beta(q))(\alpha(a))) =: \theta_2^{(\alpha, \beta)}(q)(a),
\end{equation}
for all $q \in Q, a \in A$, and
\begin{equation}\label{eqn:cohom_pc}
f_1(p, q) = \alpha^{-1}(f_2(\beta(p), \beta(q))) + b_u(p,q)
\end{equation}
for all $p,q \in Q$, and for some 2-coboundary $b_u$. In this case we write $(\theta_1, f_1)\pc (\theta_2, f_2)$.
\end{definition}

\begin{lemma}[{See, \eg, \cite[Sec.~2.7.4]{HEO05}}]\label{lem:main}
Let $\charfn$ be a characteristic subgroup function. Given two finite groups $G_1$ and $G_2$, suppose $\charfn(G_1)$ and $\charfn(G_2)$ are abelian. Then $G_1\cong G_2$ if and only if both of the following conditions hold:
\begin{enumerate}
\item $\charfn(G_1) \cong \charfn(G_2)$ (which we denote by $A$) and $G_1 / \charfn(G_1) \cong G_2 / \charfn(G_2)$ (which we denote by $Q$); 
\item $(\theta_1, f_1)\pc (\theta_2, f_2)$, where $(\theta_i, f_i)$ is the extension data of the extensions $A \hookrightarrow G_i \twoheadrightarrow Q$. 
\end{enumerate}
\end{lemma}

For a detailed proof that doesn't require reading half a textbook on group theory first, see Appendix~\ref{app:gentle}. In the next section we generalize this to the case where the normal subgroup need not be abelian.

\section{Nonabelian cohomology and its applications}\label{sec:framework} \label{sec:main_lemma_nonab}

%\subsection{The Main Lemma for nonabelian normal subgroups} 
Here we consider extensions $\extension{N}{G}{Q}$ where $N$ need not be abelian, \ie, the general case. We show that Lemma~\ref{lem:main} extends to the case when $N$ comes from a characteristic subgroup function---not necessarily abelian---showing the usefulness of the extension theory perspective in its full generality. The results of this section may be of independent interest and of further use in the future, but in this paper will only be needed for the applications in this section and the next (namely, to show how the results of \cite{BCQ} can be viewed in the same cohomological light as Lemma~\ref{lem:main}, and to reduce the case where $G$ acts by inner automorphisms on its radical to natural sub-problems). Suzuki's book \cite[Section~2.7]{Suzuki} contains a nice introduction to the extension theory in the nonabelian case, while our contribution here is to adapt this theory explicitly to the setting of isomorphism testing. While some of the basics here are already laid out in Suzuki \cite{Suzuki}---so the first part of this section can be viewed as a review to fix notation---here we consider how the extension theory with nonabelian kernels can be used to understand isomorphism of the total groups, which was only considered in a very special case there (reproduced as Theorem~\ref{thm:centerless} below).

\paragraph{The action.} The first difference to notice when $N$ is non-abelian is that the conjugation map $\theta'\colon G \to \Aut(N)$, defined by $\theta'(g) = c_g$ where $c_g(n) = g n g^{-1}$, no longer contains $N$ in its kernel, and hence no longer descends to a map $Q \to \Aut(N)$. However, the action of $N$ on itself by conjugation is by inner automorphisms (by definition) so that we do get a well-defined map $G/N \to \Aut(N) / \Inn(N)$, that is, $Q \to \Out(N)$. For ease of reference, we give a name to such maps:

\begin{definition}
An \emph{outer action} of a group $Q$ on a group $N$ is a group homomorphism $Q \to \Out(N) = \Aut(N) / \Inn(N)$.
\end{definition}

In computations, rather than represent an outer automorphism as a coset of $\Inn(N)$ in $\Aut(N)$, we simply give it by a representative automorphism, and must remember when we may need to multiply by an arbitrary element of $\Inn(N)$. Throughout this section we use $T$ to denote an action rather than $\theta$, to remind the reader that the essential object here is the outer action represented by $T$, despite the fact that we work directly with actions $T\colon Q \to \Aut(N)$.

Correspondingly, in the setting of general $N$, the problem \actcomp must be generalized to \outactcomp, which is defined as follows. Two actions $T_1, T_2 \colon Q \to \Aut(N)$ are said to be ``outer equivalent'' if there is a function $t'\colon Q \to \Inn(N)$ and an automorphism $\alpha \in \Aut(N)$ such that $T_1(q) = \alpha^{-1} \circ t'(q) \circ T_2(q) \circ \alpha$ for all $q \in Q$. 

\begin{definition}[\outactcomp] 
Given two actions $T_1, T_2\colon Q \to \Aut(N)$, decide whether there an element $\beta \in \Aut(Q)$ such that $T_1$ and $T_2 \circ \beta$ are outer equivalent, that is, whether there exists $(\beta, \alpha, t') \in \Aut(Q) \times (\Aut(N) \ltimes \Inn(N)^{Q})$ such that $T_1(q) = \alpha^{-1} \circ t'(\beta(q)) \circ T_2(\beta(q)) \circ \alpha$ for all $q \in Q$.
\end{definition}

Although this formulation of \outactcomp is more complicated than if we had represented an outer automorphism as a full coset $\theta \Inn(N)$, it will be useful when we formulate \edpc below.

Note that when $N$ is abelian there are no inner automorphisms, so $\Out(N) = \Aut(N)$, the only choice for $t'$ above is trivial, and \outactcomp then becomes \actcomp.

\begin{remark}
We note that, unlike the case of $N$ abelian, when $N$ is nonabelian it is possible that some outer actions $\theta\colon Q \to \Out(N)$ may not be induced by \emph{any} extension of $N$ by $Q$. When there \emph{is} such an extension, the outer action $\theta$ is called \emph{extendable}. Eilenberg and Mac Lane \cite[Sec.~7--9]{em2} characterize which outer actions are extendable in terms of the third cohomology group $H^3(Q, Z(N))$. As our interest is primarily in \GpI, whenever it matters (\eg, in the definition of \outactcomp) we happily restrict our attention to extendable outer actions. 

The characterization in terms of cohomology with coefficients in $Z(N)$ allows one to use linear algebra over the abelian group $Z(N)$ to test in polynomial time whether a given outer action is extendable. In particular, the outer action is extendable if and only if an associated third cohomology class vanishes, \ie, is a 3-coboundary. A basis for $B^3(Q, Z(N))$ can be computed analogously to that in Proposition~\ref{prop:basis}, by enumerating over a natural basis of the 2-cochains $C^2(Q, Z(N))$ (which has dimension $|Q|^2$) and applying the coboundary operator. Then we just have to compute the 3-cocycle corresponding to the given outer action, and check that it lies in the linear subspace spanned by (subgroup generated by) the 3-coboundaries. Note that we may treat 3-cochains as $k \times |Q|^3$ matrices, analogous to how we treat 2-cochains.
\end{remark}

\paragraph{The cohomology class.} As in the case of $N$ abelian, when $N$ is nonabelian, one may still choose a set-theoretic section $s\colon Q \to G$ and get a 2-cocycle $f_s \colon Q \times Q \to N$. This section $s$ gives an action (not just outer action) $T_s\colon Q \to \Aut(N)$ by conjugation, namely $T_s(q)(n) = s(q) n s(q)^{-1}$. Starting from associativity in $G$, one may then work out, as in the abelian case, the 2-cocycle condition: 
\[
f_s(q_1, q_2) f_s(q_1 q_2, q_3) = T_s(q_1)(f_s(q_2, q_3)) f_s(q_1, q_2 q_3) \qquad \text{(the nonabelian 2-cocycle identity)}.
\] 
However, this condition depends not just on the action $T_s$ and the 2-cochain $f_s$, but also on some relationship between $T_s$ and $f_s$ (in this case, that they come from the \emph{same} section $s$). We would much prefer a condition that does not depend on the ambient extension group $G$. To get this condition, note that the action satisfies $T_s(q_1) T_s (q_2) = c_{f_s(q_1, q_2)} T_s(q_1 q_2)$, where $c_n \colon N \to N$ denotes the inner automorphism given by conjugation by $n \in N$: $c_n(m) = n m n^{-1}$. This leads us to the definition of extension data for general $N$:

\begin{definition}[Extension data for general $Q, N$]
Let $Q$ and $N$ be groups. We say that a pair $(T, f)$ of an action $T\colon Q \to \Aut(N)$ and a function $f\colon Q \times Q \to N$ is \emph{extension data} if, for all $q_i \in Q$, 
\begin{equation} \label{eqn:general_action}
T(q_1) T(q_2) = c_{f(q_1, q_2)} T(q_1 q_2)
\end{equation}
and
\begin{equation} \label{eqn:general_cocycle}
f(q_1, q_2) f(q_1 q_2, q_3) = T(q_1)\left( f(q_2, q_3) \right) f(q_1, q_2 q_3).
\end{equation}
In this case, we refer to $f$ as a 2-cocycle with respect to the action $T$.
\end{definition}

Note that condition (\ref{eqn:general_action}) very nearly determines $f$: it 
determines $f(q_1, q_2)$ up to an element of $Z(N)$. In particular, when $Z(N) = 
1$ condition (\ref{eqn:general_action}) actually does determine $f$ completely, a 
fact we will take advantage of when discussing the polynomial-time algorithm for 
groups with no abelian normal subgroups \cite{BCQ} (see also 
Theorem~\ref{thm:centerless}).

Another difference in the case of $N$ nonabelian is that, although we might denote the set of 2-cocycles by $Z^2(Q, N, T)$, this set will not in general be a group in any natural way, let alone an abelian group. (Also note that it depends on the action $T$, whereas we know that the action is not intrinsic to the extension but only the corresponding outer action is.) However, the quotient of two 2-cocycles with respect to the same action $T$ will land in the center $Z(N)$, allowing us to reduce part of the question back to the case of $N$ abelian. To see this, let $f_1, f_2$ be two 2-cocycles with respect to an action $T\colon Q \to \Aut(N)$, and consider conjugating by their quotient $f_1(q_1, q_2) f_2(q_1, q_2)^{-1}$:
\[
c_{f_1(q_1, q_2) f_2(q_1, q_2)^{-1}} = T(q_1) T(q_2) T(q_1 q_2)^{-1} T(q_1 q_2) T(q_2)^{-1} T(q_1) = \id_{N}.
\]
As the center $Z(N)$ consists exactly of those $n \in N$ such that $c_n = \id_{N}$, we find that the quotient $f_1(q_1, q_2) f_2(q_1, q_2)^{-1}$ lies in $Z(N)$. So although there isn't really a cohomology group ``$H^2(Q, N, T)$,'' we can often nonetheless reduce to questions about cohomology classes in $H^2(Q, Z(N), T|_{Z(N)})$. Everything up to this point in this section has been classical (although it has not been leveraged much for algorithms).

\paragraph{Equivalence.} As in the case of $N$ abelian, two extensions $\extension{N}{G_i}{Q}$ are said to be \emph{equivalent} if there is an isomorphism $\gamma\colon G_1 \to G_2$ such that $\gamma$ induces the identity map on both $N$ and $Q$. However, since $Z^2(Q, N, T)$ is no longer a group and $B^2(Q, N, T)$ no longer its subgroup, the notion of equivalent extensions doesn't translate so easily to a notion of equivalence for extension data. Hence we derive this condition more or less from scratch. That is, we derive what it means for two extension data to be equivalent by analyzing how two extension data coming from the same extension may differ, when a different choice of section $s \colon Q \to G$ is made. 

Fix an extension $\extension{N}{G}{Q}$ and two sections $s_1,s_2 \colon Q \to G$. Let $t(q) \defeq s_1(q) s_2(q)^{-1}$; as the $s_i$ are both sections, $s_1(q)$ and $s_2(q)$ are in the same coset of $N$, so that $t(q) \in N$ for all $q \in Q$. Then the actions $T_1 = \theta_{s_1}$ and $T_2 = \theta_{s_2}$ differ by the inner automorphism $c_{t(q)}$: $T_1(q) = c_{t(q)} \circ T_2(q)$. Recall that we set $f_i(q, r) \defeq s_i(q) s_i(r) s_i(q r)^{-1}$ for $q,r \in Q$. Then we have that
\begin{eqnarray*}
f_1(q, r) & = & s_1(q) s_1(r) s_1(q r)^{-1} \\
 & = & t(q) s_2(q) t(r) s_2(r) s_2(q r)^{-1} t(q r)^{-1} \\
 & = & t(q) T_2(q)(t(r)) s_2(q) s_2(r) s_2(q r)^{-1} t(q r)^{-1} \\
 & = & t(q) T_2(q)(t(r)) f_2(q, r) t(q r)^{-1} \\
 & = & t(q) T_2(q)(t(r)) c_{f_2(q, r)}(t(q r)^{-1}) f_2(q, r) =: f_2^{t,T_2}(q, r)
\end{eqnarray*}
For future reference, we denote this final expression by $f_2^{t,T_2}(q, r)$.

\begin{definition} \label{def:equiv_general}
Two extension data $(T_i, f_i)$ are \emph{equivalent} if there is a map $t\colon Q \to N$ such that $T_1(q) = c_{t(q)} \circ T_2(q)$ for all $q \in Q$, and $f_1 = f_2^{t,T_2}$. 
\end{definition}

There are several aspects of this definition to take note of:
\begin{itemize}
\item By definition, two extension data can be equivalent only if $T_1(q)$ and $T_2(q)$ represent the same outer automorphism of $N$, in accord with our discussion above. 

\item When $N$ is abelian, this definition agrees with the previous definition of equivalent extensions. For, in this case, $c_{t(q)} = \id_N$ and the condition $f_1 = f_2^{t,T_2}$ exactly says that $f_1$ and $f_2$ differ by the 2-coboundary $b_t$ defined by $t$.

\item $T_1 = T_2$ if and only if $s_1(q)$ and $s_2(q)$ differ by an element of the center $Z(N)$, that is, $t$ is a map $Q \to Z(N)$. In this case, let $T = T_1 = T_2$; then $(T, f_1)$ and $(T, f_2)$ are equivalent if and only if $f_1$ and $f_2$ differ by the coboundary $b_t \in B^2(Q, Z(N), T)$. Again, this will be relevant for our discussion below of the polynomial-time algorithm for semisimple groups \cite{BCQ}.
\end{itemize}

\paragraph{Pseudo-congruence.} As before, pseudo-congruence is defined as ``equivalence up to the action of $\Aut(N) \times \Aut(Q)$:'' 

\begin{definition}[Pseudo-congruence of extension data for general $Q,N$]
Let $Q$ and $N$ be two groups. Two extension data $(T_i, f_i) \in (Q \to \Aut(N), Q \times Q \to N)$ are \emph{pseudo-congruent} if there exist $(\alpha,\beta) \in \Aut(N) \times \Aut(Q)$ such that $(T_1, f_1)$ and $(T_2^{(\alpha,\beta)}, f_2^{(\alpha,\beta)})$ are equivalent. 

In more detail, the extension data are pseudo-congruent if there exists $(\alpha,\beta) \in \Aut(N) \times \Aut(Q)$ and $t\colon Q \to N$ such that, for all $q \in Q$ and all $n \in N$:
\begin{equation} \label{eqn:general_action_pc}
T_1(q)(n) = (\alpha^{-1} \circ c_{t(\beta(q))} \circ T_2(\beta(q)) \circ \alpha)(n)
\end{equation}
and 
\begin{equation} \label{eqn:general_cohom_pc}
f_1(q_1, q_2) = \alpha^{-1}\left[ f_2^{t,T_2}(\beta(q_1), \beta(q_2))\right] =: f_2^{(\alpha,\beta,t,T_2)}(q_1, q_2).
\end{equation}

The problem \edpc is to decide, given two extension data $(T_i, f_i)$, whether they are pseudo-congruent.
\end{definition}

\begin{lemma}[{Main Lemma\footnote{\label{fn:mainlem}Although the statement of the Main Lemma may not surprise experts, and follows from standard constructions in group cohomology, we have been unable to find a reference for it, and it seems not widely-known even amongst mathematicians, despite the abelian case being very well-known. As an example of our Main Lemma not being well-known, we point out that the main theorem of a 2003 paper \cite{hammerli} in \emph{L'Enseignment Math\'ematique}, whose proof takes approximately 7 pages there even assuming knowledge of group cohomology, is a short corollary of our Main Lemma, as shown in Remark~\ref{rmk:mainlemLie}. In that paper, it is even asked whether there are larger classes of groups for which its main theorem holds \cite[Remark~4.2]{hammerli}; our Main Lemma gives quite a general answer to this question.}}] \label{lem:main_general}
Let $\charfn$ be a characteristic subgroup function. Given two finite (or Lie, see Remark~\ref{rmk:mainlemLie}) groups $G_1$ and $G_2$, $G_1\cong G_2$ if and only if both of the following conditions hold:
\begin{enumerate}
\item $\charfn(G_1) \cong \charfn(G_2)$ (which we denote by $N$) and $G_1 / \charfn(G_1) \cong G_2 / \charfn(G_2)$ (which we denote by $Q$); 
\item $(T_1, f_1)\pc (T_2, f_2)$, where $(T_i, f_i)$ is the extension data of the extensions $N \hookrightarrow G_i \twoheadrightarrow Q$. 
\end{enumerate}
\end{lemma}

\begin{proof}
First suppose that $\gamma \colon G_1 \to G_2$ is an isomorphism. Since $\mathcal{S}$ is a characteristic subgroup function, $\gamma$ restricts to an isomorphism between the copy $\mathcal{S}(G_1)$ of $N$ in $G_1$ and the copy $\mathcal{S}(G_2)$ of $N$ in $G_2$, \ie, an automorphism $\alpha \in \Aut(N)$. Consequently, $\gamma$ induces an automorphism $\beta := \overline{\gamma} \in \Aut(Q)$. After twisting by these automorphisms, the discussion preceding Definition~\ref{def:equiv_general} shows that the extension data become equivalent.

Conversely, suppose $(T_1, f_1) \cong (T_2, f_2)$, via $(\alpha, \beta, t) \in \Aut(N) \times \Aut(Q) \times (Q \to N)$. As in the abelian case we have a standard reconstruction procedure; we construct groups $H_i$ from $(T_i, f_i)$ such that $H_i \cong G_i$, and then we show how the pseudo-congruence of the extension data easily yields an isomorphism $H_1 \cong H_2$.

The underlying set of $H_i$ will be $N \times Q$, with multiplication defined by
\[
(n, p) \circ_{H_i} (m, q) = (n \cdot T_i(p)(m) \cdot f_i(p, q), p q).
\]
Note that this is the same as in the abelian case, just being careful about the order. Here we have started using dots ``$\cdot$'' to denote multiplication, as the expressions below get somewhat complicated and this helps to keep things clear. Let $s_i \colon Q \to G_i$ denote the sections used to construct the extension data $(T_i, f_i)$. Then it is readily verified that the map $(n,q) \mapsto n s_i(q)$ gives an isomorphism $H_i \stackrel{\cong}{\to} G_i$.

Finally, we claim that the map $\varphi(n,p) \defeq (\alpha(n) \cdot t(\beta(p)), \beta(p))$ is an isomorphism from $H_1$ to $H_2$. The main fact to check is that this is even a homomorphism. Consider $(n, p)$ and $(m,q) \in H_1$. On the one hand, we have 
\begin{eqnarray*}
\varphi((n,p) \circ_{H_1} (m,q)) & = & \varphi( n \cdot T_1(p)(m) \cdot f_1(p, q), pq) \\
 & = & (\alpha (n \cdot T_1(p)(m) \cdot f_1(p, q)) \cdot t(\beta(pq)), \beta(pq)) \\
 & = & (\alpha(n) \cdot \alpha(T_1(p)(m)) \cdot \alpha(f_1(p, q)) \cdot t(\beta(pq)), \beta(pq))
\end{eqnarray*}
On the other hand, we have (here we'll sometimes use square brackets $[]$ to denote application of an automorphism to help keep all the parentheses straight):
\begin{eqnarray*}
\varphi(n,p) \circ_{H_2} \varphi(m,q) & = & (\alpha(n) \cdot t(\beta(p)), \beta(p)) \circ_{H_2} (\alpha(m) \cdot t(\beta(q)), \beta(q)) \\
 & = & (\alpha(n) \cdot t(\beta(p)) \cdot T_2(\beta(p))\left[\alpha(m) \cdot t(\beta(q))\right] \cdot f_2(\beta(p), \beta(q)), \beta(p) \beta(q)) \\
 & = & (\alpha(n) \cdot t(\beta(p)) \cdot T_2(\beta(p))[\alpha(m)] \cdot T_2(\beta(p))[t(\beta(q))] \cdot f_2(\beta(p), \beta(q)), \beta(pq)) \\ 
 & = & (\alpha(n) \cdot (c_{t(\beta(p))} \circ T_2(\beta(p)))[\alpha(m)] \cdot 
 t(\beta(p)) \cdot T_2(\beta(p))[t(\beta(q))] \cdot f_2(\beta(p), \beta(q)), \\ 
 & & \quad \beta(pq))
\end{eqnarray*}
Let's work through these two expressions bit by bit. We can dispense easily with the second coordinate, as $\beta(pq) = \beta(p) \beta(q)$ since $\beta \in \Aut(Q)$. Both of the first coordinates begin with $\alpha(n)$. Next we have $\alpha(T_1(p)(m))$ on the one hand and $(c_{t(\beta(p))} \circ T_2(\beta(p)))[\alpha(m)]$ on the other. From the definition of pseudo-congruence, we have that $T_1(p)(m) = \alpha^{-1} ( c_{t(\beta(p)} \circ T_2(\beta(p)) )[\alpha(m)]$. Applying $\alpha$ to both sides of this equation we see that these two terms are equal.

The remainder of the first coordinate is then $\alpha(f_1(p, q)) \cdot t(\beta(pq))$ in the first case. From the definition of pseudo-congruence we have:
\begin{eqnarray*}
\alpha(f_1(p, q)) t(\beta(pq)) & = & f_2^{t,T_2}(\beta(p), \beta(q)) \cdot t(\beta(pq)) \\
 & = & t(\beta(p)) \cdot T_2(\beta(p))[t(\beta(q))] \cdot f_2(\beta(p), \beta(q)) \cdot t(\beta(pq))^{-1} \cdot t(\beta(pq)) \\
 & = & t(\beta(p)) \cdot T_2(\beta(p))[t(\beta(q))] \cdot f_2(\beta(p), \beta(q)),
\end{eqnarray*}
which is exactly the remainder of the first coordinate in the second case, as desired. Hence $\varphi$ is a homomorphism.

Finally, it suffices to show that $\varphi$ is injective, for as $|H_1| = |N| |Q| = |H_2|$, it will then follow that $\varphi$ is bijective and hence an isomorphism. Consider the kernel of $\varphi$: $\varphi(n,p) = (1, 1)$. As the second coordinate is $1$, we have $\beta(p) = 1$ and hence $p = 1$. As the first coordinate is $1$, we have $\alpha(n) t(\beta(p)) = \alpha(n) = 1$, so we also have $n=1$. (In the first equality we use the fact that $t(1) = 1$, which follows from $T_i(1_Q) = \id_N$.) Hence $\varphi$ is injective, and thus an isomorphism.
\end{proof}

\begin{remark}[The Main Lemma for Lie groups] \label{rmk:mainlemLie}
Lemma~\ref{lem:main_general} also extends to the case of Lie groups, allowing us to show that the main theorem of H\"{a}mmerli \cite{hammerli} follows from our Main Lemma as a quick corollary, as well as giving a very general answer to a question he posed. The main theorem of H\"{a}mmerli \cite{hammerli} is essentially the special case of our Main Lemma in which the characteristic subgroup is taken to be the connected component of the identity. H\"{a}mmerli asked \cite[Remark~4.2]{hammerli} whether his main theorem extended to other classes of groups; our Main Lemma extends it greatly, and shows that the result has very little to do with Lie groups \emph{per se}.

To see that our Main Lemma extends to Lie groups, note that the only place we used finiteness in the proof of the Main Lemma~\ref{lem:main_general} is in the final paragraph, to get surjectivity from injectivity. In the case of Lie groups, we also note that the map in the above proof is continuously differentiable (even smooth). The rest of the argument is essentially dimension-counting, but we give it here for completeness. First, because the homomorphism is differentiable, it descends to a map of Lie algebras, and injectivity of the map implies injectivity of the corresponding map of Lie algebras. As the Lie algebras are, in particular, finite-dimensional vector spaces of the same dimension, injectivity and linearity imply surjectivity, so we have an isomorphism of Lie algebras. For a Lie group $G$, let $G^{(0)}$ denote the connected component of the identity of $G$. Continuity of the homomorphism and the fact that it induces an isomorphism of Lie algebras implies that it induces an isomorphism of $H_1^{(0)} \cong H_2^{(0)}$. As $G / G^{(0)}$ is a finite group for any Lie group $G$, injectivity and continuity together imply that we have an isomorphism of the component groups $H_1 / H_1^{(0)} \cong H_2 / H_2^{(0)}$. As the map we started with was an injective homomorphism that induces an isomorphism of the identity components and of the component groups, it is an isomorphism. \rmkqed
\end{remark}

\subsection{Application to extensions with trivial outer action}
%\paragraph{Extensions with trivial outer action.} 
We did not define \cohiso for general $N$ and then proceed to pseudo-congruence, 
as in the abelian case, because it turns out that when the outer action is 
trivial, \cohiso for action-trivial extensions of $N$ by $Q$ reduces to \cohiso 
for extensions of $Z(N)$ by $Q$. To prove this we use one additional concept, that 
of a central product. Although this notion generalizes to an arbitrary number of factors, we only need the two-factor case:

\begin{definition}[{Central decomposition; see, \eg, \cite{Wil09a}}]
A pair $\{G_1, G_2\}$ of subgroups of a finite group $G$ is a \emph{central decomposition} of $G$ if $G$ is generated by $G_1$ and $G_2$ ($G = \langle G_1, G_2 \rangle$) and $G_1$ and $G_2$ commute ($[G_1, G_2] = 1$).
\end{definition}

\begin{lemma} \label{lem:acttriv}
Let $\extensionl{N}{}{G}{\pi}{Q}$ be an extension of $N$ by $Q$ which induces the trivial outer action $\theta(q) = \id_{N}\Inn(N)$ for all $q \in Q$. Then there is a subgroup $H$ of $G$ such that $H \cap N = Z(N)$, $\pi(H) = Q$, and $\{N, H\}$ is a central decomposition of $G$. We denote this subgroup $H$ by $G|_{Z(N)}$.
\end{lemma}

\begin{proof}
There is a section $s\colon Q \to G$ such that $c_{s(q)} = \id_{N}$ for all $q \in Q$. Let $f(p,q) = f_s(p,q) = s(p) s(q) s(pq)^{-1}$ be the 2-cocycle corresponding to $s$. As $c_{s(q)} = \id_{N}$ for all $q \in Q$, we also have that $c_{f(p,q)} = \id_{N}$ for all $p,q \in Q$. As $f(p,q) \in N$, this implies that $f(p,q) \in Z(N)$. Hence $f$ is a 2-cocycle in $H^2(Q, Z(N))$ (for the trivial action of $Q$ on $Z(N)$). Let $H$ be the subgroup generated by $Z(N)$ and $s(Q)$. Since $f(p,q) \in Z(N)$ for all $p,q \in Q$, it follows that every element of $H$ can be represented uniquely in the form $zs(q)$ for $z \in Z(N), q \in Q$. 

From the uniqueness of the representation $z s(q)$, it follows immediately that $H \cap N = Z(N)$. Since $H$ included $s(q)$ for all $q \in Q$, it follows that $\pi(H) = Q$. Finally, to see that $[H,N] = 1$, consider $z s(q) n (z s(q))^{-1} n^{-1} = z s(q) n s(q)^{-1} z^{-1} n^{-1}$. Since $c_{s(q)} = \id_N$, this equals $z n z^{-1} n^{-1}$, but since $z \in Z(N)$, the latter is trivial.
\end{proof}

The preceding lemma nearly allows us to reduce group isomorphism when the outer action of $G$ on $\rad(G)$ is trivial to isomorphism of a pair of central radical groups and a pair of solvable groups. However, up to this point we have brushed over the fact that central products are not uniquely determined by their factors. In a central decomposition $\{G_1, G_2\}$, as in a direct product decomposition, it is true that both $G_i$ are normal subgroups. Unlike a direct decomposition, however, $G_1 \cap G_2$ need not be trivial. To make the discussion a little clearer, we introduce a standard alternative viewpoint on central decompositions:

\begin{definition}[{Central product; see, \eg, \cite[(11.1)]{aschbacher}}]
Given two groups $H_1, H_2$ and an isomorphism $\varphi \colon Y_1 \to Y_2$ between two subgroups $Y_i \leq Z(H_i)$ of their centers, the quotient of $H_1 \times H_2$ by $\{(y^{-1}, \varphi(y)) : y \in Y_1\}$ is the \emph{central product} of $H_1$ and $H_2$ along $\varphi$, denoted $H_1 \times_{\varphi} H_2$. 
\end{definition}

Central products and central decompositions are essentially equivalent. More specifically, if $\{G_1, G_2\}$ is a central decomposition of $G$, then if we let $Y_i = G_1 \cap G_2 \leq G_i$ for $i=1,2$, and define $\varphi\colon Y_1 \to Y_2$ to be the map induced by the identity map on $G_1 \cap G_2$ (thinking of both $G_1$ and $G_2$ as subgroups of $G$), then we see that $G$ is the central product of $G_1$ and $G_2$ along $\varphi$. Conversely, if $H = H_1 \times_{\varphi} H_2$ is a central product, then every element of $H$ can be written (not uniquely!) as the equivalence class of $(h_1, h_2)$ in the quotient $H_1 \times_{\varphi} H_2$, for some $h_1 \in H_1, h_2 \in H_2$. Let us denote this equivalence class by $\overline{(h_1, h_2)}$. Then it is readily verified that $G_1 = \{\overline{(h_1, 1)} : h_1 \in H_1\}$ is a subgroup of $H$ isomorphic to $H_1$, that $G_2 = \{\overline{(1,h_2)} : h_2 \in H_2\}$ is a subgroup of $H$ isomorphic to $H_2$, and that $\{G_1, G_2\}$ is a central decomposition of $H$.

When dealing with isomorphisms between central products, the fact that $G_1 \cap G_2$ is nontrivial becomes a source of difficulty, as in the following lemma. Although this lemma applies in a more general situation, we state it for the situation we are most interested in.

\begin{lemma} \label{lem:centprod}
Let $A$ be a solvable group and $B$ a central-radical group. Suppose that $\varphi_1, \varphi_2$ are two isomorphisms $Z(A) \to Z(B)$. Then $A \times_{\varphi_1} B$ is isomorphic to $A \times_{\varphi_2} B$ if and only if there are automorphisms $\alpha \in \Aut(A)$ and $\beta \in \Aut(B)$ such that $\varphi_1 = \beta^{-1} \circ \varphi_2 \circ \alpha$.
\end{lemma}

\begin{proof}
Let $G_i$ denote $A \times_{\varphi_i} B$, and let $Z_i$ denote the copy of $Z(A) \cong Z(B)$ in $G_i$ for $i=1,2$. Every element of $G_i$ can be written---not necessarily uniquely---as $ab$ with $A \in A$ and $b \in B$; $ab=1$ in $G_i$ if and only if $a \in Z(A)$, $b \in Z(B)$, and $\varphi_i(a) = b^{-1}$. 

First, suppose there are $\alpha \in \Aut(A)$ and $\beta \in \Aut(B)$ such that $\varphi_1 = \beta^{-1} \circ \varphi_2 \circ \alpha$. Then we claim that the map sending $ab \in G_1$ to $\alpha(a) \beta(b) \in G_2$ is both well-defined and an isomorphism. To see that it is well-defined note that if $ab=1$ in $G_1$, then $\varphi_1(a) = b^{-1} \in Z(G_1)$. By assumption, we then have $\beta^{-1}(\varphi_2(\alpha(a))) = b^{-1}$, or equivalently $\varphi_2(\alpha(a)) = \beta(b)^{-1}$, which means that $\alpha(a)\beta(b) = 1$ in $G_2$. From this, one concludes that if $ab = a'b'$ then $\alpha(a) \beta(b) = \alpha(a') \beta(b')$; it is then easily verified that this map is in fact a homomorphism. Since these are finite groups, injectivity then suffices to show it is an isomorphism; injectivity follows using the preceding argument run in reverse.

For the reverse direction, we will need to use the following fact. Let $\pi_i \colon A \times B \to A \times_{\varphi_i} B$ be the natural quotient map. Then we claim that $\pi_i(1 \times Z(B)) = Z(G_i)$. It is easy to see that $\pi_i(1 \times Z(B)) \leq Z(G_i)$. For the reverse inclusion, suppose $z \in Z(G_i)$. Then $z$ must commute with both $\pi_i(1 \times B)$, and since $\pi_i$ maps $1 \times B$ isomorphically onto its image, this means that $z \in \pi_i(1 \times Z(B))$. 

Now, suppose there is an isomorphism $\gamma\colon G_1 \to G_2$. Note that, as $A$ is a solvable normal subgroup of $G_i$ and $G_i / A \cong B / Z(B)$ is semisimple, $A$ must be equal to $\rad(G_i)$. In particular, since the radical is a characteristic subgroup function, $\gamma(A) = A$ and thus $\alpha := \gamma|_{A}$ is an automorphism of $A$. 

Next we show that $\gamma$ sends $B$ to $B$, thus inducing an automorphism of $B$. Note that $G_i / Z(G_i) \cong (A / Z(A)) \times (B / Z(B))$. Since the center is a characteristic subgroup function, $\gamma$ induces an isomorphism $G_1 / Z(G_1) \to G_2 / Z(G_2)$, that is, an automorphism of $(A / Z(A)) \times (B / Z(B))$. Since $A$ is solvable and $B/Z(B)$ is semisimple, $\gamma$ thus induces an automorphism of $B/Z(B)$ as well. Since $\gamma(Z(B)) = \gamma(Z(G_1)) = Z(G_2) = Z(B) \leq B$, by a counting argument $\gamma$ must send $B$ to $B$. Thus $\gamma$ induces $\beta \in \Aut(B)$. 

Thus, by construction, for $a \in A$ and $b \in B$, $\gamma(a) = \alpha(a)$ and $\gamma(b) = \beta(b)$. Since $\gamma$ is a homomorphism, it follows that $\gamma(ab) = \alpha(a) \beta(b)$ for all $a \in A$ and all $b \in B$. To see that $\varphi_1 = \beta^{-1} \circ \varphi_2 \circ \alpha$, one uses the same calculation as in the first direction.
\end{proof}

The \algprob{Central Amalgam Problem} is: given two automorphisms $\varphi_i \in 
\aut(Z)$, $i=1, 2$, of an abelian group $Z$, two black-box groups $G, H$ (think of 
these as $\aut(A)$ and $\aut(B)$ in the preceding lemma), and actions of $G$ and 
$H$ on $Z$---given by specifying the matrix actions of \emph{generating sets} of $G$ and 
$H$---decide whether there exists $\alpha \in G$ and $\beta \in H$ such that 
$\varphi_1  = \beta^{-1} \circ \varphi_2 \circ \alpha$.

%For a prime $p$ and an abelian group $A$, we say that the \emph{$p$-rank} of $A$ is the maximal dimension over $\Z_p$ of any elementary abelian $p$-group that is a quotient of $A$. Equivalently, if $A \cong \Z_{p_1^{\mu_1}} \times \dotsb \times \Z_{p_k^{\mu_k}}$ with the $p_i$ prime, the the $p$-rank of $A$ is simply the number of indices $i$ with $p_i = p$. 

\begin{proposition} \label{prop:action_triv}
Let $\charfn$ be a polynomial-time computable characteristic subgroup function. Suppose that $G_1, G_2$ are two groups for which the induced outer action of $G_i / \charfn(G_i)$ on $\charfn(G_i)$ by conjugation is trivial (equivalently: the induced action is by inner automorphisms of $\charfn(G_i)$). Then the group isomorphism problem for $(G_1, G_2)$ reduces in polynomial time to finding a generating set for $\Aut(\charfn(G_i))$ and $\Aut(G_i|_{Z(\charfn(G_i))})$ and solving the \algprob{Central Amalgam Problem}.

In particular, group isomorphism for groups for which the outer action of $G/\rad(G)$ on $\rad(G)$ is trivial reduces in $n^{O(\log \log n)}$-time to finding generating sets of the automorphism group of solvable groups and solving the \algprob{Central Amalgam Problem}.
\end{proposition}

\begin{remark}\label{remark:action_triv}
When $|\aut(Z(\rad(G))|$ is bounded by $\poly(|G|)$, the \algprob{Central Amalgam Problem} can be solved in $\poly(|G|)$ time by standard permutation group algorithms. Note that the class of solvable groups $S$ whose centers have automorphisms groups of polynomial size may seem restrictive, but is in fact quite rich. In particular, it includes solvable groups whose centers are abelian groups with arbitrarily many factors, as long as each prime appears in a bounded number of factors, and also includes all centerless solvable groups (itself quite a nontrivial class of groups).
\end{remark}

\begin{proof}
We show how to construct a central decomposition as in Lemma~\ref{lem:acttriv} in polynomial time. Let $N = \charfn(G)$ and $Q = G / N$. By assumption, the subset $N = \charfn(G)$ can be identified in $\poly(|G|)$ time. Next, choose any section $s\colon Q \to G$. It may be that some $s(q)$ acts non-trivially on $N$ via conjugation. However, by the assumption that the outer action is trivial, $c_{s(q)}$ must be some inner automorphism of $N$, say $c_{n(q)}$ for some $n(q) \in N$. To find this $n(q)$, we may search through $N$ exhaustively in at most $O(|N|^2) \leq O(|G|^2)$ time: essentially $|N|$ steps to check the action of a given $n$ on $N$ by conjugation, and there are $|N|$ possible $n$'s to check. Then let $s'(q) = s(q) n(q)^{-1}$; as $n(q) \in N$, $s'$ is another section, and by construction $c_{s'(q)} = \id_{N}$ for all $N$. Finally, let $f(p,q) = s'(p) s'(q) s'(pq)^{-1}$. Computing all the values of $f$ takes essentially $O(|Q|^{2}) \leq O(|G|^2)$ time, and then we construct $G|_{Z(\mathcal{S}(G))}$ as the subgroup of $G$ generated by $s(Q)$ and $Z(N)$ in $\poly(|G|)$ time. %Note that, since we constructed $G|_{Z(\charfn(G))}$ explicitly so that $Z(G|_{Z(\charfn(G))}) = Z(\charfn(G))$, we can easily constuct the isomorphism $\varphi\colon Z(G|_{Z(\charfn(G))}) \to Z(\charfn(G))$ that is implicit in the notation $\charfn(G) \times_{Z(\charfn(G))} G|_{Z(\charfn(G))}$.
\end{proof}

\paragraph{Extensions of centerless groups.} We have already mentioned a few useful properties of extensions of centerless groups, that is, when $Z(N) = 1$. One that is implicit in what we have already said is that every outer action $Q \to \Out(N)$ is extendable, that is, it is induced from some extension of $N$ by $Q$. These properties culminate in the following very useful theorem:

\begin{theorem}[{see, \eg, \cite[Thm.~2.7.11]{Suzuki}}] \label{thm:centerless}
Let $N$ be a centerless group, $Q$ any group, and $G$ an extension of $N$ by $Q$. Then $G$ is determined up to isomorphism by the induced outer action of $Q$ on $N$.

Furthermore, every such extension is equivalent to a subgroup $\Gamma \leq Q \times \Aut(N)$ satisfying $\Gamma \cap \Aut(N) = \Inn(N)$ and $\pi_Q(\Gamma) = Q$, where $\pi_Q \colon Q \times \Aut(N) \to Q$ is the projection onto the first factor.
\end{theorem}

In particular, if $\charfn$ is a characteristic subgroup function computable in polynomial time, and $\gpcls$ is a class of groups for which $\charfn(G)$ is centerless for every $G \in \gpcls$, then isomorphism of groups in $\gpcls$ reduces to isomorphism of groups of the type $G/\charfn(G)$ for $G \in \gpcls$, groups of the type $\charfn(G)$ for $G \in \gpcls$, and \outactcomp.

% !TEX root = main.tex

\section{The strategy}\label{sec:strategy}
Suppose we are given two groups $G_1$ and $G_2$ from some class of groups 
$\gpcls$. Our Main Lemma~\ref{lem:main_general} suggests (and indeed was motivated 
by) a divide-and-conquer strategy to test isomorphism (Section~\ref{sec:recipe}). 
This strategy highlights important structural features of \GpI, which we show are 
formally necessary in Section~\ref{sec:necessary}. It naturally suggests new 
group classes for which polynomial-time isomorphism tests might be within reach, 
and also suggests \emph{a priori} many group classes for which polynomial-time 
algorithms have previously been achieved. 
%how this strategy provides  several recent polynomial-time algorithms for special 
%classes of 
%groups (with the easy exceptions of abelian groups and groups generated by $O(1)$ 
%elements, such as simple groups).

However, before we proceed, let us emphasize that the 
extension viewpoint only 
helps with a \emph{conceptual} understanding of these previous works. Given this 
viewpoint, to tackle each group class may still require novel mathematical ideas 
and technically demanding algorithms. The extension viewpoint is mostly used 
to set the stage for applications of such mathematical and algorithmic 
techniques. In other words, instead of looking at a table encoding an 
\emph{abstract} group, an application of the Main Lemma usually transfers us to a 
more 
\emph{concrete} setting where we need to solve problems about, e.\,g., bilinear 
maps and permutation groups. We also stress that not all recent 
progress on \GpI can be captured from this extension viewpoint, e.\,g., 
\cite{Wil13,GR16}. 

\subsection{A recipe for group isomorphism} \label{sec:recipe}

\begin{enumerate}
\item Choose wisely a polynomial-time computable characteristic subgroup function $\charfn$. Note that if $\charfn(G)$ is always abelian, then the technically simpler abelian Main Lemma~\ref{lem:main} can be applied.

\item Test whether $\charfn(G_1)\cong\charfn(G_2)$ (which we henceforth refer to as $N$) and $G_1/\charfn(G_1)\cong G_2/\charfn(G_2)$ (which we refer to as $Q$). If either of these fails, then $G_1 \not\cong G_2$.

\item Extract the extension data $(T_i, f_i)$ from the extension $N \hookrightarrow G_i \twoheadrightarrow Q$ for $i=1, 2$ by picking arbitrary sections $s\colon Q \to G_i$ and computing the action and cohomology class.

\item Test pseudo-congruence of the two extension data. That is, find $(\alpha, \beta)\in\aut(N)\times \aut(Q)$, and a function $t\colon Q\to N$ such that $T_1(q)=c_{t(q)}\circ T_2^{(\alpha, \beta)}(q)$ and $f_1=(f_2^{(\alpha, \beta)})^{t,T_2}$. If the abelian Main Lemma~\ref{lem:main} applies, then $t$ is unnecessary.
 \end{enumerate}

Some general remarks are due for each of these steps:
\begin{enumerate}
\item A seemingly obvious requirement would be that $\charfn(G)$ should not be trivial for any $G\in\gpcls$. However, even if this is not the case, it may be fruitful to consider separately the class of groups for which $\charfn(G)$ is trivial. For example, semisimple groups arise this way, as those groups for which $\rad(G)$ is trivial.

\item Due to the nature of the divide and conquer strategy, $\charfn(G)$ and $G/\charfn(G)$ should be from group classes with known efficient isomorphism tests. Alternatively, if, say, $\charfn(G)$ is not from such a class, it may be possible to use this strategy to reduce isomorphism of groups in $\gpcls$ to isomorphism of groups of the form $\charfn(G)$ for $G \in \gpcls$ (or similarly for $G/\charfn(G)$).

\item This step is easy ($\poly(|G|)$ time). 
%\jnote{Removed reference to Cayley table} %in the Cayley table model. 
Based on the group class $\gpcls$ the extension data will hopefully turn out to have nice mathematical structure; indeed, looking for this nice mathematical structure is a nice heuristic that can help motivate and suggest various choices for $\gpcls$.

\item This pseudo-congruence test is the main bottleneck. Choosing $\charfn$ so that this step can take advantage of known cohomological results may be helpful. For example, if $\charfn(G) \leq Z(G)$ then \edpc simplifies to \cohiso; at the opposite end of the spectrum, if $G = \charfn(G) \rtimes (G/\charfn(G))$ then \edpc simplifies to \actcomp. As another example, if $\charfn(G)$ is centerless, one may take advantage of Theorem~\ref{thm:centerless}, as in the case of semisimple groups (see below).
\end{enumerate}

\subsection{Some recent results from the point of view of the main lemma} 
\label{sec:previous}

As mentioned in the introduction, there have been some recent polynomial-time 
algorithms for 
several group classes: semisimple groups \cite{BCGQ,BCQ}, generalized 
Heisenberg 
groups \cite{LW12}, groups with abelian Sylow towers \cite{BQ}, and (in this 
paper) $n^{o(\log n)}$-time algorithms for central-radical groups. The experts would easily see how the perspective 
of group extension helps to open a venue of attack to devise efficient 
algorithms for these group classes. However, for readers who have not seen such a 
connection, the 
definitions of these group classes may at first seem obscure, and it is not 
\emph{a priori} clear why we should have found efficient algorithms for these 
particular classes of groups, as opposed to others. 
We believe that the viewpoint 
of extensions and cohomology, especially in light of the Main Lemma, gives a 
unifying perspective to these works which helps to explain the progress on these 
group classes, thereby easing certain readers' understanding of 
these previous works.

In the following, we first summarize some basic information about these works, and then explain in detail how previous works on \GpI fit into the general strategy described as above. 
\vskip 1em
\noindent \begin{tabularx}{\textwidth}{ |l|X|X|X| }
\hline References & Group class & Characteristic subgroup function & Extension type  \\ 
\hline \cite{BCGQ,BCQ} & Semisimple groups & Socle\footnote{The \emph{socle} of a group $G$ is the subgroup generated by the union of the minimal normal subgroups of $G$.}  & Extension of a centerless group   \\ 
\hline \cite{LW12} & Quotients of generalized Heisenberg groups & Center & Special 
type of central extensions of $\Z_p^k$ by $\Z_p^\ell$   \\ 
\hline \cite{Gal09,QST11,BQ} & Groups with abelian Sylow towers & Normal Hall subgroups & Split extension of $A$ by $Q$ with $(|A|, |Q|)=1$   \\ 
\hline This work & Central-radical groups & Solvable radical & (Central) Extension of abelian groups by semisimple groups  \\ 
\hline
%\jnote{Added:} 
Follow-up work \cite{GQ15} & Groups with ``tame'' radicals & Any ``tame'' abelian 
& Tame extension of $A$ by $Q$ \\ \hline
\end{tabularx}

\paragraph{Semisimple groups (groups with no abelian normal subgroups).} In the 
polynomial-time algorithm for semisimple groups \cite{BCQ} we take $\mathcal{S} = 
\soc$, \ie, $N = \soc(G)$, which is a polynomial-time characteristic subgroup 
function. Hence the general Main Lemma~\ref{lem:main_general} applies and 
isomorphism of semisimple groups reduces to \edpc. When $G$ is semisimple, its 
socle is a direct product of nonabelian simple groups, so $Z(N) = 1$. $N$ being 
centerless simplifies some of the results in the previous section, as captured in 
Theorem~\ref{thm:centerless}, which corresponds to the key lemma in 
\cite{BCGQ} (Lemma 3.1 therein), and leads to the problems considered by Babai 
\etal \cite{BCGQ, BCQ}. In particular, note that in the definition of 
pseudo-congruence for nonabelian $N$, after twisting by $(\alpha, \beta) \in 
\Aut(N) \times \Aut(Q)$ to make the actions $T_1, T_2$ become equivalent as outer 
actions, the condition on the 2-cocycles is simply that they differ by a 
2-coboundary in $B^2(Q, Z(N), T)$. In particular, when $N$ is centerless $B^2(Q, 
Z(N), T)$ is trivial, so \edpc reduces to \outactcomp.

In the case of semisimple groups, using the structure of these 
groups, one sees quickly that \outactcomp reduces to \algprob{Twisted Code 
Equivalence} (introduced in \cite{BCQ}), where the ``twisting'' groups correspond 
to the action of $\Out(N) = \Out(\soc(G))$ in the definition of \outactcomp, and 
the choice of $t\colon Q \to N$ is handled by considering codes whose codewords 
correspond to elements of $G$ rather than just elements of $Q$.

\paragraph{$p$-groups of class $2$ and exponent $p$, esp. quotients of generalized 
Heisenberg groups.} For 
$p$-groups of class $2$ and exponent $p$ with odd $p$, Baer's correspondence 
\cite{baer} suggests considering the alternating bilinear maps defined by the 
commutator bracket: isomorphism of $p$-groups corresponds to pseudo-isometry of 
these bilinear maps. These bilinear maps are 2-cocycles, and two such cocycles are 
isomorphic as cohomology classes if and only if the bilinear maps are 
pseudo-isometric, so we see that this is a particular instance of \cohiso and 
Baer's correspondence can be viewed as a special case of the abelian Main 
Lemma~\ref{lem:main}. 

The bilinear map viewpoint has been the main stage for 
the recent progress on testing isomorphism of such $p$-groups \cite{LW12,BMW15}. 
For example, in \cite{LW12}, 
Lewis and Wilson \cite{LW12} studied a decently large class of 
$p$-groups---quotients of generalized Heisenberg groups---which are 
indistinguishable to classical invariants but for which they nonetheless present a 
polynomial-time isomorphism test. Such groups admit a nice characterization from 
the bilinear map viewpoint \cite[Theorem 3.13]{LW12}, and the 
polynomial-time isomorphism test for these groups takes advantage of the special 
structure of the bilinear maps corresponding to these groups 
(\cite[Theorem~4.1]{LW12}). We remark that the algorithm 
in \cite{LW12} works with much more succinct models for representing 
groups, including permutation groups and matrix groups, and runs in time 
polynomial %\jnote{Added ``in the input size, which can be as small as''} 
in the input size, which can be as small as 
$\log |G|$ instead of $|G|$. 

%These groups can alternatively be characterized as those for which the centroid 
%of 
%the above bilinear map is a field (see \cite[Theorem~3.1]{LW12}). 

\paragraph{Groups with abelian Sylow towers.} Though solving \GpI for the obscure-sounding group class ``groups with abelian Sylow towers,'' the core of \cite{BQ} (following \cite{Gal09, QST11}) deals with the case of \emph{coprime extensions}, namely extension of an abelian $A$ by $Q$ where $(|A|, |Q|)=1$. The Schur--Zassenhaus Theorem guarantees that coprime extensions split, thus reducing \edpc to \actcomp for such groups. Assuming $\aut(Q)$ is known (via recursive divide-and-conquer), \cite{BQ} views the actions of $Q$ on $A$ as linear representations and utilizes the complete reduciblility of these representation by Maschke's Theorem, which requires coprimality. The tactic they use is to view the induced action of $\aut(Q)$ on the irreducible constituents as a permutation group action, and then to develop a parameterized permutation group algorithm to finally solve \actcomp in this case \cite{BQ}.

\paragraph{Tame extensions.} %\jnote{Added paragraphs:}  
In a follow-up work \cite{GQ15}, the current authors used the viewpoint of this 
paper to generalize the preceding from coprime extensions to so-called ``tame'' 
extensions. These are extensions of $\Z_p^k$ by $Q$ (assuming $\aut(Q)$ is known, 
e.\,g., by recursive divide-and-conquer) where the Sylow $p$-subgroups of $Q$ are 
cyclic, or $p=2$ and the Sylow 2-subgroups are dihedral, semi-dihedral, or 
generalized quaternion. (The case above is when the Sylow $p$-subgroups of $Q$ are 
trivial.) 

In fact, the story behind \cite{GQ15} is a perfect example of the utility of explicitly splitting \GpI into \actcomp and \cohiso. Namely, independently, one of the current authors had solved \actcomp for the tame case, and the other had solved \cohiso for the tame case under the assumption that \actcomp could be solved; when they met in Chicago each was eager to tell the other of their result, asking if the other ``half'' of the problem could be solved. The result was nearly immediate from there.

\paragraph{Central-radical groups.} Similarly, by considering cohomology rather than actions we will see in the following how to handle central-radical groups. An elementary way of manipulating the 2-cohomology classes yields an $n^{O(\log \log n)}$-time algorithm for groups with central radicals. For a subclass of groups with central radicals, a more detailed understanding of 2-cohomology classes (Lemma~\ref{lem:prod}) helps establish the polynomial-time algorithms in Theorem~\hyperref[thm:ecentradwsoc]{C}=\ref{thm:ecentradwsoc}. In particular, singly exponential algorithms for \lincode and for \cosetint enter inevitably in the algorithm for this theorem.

\paragraph{Summary.} In summary, although the strategy we propose will rarely solve the problem completely, even for a restricted class of groups, it has the virtue of quickly dispensing with structural issues to highlight the needed algorithmic tactics.

\subsection{Necessity of pseudo-congruence and cohomology} \label{sec:necessary}

Lemma~\ref{lem:main} suggests studying \edpc to make progress towards \GpI for 
groups with abelian normal subgroups. In this section, we shall see that 
pseudo-congruence tests for certain classes of extension data are exactly 
isomorphism tests for certain interesting group classes. While 
Lemma~\ref{lem:main} almost implies so, a pitfall is that in the reconstruction 
procedure we need the normal copy of $A$ in $H$ to be the image of a 
characteristic subgroup function. (In the setting of Lemma~\ref{lem:main} the 
standard reconstruction procedure does return groups $H_i$ with the copy of 
$A=\charfn(H_i)$, but this is because of the assumptions of that lemma.)
%\jnote{Moved footnote to main text and changed ``conditions'' to 
%``assumptions.''} 
This leads us to look at some concrete classes of extension data, for which this 
property holds.

For split extensions, a well-known example is the case when $|A|$ and $|Q|$ are 
coprime, as ensured by the Schur--Zassenhaus Theorem. In this case $G$ is said to 
be a coprime extension of $A$ by $Q$, and $A$ is a normal Hall subgroup in 
$G$. %\jnote{Removed footnote recalling Schur--Zassenhaus} %\footnote{The 
%%Schur--Zassenhaus Theorem states that a coprime extension is  split, regardless 
%%of 
%%whether $A$ is abelian or not. When $A$ is abelian the proof  is 
%%straightforward. 
%%The bulk of the proof of the Schur--Zassenhaus Theorem is devoted to the case 
%%when 
%%$A$ is nonabelian.} 
Noting that taking a normal Hall 
subgroup of a specific order is a characteristic subgroup function, with the 
standard reconstruction procedure we have:

\begin{fact} \label{fact:nec_actcomp}
There is a polynomial-time function $r$ which takes any group action $\theta \colon G \to \aut(B)$ (for any groups $B, G$) to a group $r(\theta)$ with the following property.  When $A$ is abelian, $Q$ is a group of order coprime to $|A|$, and $\theta_i\colon Q \to \aut(A)$ ($i=1,2$) are group actions,  then $(\theta_1, \theta_2) \mapsto (r(\theta_1), r(\theta_2))$ is a Karp reduction from these instances of \actcomp to \GpI.
\end{fact}

A polynomial-time algorithm for \actcomp for the case in Fact~\ref{fact:nec_actcomp} was given by Babai and Qiao \cite{BQ}, yielding a polynomial-time time algorithm for ``groups with abelian Sylow towers.''

\begin{remark}
Despite the polynomial-time algorithm for \actcomp for coprime case---which trivially implies there is a Karp reduction to \emph{any} problem, including \GpI---Fact~\ref{fact:nec_actcomp} remains nontrivial. In particular, recall that given equivalence relations $\sim_1$ on $X_1$ and $\sim_2$ on $X_2$, a \emph{kernel reduction} between them is a function $f\colon X_1 \to X_2$ such that $x \sim_1 y$ if and only if $f(x) \sim_2 f(y)$ \cite{FortnowGrochowPEq}. Fact~\ref{fact:nec_actcomp} states that there is a polynomial-time kernel reduction, which does \emph{not} follow automatically from the polynomial-time algorithm for coprime \actcomp.
%Note that, although there is a polynomial-time algorithm for the class of groups considered in Fact~\ref{fact:nec_actcomp} \cite{BQ}, this does not make Fact~\ref{fact:nec_actcomp} trivial. All this existence guarantees is that there is \emph{some} (in fact, trivial) Karp reduction between the two problems, but not necessarily one of the form specified in Fact~\ref{fact:nec_actcomp}, which is precisely why we stated it as we did. For simplicity we have stated Facts~\ref{fact:nec_cohiso} and/or \ref{fact:nec_semisimple} in terms of general Karp reductions, but they can be rephrased as Fact~\ref{fact:nec_actcomp} so that this remark applies to them as well.
\end{remark}

For central extensions, let $p\neq 2$ be a prime. If $A$ and $Q$ are both abelian, then a $\Z$-bilinear map $f:Q\times Q\to A$ is a 2-cocycle, as the cocycle identity follows directly from bilinearity. Note that the action of $\aut(A)\times\aut(Q)$ preserves bilinearity. The following proposition is known (\cite{baer,lazard}, see also \cite[Section 5]{warfield} and \cite{Wil09a}); the standard reconstruction procedure is altered to make the image of $A$ the commutator subgroup.

\begin{fact} \label{fact:nec_cohiso}
Given a prime $p\neq 2$, and finite abelian $p$-groups $A$ and $Q$
%$A=\Z_p^k$ and $Q=\Z_p^\ell$, 
let $f_i:Q\times Q\to A$ be an alternating %$\Z_p$-
bilinear map. Then \cohiso for $f_1, f_2$ Karp-reduces to \GpI.
\end{fact}

\begin{proof}
Given $f_i$ for $i=1,2$, alter the standard construction as follows. For $a, b\in A$ and $q, q'\in Q$, we define the group $G_i$ with operation $\circ$ over the set $A\times Q$ as
$
(a, q)\circ (b, q')=(a+b+\frac{1}{2}f_i(q, q'), q+q').
$
It is known that $G_i$'s are $p$-groups of class 2, %and exponent $p$, 
the copy of 
$A$ in $G_i$ is the commutator subgroup, and $f_1\pc f_2$ if and only if $G_1\cong 
G_2$ (see, \eg, \cite[Section 5]{warfield}).
\end{proof}

Finally let us examine groups whose solvable radicals are abelian, a super-class of central-radical groups. When $Q$ is semisimple, the standard reconstruction procedure sends $A$ to the solvable radical. (A solvable normal subgroup $N$ is the solvable radical if and only if $G/N$ is semisimple.) This hints at the fact that for central-radical groups, group isomorphism is equivalent to \cohiso.
\begin{fact} \label{fact:nec_semisimple}
Let $A$ be abelian and $Q$ semisimple. For $i=1, 2$, let $\theta_i:Q\to\aut(A)$ be a homomorphism, and $f_i:Q\times Q\to A$ be a 2-cocycle in $Z^2(Q, A, \theta_i)$. Then \edpc for $(\theta_1, f_1)$ and $(\theta_2, f_2)$ Karp-reduces to \GpI.
\end{fact}

% !TEX root = main.tex
\section{Preliminaries for the algorithms}\label{sec:prel_concrete}

Some general notations are described at the beginning of \Sec{prel_framework}.

\paragraph{Further notations and some group-theoretic facts.} Given a finite set $\Omega$, $\sym(\Omega)$ denotes the symmetric group consisting of all permutations of $\Omega$. A permutation group acting on $\Omega$ is a subgroup of $\sym(\Omega)$. Given $\pi\in \sym(\Omega)$ and $a\in \Omega$, the image of $a$ under $\pi$ is denoted by $a^\pi$. If $\Omega=[n]$, $n\in\N$, we use $S_n$ to denote $\sym(\Omega)$, and $A_n\leq S_n$ consists of permutations of even signs. For a vector space $V$ over a field $\F$, the general linear group $\GL(V)$ consists of all non-singular linear transformations of $V$. If $V=\F_q^n$, $q$ is a prime power, we may write $\GL(n, q)$ for $\GL(V)$.

By the Fundamental Theorem of Finite Abelian Groups, a finite abelian group is isomorphic to a direct product of cyclic groups of prime power orders.  Formally, let $A$ be an abelian group, then there exists a direct product decomposition of $A$ as $A=\langle e_1\rangle\times \langle e_2\rangle \times\dots\times \langle e_n\rangle$, where $e_i\in A$ has order $p_i^{k_i}$, such that $p_1\le p_2\le\cdots\le p_n$, and if $p_i=p_{i+1}$, then $k_i\le k_{i+1}$, for all $i$. This decomposition is called the primary decomposition of $A$, and the tuple $(e_1, \dots, e_n)$ forms a basis of $A$.  The elementary abelian groups are those groups of the form $\Z_p^n$ for some prime $p$ and any $n$. Note that $\aut(\Z_p^n)\cong\GL(n, p)$.

A group $G$ is \emph{simple} if $|G|>1$ and $G$ has no proper nontrivial normal subgroups. The celebrated Classification of Finite Simple Groups lists all finite simple groups explicitly \cite{ATLAS}. The only abelian simple groups are the cyclic groups of prime order. We use the following fact, which (currently) depends on the Classification for its proof:
\begin{fact}[\cite{Stein, AG}]\label{fact:genby2}
Every nonabelian simple group can be generated by 2 elements.
\end{fact}
Let $T$ be a nonabelian simple group, it is easily shown that $\aut(T^k)\cong \aut(T)\wr S_k$ where $\wr$ denotes the wreath product. If a group $G$ is a direct product of nonabelian simple groups, then this direct product decomposition is unique, not just up to isomorphism: if $G=T_1\times \dots \times T_k=S_1\times \dots\times S_\ell$, $T_i$, $S_j$ nonabelian simple, then $k=\ell$ and $\exists\sigma\in S_k$, $\forall i\in[k]$, $T_i=S_{i^{\sigma}}$ as subsets of $G$.

A group $G$ is \emph{perfect} if $G = [G,G]$, \emph{centerless} if $Z(G) = 1$, and \emph{(directly) indecomposable} if it cannot be written as a nontrivial direct product. Nonabelian simple groups are perfect, centerless, and indecomposable, and several properties of nonabelian simple groups generalize to this class. In particular, if $T$ is a perfect, centerless, indecomposable group then $\aut(T^k) \cong \aut(T) \wr S_k$. Furthermore, if $G = T_1 \times \dotsb \times T_k$ where each $T_i$ is perfect, centerless, and indecomposable, then the direct factors of $G$ are uniquely determined as subsets of $G$ (not just up to isomorphism), as in the case of nonabelian simple groups. 

(However, perfect, centerless, indecomposable groups can require many more than two generators, and correspondingly can have quite large automorphism groups, for example, the group $\{(\pi_1, \dotsc, \pi_k) | \pi_i \in S_5, \prod_{i \in [k]} \pi_i \in A_5\}$.)

The following observation is often useful to upgrade algorithms from merely deciding isomorphism to computing the entire coset of isomorphisms. The reader should have in mind that this observation will typically be applied in settings where the groups $X,Y,Z$ are automorphism groups of other groups. %Recall that, given an action $\varphi\colon K \to \aut(H)$, one may form the \emph{semidirect product} group $H \rtimes_{\varphi} K$ whose underlying set is $H \times K$, where $H$ is a normal subgroup of $H \rtimes K$, $K$ is a subgroup (not normal, in general), and $khk^{-1} = \varphi(k)(h)$ for all $k \in K, h \in H$.

\begin{observation} \label{obs:genSubProd}
Let $X$ be a subgroup of an extension $\extensionl{Y}{\iota}{\overline{X}}{\pi}{Z}$. If $\mathcal{Y} \subseteq \iota(Y)$ generates $\iota(Y) \cap X$, and $\mathcal{Z} \subseteq Z$ generates $\pi(X)$, and for each $z \in \mathcal{Z}$, $x_z$ is such that $x_z \in X$ and $\pi(x_z) = z$, then $\mathcal{Y} \cup \{ x_z \in X : z \in \mathcal{Z}\}$ generates $X$.
%Let $G$ be a subgroup of $H \rtimes K$, let $\pi_{K} \colon G \to K$ denote the natural projection onto $K$ with kernel $H$, and let $G_H$ denote the intersection $G \cap (H \rtimes 1)$. If $\mathcal{H} \subseteq H \rtimes 1$ generates $G_H$ and $\mathcal{K} \subseteq K$ generates $\pi_K(G)$, and for each $k \in \mathcal{K}$, $h_k$ is such that $(h_k, k) \in G$, then $\mathcal{H} \cup \{(h_k, k) \in G : k \in \mathcal{K}\}$ generates $G$.
\end{observation}

\begin{proof}
Given $x \in X$, first  write $z = \pi(x)$ as a word in the generators $\mathcal{Z}$, say $z = z_1 \dotsb z_{\ell}$, with each $z_i \in \mathcal{Z}$. Then $y = x \cdot (x_{z_1} \dotsb x_{z_\ell})^{-1}$ is in $X \cap \ker(\pi) = X \cap \iota(Y)$. Now write $y$ as a word in $\mathcal{Y}$.
%Given $(h,k) \in G$, first we write $k$ as a word in the generators $\mathcal{K}$, say $k = k_1 \dotsb k_\ell$, with each $k_i \in \mathcal{K}$. Then $(h,k)\cdot \left( (h_{k_1}, k_1) (h_{k_2}, k_2) \dotsb (h_{k_{\ell}}, k_\ell) \right)^{-1}$ is of the form $(h', 1)$, which is in $G_H$. Write $(h', 1)$ as a word in $\mathcal{H}$. 
\end{proof}

In Theorem~\ref{thm:ecentradwsoc}(1), we will need the following deep result bounding the size of the second cohomology of simple groups:

\begin{theorem}[{\cite[Thm.~B]{GKKL}}] \label{thm:GKKL}
For any finite (quasi\footnote{$G$ is \emph{quasisimple} if it is perfect, and $G/Z(G)$ is simple. Examples include $\text{SL}_n(\F_q)$ and $\GL_n(\F_q)$. We will only need the case of nonabelian simple groups, but the theorem holds in the generality of quasisimple groups, which might be useful in the future.})simple group $G$, any field $\F$, and any $\F$-representation $\theta \colon G \to \GL(V)$ ($V$ an $\F$-vector space), $\dim_{\F} H^2(G, V, \theta) \leq 17.5 \dim_{\F} V$.
\end{theorem}

\paragraph{Useful algorithms.} We shall need some known algorithmic results. %, which were also useful for previous results on group isomorphism such as \cite{BCGQ,QST11}. 
Recall that in permutation group algorithms (see \cite{Luk93,seressbook}), a coset $P\sigma\subseteq S_n$ is represented by a set of generators for $P\leq S_n$ and a coset representative $\sigma$. A particularly relevant problem on permutation groups is
the \cosetint problem: given two
cosets of subgroups of $\sym(A)$, find their intersection.
\GrI can be Karp-reduced to \cosetint \cite{luks-bounded}.
The \cosetint problem for permutation groups of degree $n$ can be
solved in quasi-polynomial ($\exp((\log n)^{O(1)})$) time \cite{Bab16}, 
%$\exp({\tilde{O}}(\sqrt{n}))$ time \cite{Bab83} (see also \cite{Bab10, BKL83}), 
while a relatively easier singly exponential ($\exp(O(n))$) 
algorithm has been obtained by Luks \cite{Luk99}. Algorithms over 
finite-dimensional algebras have been considered in, \eg, \cite{CIK97,BL08,IKS10}. 
In particular, over finite fields, polynomial-time algorithms for \modiso and 
\modcyc are first shown in \cite{CIK97}. 

A linear code of length $n$ is a linear subspace $V \leq \F^n$, represented by a $d\times n$ matrix where $d = \dim(V)$, and the rows form a linear basis of $V$. $S_n$ acts on a linear code by permuting the coordinates (that is the columns of the matrices). Two linear codes $V, U\leq \F^n$ are equivalent if there exists a permutation $\sigma\in S_n$ such that $V^\sigma=U$ as linear subspaces. Such a $\sigma$ is called an equivalence between $V$ and $U$, and the set of all equivalences, denoted by $\CodeEq(V, U)$ is either empty or a coset in $S_n$. This problem is \GrI-hard in general \cite{PR97} while Babai presents a singly exponential time algorithm:

\begin{theorem}[Babai, \cite{BabCode}, cf. \cite{BCGQ}]\label{thm:code_eq}
The set of equivalences of two linear codes of length $n$
(over any field) given by generator matrices can be found in
$(2 + o(1))^{n}$ time, assuming field operations at unit cost.
\end{theorem}

We will also need the following results of Babai \etal \cite{BCGQ}:

\begin{theorem}[{\cite[Thm.~1.1]{BCGQ}}]\label{thm:semisimple}
All isomorphisms between two semisimple groups $Q_1$ and $Q_2$ of order $n$, can be listed in time $n^{c\log\log n+O(1)}$, where $c=1/\log(60)\approx 0.16929$.
\end{theorem}

It is also noted in \cite{BCGQ} that there exist semisimple groups $G$ of order $n$ with $|\aut(G)|\geq n^{c\log\log n}$, namely $G=A_5^k$. Hence, for listing all isomorphisms this result is essentially optimal.

The number of minimal normal subgroups of any semisimple group of order $n$ is at 
most $O(\log n)$. If it happens to be $O(\log n / \log \log n)$, they show:

\begin{theorem}[{\cite[Cor.~4.4]{BCGQ}}]\label{thm:semisimple2}
Suppose $Q_1$ and $Q_2$ are semisimple groups of order $n$ with at most $O(\log n / \log \log n)$ minimal normal subgroups. Then all isomorphisms between $Q_1$ and $Q_2$ can be listed in polynomial time.
\end{theorem}

In \Sec{autq} we extend both of these results to algorithms for isomorphism of groups with central radicals, with the same time bounds as above.

To find the coset of isomorphisms efficiently in Theorem~\ref{thm:ecentradwsoc}, beyond just deciding isomorphism or finding a single isomorphism, we need the following result. Recall that an $d \times d$ \emph{matrix algebra} over a field $\F$ is a linear subspace of $d \times d$ matrices over $\F$ that is also closed under multiplication of matrices. The \emph{unit group} of a (matrix) algebra $A$ is the group consisting of all invertible elements of $A$ (the group operation being the same as the product in the algebra).

\begin{theorem}[{\cite[Thm.~2.3]{BrooksbankOBrien} and \cite{RonyaiUnit}}] \label{thm:unit}
Given a linear basis of a $d \times d$ matrix algebra $A$ over a finite field $\F_q$, generators of the unit group of $A$ can be computed in time polynomial in $d$ and $q$.
\end{theorem}

%\jnote{Moved from footnote to main text:} 
Although the preceding theorem from the original papers is stated as a Las Vegas 
randomized algorithm that runs in time polynomial in $d$ and $\log q$, they only 
need randomization to get an algorithm whose running time depends polynomially on 
$\log q$ rather than $q$. Essentially the same algorithm works deterministically 
in time polynomial in $d$ and $q$.

We also mention a useful result for groups in the Cayley table model by Kayal and Nezhmetdinov \cite{KN09}, though it is not strictly required in the following. They show that decomposing a group $G$ into indecomposable direct factors can be done in polynomial time. Even in the stronger setting of permutation groups given by generators, Wilson showed \cite{Wil10} that this task can be performed in polynomial time.

\section{When enumerating \texorpdfstring{$\aut(Q)$}{Aut(Q)} is allowed} \label{sec:autq}

Our main results in this section are $n^{O(\log\log n)}$-time algorithms to test 
isomorphism of (1) groups with central radicals 
(Corollary~\ref{cor:quotient_list}) and (2) groups with elementary abelian 
radicals that need not be central (Corollary~\ref{cor:quotient_list_general}). 
\cosetComment{
%\jnote{Changed to case (1), instead of both cases.} 
In case (1), if the radical is furthermore \emph{elementary} abelian, we can 
compute the coset of isomorphisms in the same time bound.} These results follow 
from our more general Theorems~\ref{thm:quotient_list} and 
\ref{thm:quotient_list_general}, respectively, and a theorem on semisimple groups 
from \cite{BCGQ} (reproduced above as Theorem~\ref{thm:semisimple}). 

\subsection{For central extensions of general abelian groups (Theorem~A)}

We first consider the case when both extensions $G$ and $H$ are central.
\begin{theorem}\label{thm:quotient_list}
Let $\charfn$ be a polynomial-time-computable characteristic subgroup function. For two groups $G, H$ of order $n$, if $\charfn(G) \leq Z(G)$ and $\aut(G/\charfn(G))$ can be listed in time $t(n)$, then isomorphism of $G$ and $H$ can be decided in time $t(n) n^{O(1)}$.

\cosetComment{If, furthermore, $\charfn(G)$ is elementary abelian, then the coset of isomorphisms can be found in the same time bound.}
\end{theorem}

Before proving Theorem~\ref{thm:quotient_list}, let us see how it is applied to groups with central radicals. Combining Theorem~\ref{thm:quotient_list} with Theorem~\ref{thm:semisimple}, respectively Theorem~\ref{thm:semisimple2} we have our first two main results:

\begin{corollary}\label{cor:quotient_list}
Isomorphism of central-radical groups of order $n$ can be decided in time $n^{c \log \log n + O(1)}$, for $c = 1/\log_2(60) \approx 0.169$. \cosetComment{If, furthermore, the center is elementary abelian, the coset of isomorphisms can be found in the same time bound.}
\end{corollary}

\begin{corollary} \label{cor:quotient_list2}
Let $G$ and $H$ be central-radical groups of order $n$. If $G/\rad(G)$ has $O(\log n / \log \log n)$ minimal normal subgroups, isomorphism between $G$ and $H$ can be decided in $\poly(n)$ time. \cosetComment{If, furthermore, the center is elementary abelian, the coset of isomorphisms can be found in the same time bound.}
\end{corollary}

We give two different proofs of Theorem~\ref{thm:quotient_list}. The first proof has the advantage of dealing with cohomology classes in a very direct manner and working for arbitrary central $\charfn(G)$, but the disadvantage of not making it obvious how to compute the full coset of isomorphisms. For clarity, we first give this ``direct'' proof in the elementary abelian case, then in \Sec{elem} give it for the general abelian case, and finally in \Sec{coset} go back to the elementary abelian case to show how to compute the coset of isomorphisms.

%for clarity we first deal with the elementary abelian case, \ie,  when $A=\Z_p^k$; the general abelian case will be handled in Section~\ref{sec:elem}. 

Let us consider how to work with 2-cohomology classes in algorithms. Let $G$ be a central extension of $A = \Z_p^k$ by $Q$ (thinking of $A = \charfn(G)$ and $Q = G/\charfn(G)$). As the action is trivial in central extensions, we drop it from the notation, as in $Z^2(Q, A)$, $B^2(Q, A)$ and $H^2(Q, A)$. By choosing an arbitrary section, we get a 2-cocycle $f:Q\times Q\to A$. Let $e_1, \dotsc, e_k$ be the standard basis of $\Z_p^k$. We may view $f$ as a $k\times |Q|^2$-size $\Z_p$-matrix, which we denote by $M_{f}$. The rows are indexed by the set $[k]$ and the columns are indexed by $Q \times Q$. For $i\in[k]$ and $(q, q')\in Q\times Q$, $M_{f}[i, (q, q')]$ is the $i$th coordinate of $f(q, q')$ relative to the basis $\{e_1, \dotsc, e_k\}$. Note that the actions of $\aut(A)$ and $\aut(Q)$ commute.

Under the above identification, the set $C^2(Q, A)$ of 2-cochains is identified with the set of all $k\times |Q|^2$ matrices over $\Z_p$. Then $Z^2(Q, A)$ is not just a subgroup (under matrix addition), but also a $\Z_p$-linear subspace of $C^2(Q, A)$, and similarly $B^2(Q, A)$ is a $\Z_p$-linear subspace of $Z^2(Q, A)$. $\Aut(A) \cong \GL(k, p)$ acts on $C^2(Q,A)$ by left multiplication, and $\Aut(Q)$ acts on $C^2(Q,A)$ by permuting the columns according to the diagonal action of $\Aut(Q)$ on $Q \times Q$.

\begin{proposition}\label{prop:basis_1}
A basis of $B^2(Q, \Z_p)$ can be computed in time $O(|Q|^3 (\log|Q| + \log p))$.
\end{proposition}

In the larger context of \GpI, note that the running time here is $O(|G|^3 \log |G|)$.

\begin{proof}
For $q\in Q$, $q\neq \id$, let $u_q:Q\to \Z_p$ be $u_q(q')=\delta(q, q')$ where $\delta$ is the Kronecker delta. Let $f_q:Q\times Q\to \Z_p$ be the 2-coboundary based on $u_q$. $V:=\{f_q\mid q\in Q\}$ then forms a basis of $B^2(Q, \Z_p)$. There are $|Q|$ basis elements, each of which is constructed by computing its $|Q|^2$ values; each value can be computed by a constant number of additions in $\Z_p$ (taking $O(\log p)$ steps) and one table lookup to compute a single product in $Q$ (taking $O(\log |Q|)$ steps).
\end{proof}

As we identified $C^2(Q, A)$ as $k\times |Q|^2$ matrices over $\Z_p$, let $E_{i, 
j}$, $i\in[k]$, $j\in Q\times Q$ be the $k\times |Q|^2$ matrix with a $1$ in the $(i,j)$ position and $0$ everywhere else. Then $\{E_{i, j}\mid i\in[k], 
j\in Q\times Q\}$ is a basis of $C^2(Q, A)$. Let $U_i$ be the subspace of $C^2(Q, 
A)$, spanned by $\{E_{i, j}\mid j\in Q\times Q\}$, corresponding to matrices whose 
only nonzero entries are in the $i$-th row. Then $C^2(Q, A)=\oplus_{i\in[k]} U_i$. 
The following proposition says that not only does $C^2(Q, A)$ split as a direct 
sum over the rows, but $B^2(Q, A)$ does as well. It follows directly from the fact 
that the condition to be a 2-coboundary in $B^2(Q, A)$ only depends on the columns 
($Q \times Q$) and not on the rows ($[k]$).

\begin{proposition}\label{prop:basis}
Let $V$ be the basis of $B^2(Q, \Z_p)$ constructed in the proof of Proposition~\ref{prop:basis_1}, and let $V_i\leq C^2(Q, A)$ be a copy of $V$ in $U_i$. Then $\bigsqcup_{i\in[k]}V_i$ (disjoint union) is a basis of $B^2(Q, \Z_p^k)$.
\end{proposition}

Given two 2-cocycles $f_1$ and $f_2$, let $M_i$ be the matrix representation of $f_i$, $i=1, 2$, and $R_i\subseteq \Z_p^{|Q|^2}$ be the set of rows in $M_i$. Recall that $\alpha\in\GL(k, p)$ acts on the rows of $M_i$. 
\begin{proposition}\label{prop:span}
With notation as above, there exists $\alpha\in\GL(k, p)$ such that $f_1$ and $f_2^\alpha$ are cohomologous if and only if $\langle R_1, B^2(Q, \Z_p) \rangle=\langle R_2, B^2(Q, \Z_p) \rangle$, where $\langle \cdot\rangle$ denotes the $\Z_p$-linear span.
\end{proposition}
\begin{proof}
Let $r_{i, j}\in\Z_p^{|Q|^2}$ be the $j$th row in $M_i$, $j\in[k]$, $i=1, 2$. Let $B$ denote $B^2(Q, \Z_p)$. Note that Proposition~\ref{prop:basis} says that $B^2(Q, \Z_p^k) = B \oplus B \oplus \dotsb \oplus B$ ($k$ summands).

($\Rightarrow$) $f_1$ and $f_2^\alpha$ are cohomologous if and only if $f_1-f_2^\alpha \in B^2(Q, \Z_p^k)$. Let $r^{\alpha}_{2, j}$ be the $j$th row in the matrix representation of $f_2^{\alpha}$. By Proposition~\ref{prop:basis}, for every $i\in [k]$, $r_{1, i} - r^{\alpha}_{2, i}\in \langle V_i \rangle = B^2(Q, \Z_p) = B$. That is $r_{1, i}\in\langle R_2, B\rangle$ as $r^{\alpha}_{2, j}\in\langle R_2\rangle$ (note that the linear span of $R_2$, \ie, the rowspan of $M_2$, is left unchanged by the action of $\alpha$). Similarly we have $r_{2, i}\in\langle R_1, B\rangle$, $\forall i\in[k]$. This shows $\langle R_1, B\rangle=\langle R_2, B\rangle$.

($\Leftarrow$) For $\alpha\in\GL(k, p)$, again let $r^{\alpha}_{2, j}$ be the $j$th row of $f_2^{\alpha}$. Given $\langle R_1, B\rangle = \langle R_2, B\rangle$, we have $\langle R_1, B\rangle / B$ and $\langle R_2, B\rangle /B$ are the same as subspaces of $\Z_p^{|Q|^2}/B$. That means that we can choose $\alpha\in\GL(k, p)$ such that $r_{1, i}+ B=r^{\alpha}_{2, i} + B$, $\forall i\in[k]$. This gives $f_1-f_2^\alpha\in B^2(Q, A)$.
\end{proof}

\begin{proof}[Proof of Theorem~\ref{thm:quotient_list} (decision version only) when $\charfn(G)$ is elementary abelian]
We list $\aut(Q)$ in time $t(n)$. For $i=1, 2$, choose an arbitrary section of $Q$ in $G_i$ to get a 2-cocycle $f_i$. By the Main Lemma~\ref{lem:main}, it is necessary and sufficient to test whether there exists an $(\alpha, \beta) \in \aut(A) \times \aut(Q)$ such that $f_1$ and $f_2^{(\alpha, \beta)}$ are cohomologous. 

For each $\beta\in\aut(Q)$ we get $f_2'=f_2^{(\id, \beta)}$. We first use Proposition~\ref{prop:basis_1} to get a basis $V$ of $B^2(Q, \Z_p)$. Let $M_1$ be the matrix representation of $f_1$, and $M_2$ for $f'_2$. We now need to determine whether there exists $\alpha$ such that $f_1$ and $f_2'$ are cohomologous. By Proposition~\ref{prop:span} it is enough to decide whether the linear span of the rows of $f_1$ with $V$, and the linear span of the rows of $f_2'$ with $V$, are the same. This is a standard task in linear algebra and can be determined in time polynomial in $|Q|$ and $\dim_{\Z_p}|A|=k$. 
Lemma~\ref{lem:main} implies that $G_1 \cong G_2$ if and only if the preceding test succeeds for some $\beta \in \aut(Q)$.
\end{proof}

\subsubsection{From elementary abelian to general abelian}\label{sec:elem}
The proof here follows the same steps as in the elementary abelian case. As each abelian group $A$ is the direct product of its Sylow $p$-subgroups $A_p$, we essentially treat the case of a single Sylow $p$-subgroup, that is, when $A$ is an abelian $p$-group $\Z_{p^{\mu_1}} \times \dotsb \times \Z_{p^{\mu_k}}$ (not necessarily elementary). We begin by extending Propositions~\ref{prop:basis_1}--\ref{prop:span} to the case of abelian $p$-groups.

As such groups are no longer just vector spaces over $\Z_p$, we must speak of subgroups of $A$ rather than subspaces, and generating sets rather than $\Z_p$-bases. To emphasize the similarities, we use the terminology ``$\Z$-basis'' for ``irredundant generating set.''  Similarly for $C^2(Q, A)$, $Z^2(Q, A)$, and $B^2(Q, A)$. We represent a 2-cochain $f\colon Q \times Q \to A$ by a $k \times |Q|^2$ integer matrix, where we consider the entries in the $i$-th row modulo $p^{\mu_i}$, that is, as elements of $\Z_{p^{\mu_i}}$. As before, we use $U_i$ to denote the subgroup of $C^2(Q, A)$ consisting of matrices whose only nonzero entries are in the $i$-th row (in particular, $U_i \cong \Z_{p^{\mu_i}}^{|Q|^2}$). 

For these first two propositions, the proofs are the same as the analogous propositions above for elementary abelian $A$. 

\begin{proposition} \label{prop:basis_mu}
For any $\mu \geq 1$, a $\Z$-basis of $B^2(Q, \Z_{p^\mu})$ can be computed in time $O(|Q|^3 (\log|Q| + \mu \log p))$.
\end{proposition}

\begin{proposition} \label{prop:basis_1_mu}
Let $V^{(\mu)}$ denote the $\Z$-basis of $B^2(Q, \Z_{p^{\mu}})$ constructed in the proof of Proposition~\ref{prop:basis_mu}, and let $V_i^{(\mu_i)} \leq C^2(Q, A)$ be a copy of $V^{(\mu_i)}$ in $U_i$. Then $\sqcup_{i \in [k]} V_i^{(\mu_i)}$ (disjoint union) is a $\Z$-basis of $B^2(Q, \Z_{p^{\mu_1}} \times \dotsb \times \Z_{p^{\mu_k}})$.
\end{proposition}

Before giving the analog of Proposition~\ref{prop:span} for general abelian $A$, we recall the structure of $\aut(A)$ (see, \eg, the exposition in \cite{hillarRhea}); it is only slightly more complicated than the fact that $\aut(\Z_p^k) = \GL(k, p)$. First, if $A_p$ is the $p$-Sylow subgroup of $A$, then $\aut(A) = \aut(A_{p_1} \times \dotsb A_{p_d}) = \aut(A_{p_1}) \times \dotsb \times \aut(A_{p_d})$ where $p_1, \dotsc, p_d$ are the distinct primes dividing $|A|$. So we reduce to the case where $A$ is an abelian $p$-group $\Z_{p^{\mu_1}} \times \dotsb \times \Z_{p^{\mu_k}}$ with $1 \leq \mu_1 \leq \mu_2 \leq \dotsb \leq \mu_k$. Think of elements of $A$ as integer column vectors of length $k$, where the $i$-th entry is considered modulo $p^{\mu_i}$. As in the elementary abelian case (where $\mu_1 = \dotsb = \mu_k = 1)$, an automorphism may replace each entry with a $\Z$-linear combination of the entries, as follows. For $i < j$, the $i$-th coordinate can contribute to the $j$-th coordinate by multiplying by $p^{\mu_j - \mu_i}$---in other words, by using the unique inclusion $\Z_{p^{\mu_i}} \hookrightarrow \Z_{p^{\mu_j}}$. In the opposite direction, the $j$-th coordinate can contribute to the $i$-th coordinate by taking the $j$-th coordinate modulo $p^{\mu_i}$---in other words, using the natural surjection $\Z_{p^{\mu_j}} \twoheadrightarrow \Z_{p^{\mu_i}}$. (Note that when $\mu_i = \mu_j$ these two operations are the same, corresponding to the identity map on $\Z_{p^{\mu_i}}$.) 

More symbolically, we may consider each element of $\aut(A)$ as an integer $k \times k$ matrix $\alpha$ such that: (1) for $i > j$, $p^{\mu_j - \mu_i}$ divides the $(i,j)$ entry, (2) the entries in row $i$ are considered modulo $p^{\mu_i}$, and (3) $\alpha$ is invertible when taken modulo $p$. 

Finally, consider (row) subgroups $R \leq \Z_{p^{\mu_i}}^{|Q|^2}$. In accord with the above description of the automorphisms of $A$, for $\mu < \mu_i$ let $R^{(\mu)}$ denote the subgroup of $\Z_{p^{\mu}}^{|Q|^2}$ that is given by taking $R$ modulo $p^{\mu}$; for $\mu > \mu_i$, let $R^{(\mu)}$ denote the subgroup of $\Z_{p^{\mu}}^{|Q|^2}$ that is given by multiplying every element of $R$ by $p^{\mu - \mu_i}$. For any prime $q$, let $R^{(q,\mu)}$ denote $R^{(\mu)}$ if $q = p$ and the trivial subgroup $0$ otherwise.

Now, return to $A$ being an arbitrary abelian group. For a 2-cochain $f_1 \in C^2(Q, A)$ with corresponding $k \times |Q|^2$ matrix $M$ with $i$-th row $R_{1,i} \leq \Z_{p_i^{\mu_i}}^{|Q|^2}$, let $R_1^{(p,\mu)}$ denote the subgroup of $\Z_{p^\mu}^{|Q|^2}$ generated by all the $R_{1,i}^{(p,\mu)}$; we write $R_1^{(p,\mu)} = \langle R_{1,1}^{(p,\mu)}, \dotsc, R_{1,k}^{(p,\mu)} \rangle$. 

\begin{proposition} \label{prop:span_mu}
Let $A = \Z_{p_1^{\mu_1}} \times \dotsb \Z_{p_k^{\mu_k}}$ be an arbitrary abelian group (the $p_i$ are primes, not necessarily distinct). With other notation as above, there exists $\alpha \in \aut(A)$ such that $f_1$ and $f_2^{\alpha}$ are cohomologous if and only if $\langle R_1^{(p_i,\mu_i)}, B^2(Q, \Z_{p_i^{\mu_i}}) \rangle = \langle R_2^{(p_i,\mu_i)}, B^2(Q, \Z_{p_i^{\mu_i}}) \rangle$ for each $1 \leq i \leq k$, where $\langle \cdot \rangle$ denotes the $\Z$-span (=group generated by). 
\end{proposition}

\begin{proof}
Let $r_{i, j}\in\Z_{p_j^{\mu_j}}^{|Q|^2}$ be the $j$-th row in $M_i$, $j\in[k]$, $i=1, 2$. Let $B^{(p,\mu)}$ denote $B^2(Q, \Z_{p^\mu})$. Note that Proposition~\ref{prop:basis} says that $B^2(Q, \Z_{p_1^{\mu_1}} \times \dotsb \Z_{p_k^{\mu_k}}) = B^{(p_1,\mu_1)} \oplus \dotsb \oplus B^{(p_k,\mu_k)}$.

($\Rightarrow$) $f_1$ and $f_2^\alpha$ are cohomologous if and only if $f_1-f_2^\alpha \in B^2(Q, A)$. Let $r^{\alpha}_{2, j}$ be the $j$-th row in the matrix representation of $f_2^{\alpha}$. By Proposition~\ref{prop:basis_mu}, for every $i\in [k]$, $r_{1, i} - r^{\alpha}_{2, i}\in \langle V_i^{(p,\mu_i)} \rangle = B^2(Q, \Z_{p_i^{\mu_i}}) = B^{(p_i,\mu_i)}$. That is $r_{1, i}\in\langle R_2^{(p_i,\mu_i)}, B^{(p_i,\mu_i)}\rangle$ as $r^{\alpha}_{2, j}\in\langle R_2^{(p_i,\mu_i)}\rangle$ (note that the subgroup generated by of $R_2^{(p_i,\mu_i)}$ is, by definition, left unchanged by the action of $\alpha$). Similarly we have $r_{2, i}\in\langle R_1^{(p_i,\mu_i)}, B^{(p_i,\mu_i)}\rangle$, $\forall i\in[k]$. This shows $\langle R_1^{(p_i,\mu_i)}, B^{(p_i,\mu_i)}\rangle=\langle R_2^{(p_i,\mu_i)}, B^{(p_i,\mu_i)}\rangle$ for each $i$.

($\Leftarrow$) For $\alpha\in\aut(A)$, again let $r^{\alpha}_{2, j}$ be the $j$th row of $f_2^{\alpha}$. Given $\langle R_1^{(p_i,\mu_i)}, B^{(p_i,\mu_i)}\rangle = \langle R_2^{(p_i,\mu_i)}, B^{(p_i,\mu_i)}\rangle$ for each $i$, we have $\langle R_1^{(p_i,\mu_i)}, B^{(p_i,\mu_i)}\rangle / B^{(p_i,\mu_i)}$ and $\langle R_2^{(p_i,\mu_i)}, B^{(p_i,\mu_i)}\rangle /B^{(p_i,\mu_i)}$ are the same as subgroups of $\Z_{p_i^{\mu_i}}^{|Q|^2}/B^{(p_i,\mu_i)}$. That means that we can choose $\alpha\in \aut(A)$ such that $r_{1, i}+ B^{(p_i,\mu_i)}=r^{\alpha}_{2, i} + B^{(p_i,\mu_i)}$, $\forall i\in[k]$. This gives $f_1-f_2^\alpha\in B^2(Q, A)$.
\end{proof}

Finally, we come to the proof of Theorem~\ref{thm:quotient_list} for general abelian $A$:

\begin{proof}[Proof of Theorem~\ref{thm:quotient_list} (decision version only) for general abelian $\charfn(G)$]
The proof is the same as for the elementary abelian case, but using Propositions~\ref{prop:basis_1_mu} and \ref{prop:span_mu} instead of Propositions~\ref{prop:basis_1} and \ref{prop:span}, respectively. Checking the condition of Proposition~\ref{prop:span_mu} amounts to solving a system of equations over the abelian group $A$ (``linear algebra over $A$''), which can be done in polynomial time (see, \eg, \cite{GR02}).
\end{proof}

\subsubsection{Computing the coset of isomorphisms in the elementary abelian case} \label{sec:coset}
Here we give an alternative proof of Theorem~\ref{thm:quotient_list} in the elementary abelian case, which also provides the structure needed to compute the full coset of isomorphisms. The proof follows the same general lines as the proof above, except instead of including a basis of $B^2(Q,A)$ among the rows of an extended matrix $\widetilde{\tb{f_j}}$, it uses the following $\aut(A)$-invariant projection from $C^2(Q,A)$ to $C^2(Q,A)/B^2(Q,A)$. One might think that this approach could be generalized to the general abelian case following the idea in the previous section, thereby computing the coset of isomorphisms in that case as well; see the end of this section for a discussion of the difficulties in using that approach.

\begin{proposition}\label{prop:cob}
For $A=\Z_p^k$ and a group $Q$, let $n=|A|\cdot |Q|$. In time $O(n^2\log n)$ one can compute an $\aut(A)$-invariant complement $W$ of $B^2(Q, A)$ in $C^2(Q, A)$, and a $\Z_p$-linear projection $\pi$ from $C^2(Q, A)$ to $W$ that commutes with every $\alpha \in \aut(A)$.
\end{proposition}

\begin{proof}
%As in \Sec{autq}, we identify the set $C^2(Q,A)$ of 2-cochains with the set of all $k \times |Q|^2$ matrices over $\Z_p$. 
Let $V_0$ be the basis of $B^2(Q, \Z_p)$ from Proposition~\ref{prop:basis_1}, and let $W_0$ be a linear complement to $V_0$ in $C^2(Q, \Z_p)$ (which we can think of as the space of row vectors of length $|Q|^2$). For each row index $i \in [k]$, let $R_i$ denote the subgroup of $C^2(Q,A)$ consisting of those matrices whose only nonzero entries are in the $i$-th row, let $B_i$ be the copy of $B^2(Q, \Z_p)$ in $R_i$, let $W_i$ be the copy of $W$ in $R_i$, and let $\pi_i \colon R_i \to W_i$ be the projection of $R_i$ to $W_i$ along $B_i$. (Here we are identifying each $R_i$ with $C^2(Q, \Z_p)$ as in Proposition~\ref{prop:basis}.) Collectively, the $\pi_i$'s define a projection $\pi$ to $\oplus_i W_i$ along $\oplus_i B_i$ in $\oplus_i R_i=C^2(Q, A)$.

This $\pi$ makes it easy to identify the 2-cohomology class of $f\in Z^2(Q, A)\leq C^2(Q, A)$: $f, g\in Z^2(Q, A)$ are cohomologous if and only if $\pi(f)=\pi(g)$. Furthermore, since we have chosen $\pi$ to be ``the same'' on each row, it follows that for any $\alpha\in\aut(A)$, $\alpha\pi=\pi\alpha$---an important fact we shall need later. Also, the above procedure involves only standard linear algebra tasks, and such $\pi$ can be constructed in time $O(n^2\log n)$.
\end{proof}

To summarize, each $\alpha\in\aut(A)$, $\beta\in\aut(Q)$, and the map $\pi\colon C^2(Q,A) \to C^2(Q,A)/B^2(Q, A)$ just introduced can be viewed as a linear map on $C^2(Q, A)$, with $\alpha$ and $\beta$ nonsingular. We then note: (1) $\alpha$ and $\beta$ commute; (2) $\alpha$ and $\pi$ commute.

\begin{proof}[Proof of Theorem~\ref{thm:quotient_list} for elementary abelian $\mathcal{S}(G)$, with the coset of isomorphisms]
We list $\aut(Q)$ in time $t(n)$. For $i=1, 2$, choose an arbitrary section of $Q$ in $G_i$ to get a 2-cocycle $f_i$. Use Proposition~\ref{prop:cob} to get the projection $\pi:C^2(Q, A)\to W$ for some complement $W$ of $B^2(Q, A)$ such that $W$ is invariant under $\alpha$ and $\alpha\pi=\pi\alpha$ for every $\alpha\in\aut(A)$. By the Lemma~\ref{lem:main}, it is necessary and sufficient to test whether there exists an $(\alpha, \beta) \in \aut(A) \times \aut(Q)$ such that $\pi(f_1) = \pi(f_2^{(\alpha, \beta)})$. As every $\alpha$ commutes with both $\pi$ and $\beta$, this condition is equivalent to $\pi(f_1)=(\pi(f_2^{(\id, \beta)}))^{(\alpha, \id)}$. In other words, we may leave $\alpha$ unspecified until the final step.

For each $\beta\in \aut(Q)$, compute $f_1'=\pi(f_1)$, and $f_2'=\pi(f_2^{(\id, \beta)})$. Note that $f_i'$ are in $W\leq C^2(Q, A)$. The goal is then to find $\alpha\in\aut(A)$ such that $f_1'=\alpha^{-1}(f_2')$. %This is because $\alpha\pi=\pi\alpha$ allows us to apply $\pi$ to $f_2^{(\id, \beta)}$ first, and leave $\alpha$ to be determined later. 
$f_1'=\alpha^{-1}(f_2')$ for some $\alpha\in\aut(A)$ if and only if the row spans of $\tb{f_1'}$ and $\tb{f_2'}$ are the same in $\Z_p^{|Q|^2}$. The latter task is standard in linear algebra and can be checked in time $O(|Q|^6\log p)$. Lemma~\ref{lem:main} implies that $G_1 \cong G_2$ if and only if the preceding test succeeds for some $\beta \in \aut(Q)$.

\paragraph{Computing the coset of isomorphisms.} (This is essentially the same procedure as in Smith's thesis \cite{smith}; we review the procedure here in detail, in order to give a rigorous analysis of its running time.) Note that $\aut(G)$ maps into $\aut(A) \times \aut(Q)$; let $\rho$ denote this homomorphism. We apply Observation~\ref{obs:genSubProd} twice: once, to split the computation into computing generators of $\im(\rho) \leq \aut(A) \times \aut(Q)$ and $\ker(\rho)$, and then a second time to split the computation of $\im(\rho)$ into computing its projections onto $\aut(A)$ and $\aut(Q)$, respectively. We handle this latter step first.

Let $\aut(A)$ play the role of $Y$ in Observation~\ref{obs:genSubProd} 
and $\aut(Q)$ play the role of $Z$. In our algorithm, we may in fact take 
$\mathcal{Z}$ to be the entirety of the projection of $\aut(G)$ into $\aut(Q)$, since we enumerate 
over all automorphisms of $Q$. To use Observation~\ref{obs:genSubProd}, for each 
automorphism of $Q$ that admits a compatible automorphism of $A$, we need one such 
automorphism of $A$, which are easily  found as solutions to the linear algebra, as described above. 

We then also need a generating set of the automorphisms of $A$ that preserve the cohomology class $f\colon Q \times Q \to A$, that is, those elements of $\aut(A)$ that send $f$ to itself \emph{modulo $B^2(Q, A)$}. Equivalently, we want the stabilizer of the matrix $M_{\pi(f)}$ in $\aut(A) \cong \GL(k, p)$. Since $\aut(A)$ acts on each column of a 2-cochain independently, this stabilizer is the same as the pointwise stabilizer of the \emph{set} of columns of $M_{\pi(f)}$. If the columns of $M_{\pi(f)}$ span $A$, then this stabilizer is trivial, and we are done. More generally, let $B \leq A$ be the $\Z_p$-linear span of the columns of $M_{\pi(f)}$ (although we won't need it here, Lemma~\ref{lem:direct_factor} may help the reader's intuition, which says that, in this case, there is a complementary subgroup $B' \leq A$ such that $A \cong B \oplus B'$ and $G \cong (G/B') \times B'$). Choose a basis $\{e_1, \dotsc, e_k\}$ of $A$ such that a prefix of this basis, say $\{e_1, \dotsc, e_{\dim B}\}$ is a basis of $B$. In this basis, the stabilizer of $M_{\pi(f)}$ in $\GL(k,p)$ then consists of all block matrices of the form 
\[
\left( 
\begin{array}{cc}
\id_B & * \\
0 & \eta
\end{array}
\right),
\]
where $\eta \in \GL(k-\dim B, p)$ and ``$*$'' indicates any $(\dim B) \times (\dim A - \dim B)$ matrix. A generating set for this subgroup is easily written down directly, taking advantage of a standard generating set for $\GL$ (\eg, elementary matrices). Observation~\ref{obs:genSubProd} then gives us a generating set for $\im(\rho) \leq \aut(A) \times \aut(Q)$.

Finally, we apply Observation~\ref{obs:genSubProd} again, now with $\ker(\rho)$ in the role of $Y$ and $\im(\rho)$ in the role of $Z$. We have already computing a generating set $\mathcal{Z}$. For each such generator, we must compute a lift of it in $\aut(G)$, and then we also need to compute generators for $\ker(\rho)$. The lifts of $\im(\rho)$ to $\aut(G)$ are easily computed, once we have chosen our section $s\colon Q \to G$: given $(\alpha,\beta) \in \im(\rho) \leq \aut(A) \times \aut(Q)$, every element of $G$ is uniquely represented as $as(q)$ for some $a \in A, q \in Q$, and then the corresponding automorphism of $G$ acts by sending $as(q)$ to $\alpha(a) s(\beta(q))$. It is readily verified that this is an automorphism of $G$.

The kernel of $\rho$ consists of those automorphisms of $G$ that fix $Q$ pointwise, and fix $A$ pointwise. Thus, all they can do is move elements around within their respective cosets of $A$. That is, any automorphism in $\ker(\rho)$ is determined by a map $\delta\colon Q \to A$ such that the automorphism is given by sending $as(q)$ to $a\delta(q) s(q)$. Since the automorphism fixes $A$ pointwise, we must have $\delta(\id_Q) = \id_A$. To see the condition to be a homomorphism, recall that $(a s(q))(a' s(q')) = a (s(q) a' s(q)^{-1}) f_s(q,q') s(q q')$ so we must have
$
\delta(q) \delta(q') = \delta(qq').
$
These $|Q|^2$ equations are homogeneous linear equations in $|Q|\dim A$ variables over the field $\Z_p$; such linear equations can be solved in polynomial time, yielding a generating set for $\ker(\rho)$. Putting it all together, this yields a generating set for $\aut(G)$.
\end{proof}

To see the difficulty in extending this proof to the general abelian case, we first mention an alternative viewpoint on the preceding proof: The set $C^2(Q, \Z_p^k)$, viewed as $k \times |Q|^2$ matrices, is acted on by $\aut(A) \cong \GL(k,p)$ by left multiplication. Let $V \cong \Z_p^k$ be the defining representation of $\GL(k,p)$; then $C^2(Q,A)$ is isomorphic, as an $\aut(A)$-module, to $V^{\oplus |Q|^2}$. Since $V$ is irreducible, this module is semisimple, so every submodule is a direct summand, and thus every submodule is of the form $V^{\oplus \ell}$ for some $\ell$. So $H^2(Q, A)$ is isomorphic to $V^{\oplus \ell}$ as an $\aut(A)$-module, and finding the $\aut(A)$-stabilizer of a point in $V^{\oplus \ell}$ amounts to finding the stabilizer of a point in $V$ several times, which is easily done since $V$ is the defining representation of $\aut(A) \cong \GL(V)$.

The difficulty in extending this proof to the general abelian case is that $B^2(Q, A)$ need not, in general, be a direct summand of $Z^2(Q, A)$, nor of $C^2(Q, A)$; in the general case $C^2(Q,A) \cong A^{\oplus |Q|^2}$ as $\aut(A)$-modules, but this module is no longer semisimple in general. For example, $C^2(A_5, \Z_4) \cong \Z_4^{60^2}$, which contains many copies of $\Z_2$ that are submodules but not direct summands. Note that $H^2(A_5, \Z_4) \cong \Z_2$ (use the well-known fact that $H_2(A_5, \Z) \cong \Z_2$, together with the Universal Coefficient Theorem, \eg, \cite[Theorem~3.2]{hatcher}, and the fact that $H_1(A_5,\Z) = 0$, since $A_5$ is perfect). So either $Z^2(A_5, \Z_4)$ is not a direct summand of $C^2(A_5, \Z_4)$ or $B^2(A_5, \Z_4)$ is not a direct summand of $Z^2(A_5, \Z_4)$. 

Without a direct summand decomposition, the straightforward way to obtain an action of $\aut(A)$ on $H^2(Q, A)$ is to realize $H^2(Q, A)$ explicitly as a quotient module. This can be done quite easily, using the fact that $H^2(Q, \Z_{p_1^{\mu_1}} \times \dotsb \Z_{p_k^{\mu_k}}) \cong \bigoplus_{i \in [k]} H^2(Q, \Z_{p_i^{\mu_i}})$. This even yields an $\aut(A)$-invariant projection $Z^2(Q, A) \to H^2(Q, A)$ (the best we can hope for in the absence of a direct sum decomposition). But then the difficulty becomes that we do not know the structure of this quotient module well enough. If it were just the action of $\aut(A)$ on $A^{\oplus \ell}$ for some $\ell$, then we would be very close to done, as one can find the linear subspace of $\aut(A)$ that fixes a point in $A^{\oplus \ell}$ using Smith normal form, and then all that would be needed is the general abelian analogue of Theorem~\ref{thm:unit} (see Open Problem~\ref{open:unit}). However, while the group action of $\aut(A)$ on $H^2(Q, A)$ looks very similar to the action of $\aut(A)$ on $A^{\oplus \ell}$ for some $\ell$, it differs in an important way: The elements in $\aut(A)$ that would normally induce surjections $\Z_{p^{\mu}}^d \twoheadrightarrow \Z_{p^{\mu'}}^d$ with $\mu' < \mu$ induce maps $H^2(Q, \Z_{p^\mu}) \to H^2(Q, \Z_{p^{\mu'}})$ which need not be surjections. (The failure to be a surjection is measured precisely by $H^3(Q, \Z_{p^{\mu - \mu'}})$.) So in this case, the best we can do naively is to apply generic algorithms for finding stabilizers in matrix groups, but the current state of the art is not fast enough for this to run in time $\poly(|Q|,|A|)$. 

\subsection{For general extensions of elementary abelian groups (Theorem~B)} \label{sec:abelian}

\begin{theorem}\label{thm:quotient_list_general}
Let $\charfn$ be a polynomial-time-computable characteristic subgroup function. For two groups $G, H$ of order $n$, if $\charfn(G) \cong \Z_p^k$ and $\aut(G/\charfn(G))$ can be listed in time $t(n)$, then isomorphism of $G$ and $H$ can be decided\cosetComment{, and the coset of isomorphisms found,} in time $t(n) n^{O(1)}$.
\end{theorem}

As before, let us first see how this applies to groups with elementary abelian radicals. Combining Theorem~\ref{thm:quotient_list_general} with Theorem~\ref{thm:semisimple}, respectively Theorem~\ref{thm:semisimple2}, we have:

\begin{corollary}\label{cor:quotient_list_general}
Isomorphism of groups of order $n$ with elementary abelian radicals can be decided\cosetComment{, and the coset of isomorphisms found,}  in time $n^{c \log \log n + O(1)}$, for $c = 1/\log_2(60) \approx 0.169$.
\end{corollary}

\begin{corollary} \label{cor:quotient_list2_general}
Let $G$ and $H$ be groups of order $n$ with elementary abelian radicals. If $G/\rad(G)$ has $O(\log n / \log \log n)$ minimal normal subgroups, isomorphism between $G$ and $H$ can be decided\cosetComment{, and the coset of isomorphisms found,}  in $\poly(n)$ time. 
\end{corollary}

The proof of Theorem~\ref{thm:quotient_list_general} is a reduction to \modcyc, 
for which a deterministic polynomial-time algorithm over finite fields is provided 
by Chistov, Ivanyos and Karpinski \cite{CIK97}; this was recently generalized to 
finite modules (not necessarily over a field) \cite{CT15}, 
%\jnote{Made our use of CT accurate} 
which we use to decide the case of general abelian radicals, but where the 
extension by the radical is split (see Remark~\ref{rmk:CT}), and we expect to have 
further uses.
%for the general abelian case. 
Before the reduction it might be helpful to see this problem in a special case, when the extensions are split. 

%\begin{remark}
%If the algorithm for \modcyc \cite{CIK97} can be generalized to the case when the underlying module is an abelian group (rather than a vector space), then the three preceding results can be generalized to groups with arbitrary abelian radicals. See also Section~\ref{sec:future:abelian}.
%\end{remark}

\begin{proof}[Proof of Theorem~\ref{thm:quotient_list_general} for split extensions, a.k.a. {\sc Module Isomorphism}] Recall that $G_1$ and $G_2$ are extensions of $A=\Z_p^k$ by $Q$. Furthermore suppose both extensions split. Then to test isomorphism we are left with the \actcomp, that is, we extract the actions of $Q$ on $A$ in $G_i$ as $\theta_i:Q\to \aut(A)=\GL(k, p)$, and the goal is to find $(\alpha, \beta)\in\aut(A)\times \aut(Q)$ such that $\theta_1=\theta_2^{(\alpha, \beta)}$. As $\aut(Q)$ is enumerable, we fix a $\beta$ and all that remains is to test whether there exists $\alpha\in\GL(k, p)$ such that $\forall q\in Q$ $\theta_1(q)=\alpha^{-1}\theta_2(q)\alpha$. In other words, viewing $\theta_i$ as linear representations of $Q$ over the field $\F_p$, the problem is to test whether these two representations are equivalent. This can also be formulated as finding a nonsingular matrix $\alpha$ such that $\alpha\theta_1(q)=\theta_2(q)\alpha$, $\forall q\in Q$, namely \modiso. Over finite fields this problem admits deterministic polynomial-time algorithms \cite{CIK97,BL08}.

\paragraph{Computing the coset of isomorphisms.} When $A$ is elementary abelian, we can also find the full coset of isomorphisms. For this, we need a generating set of the group of units of the matrix algebra $\{\alpha \in M(k,p) : \alpha \theta_1(q) = \theta_2(q) \alpha \forall q \in Q\}$. First, compute an $\F_p$-linear basis for this algebra by solving the given equations, which are linear in the entries of $\alpha$. Next, from this linear basis for this matrix algebra, use Theorem~\ref{thm:unit} to compute a generating set of the group of units.
\end{proof}

\begin{remark} \label{rmk:CT}
In the case of split extensions, we can also decide isomorphism even when $A$ is a general abelian group, not necessarily elementary. For this, we cannot use $\GL(k,p)$, but must stick with the more general group $\aut(A)$; the rest of the above proof for deciding isomorphism goes through \emph{mutatis mutandis}, reducing now to \modiso for finite $\Z G$-modules, rather than finite(-dimensional) $\F_p G$-modules. \modiso for finite $\Z G$-modules can also be solved deterministically in polynomial time \cite{CT15}. For computing the coset of isomorphisms, we do not yet know how to compute the unit group of a finite $\Z$-algebra (rather than one over a field). See Open Question~\ref{open:unit}.
\end{remark}

Now we present the reduction for the general case. 

\begin{proof}[Proof of Theorem~\ref{thm:quotient_list_general}: reduction to {\sc Module Cyclicity}] As always, we start with the elementary abelian case. Let $G_1$ and $G_2$ be extensions of $A=\Z_p^k$ by $Q$, and $(\theta_i, f_i)$ the extension data of $\extension{A}{G_i}{Q}$. It can be verified that if $(\alpha, \beta)$ satisfies $\theta_1=\theta_2^{(\alpha, \beta)}$, then $(\alpha, \beta)$ sends $Z^2(Q, A, \theta_2)$ to $Z^2(Q, A, \theta_1)$ and sends $B^2(Q, A, \theta_2)$ to $B^2(Q, A, \theta_1)$. As $\aut(Q)$ is enumerable, the problem is to find $\alpha\in\GL(k, p)$ such that (1) $\forall q\in Q$, $\alpha\theta_1(q)=\theta_2(q)\alpha$; (2) $\alpha f_1=f_2$ as cohomology classes in $Z^2(Q, A, \theta_2)$ (that is $[\alpha f_1]=[f_2]$). 

This task can be reduced to \modcyc over finite-dimensional algebras, in almost 
the same way as the reduction from \modiso to \modcyc \cite{CIK97}. We include a 
full proof here for completeness. The basic idea is as follows: just as $\Iso(G_1, 
G_2)$ is a coset of $\aut(G_1)$, $\Hom(G_1, G_2)$ is a ``coset of''---more 
precisely, is acted on by---$\End(G_1) := \Hom(G_1, G_1)$. If one could find 
$\Hom(G_1, G_2)$, the problem would be reduced to finding an invertible element 
there. The following is the linear-algebraic version of these ideas, in which 
finding an invertible element inside a set of not-necessarily-invertible 
homomorphisms is reduced to \modcyc, for the analog of $\Hom(G_1, G_2)$ considered as 
a module over the analog of $\End(G_1)$.

Let $M(k, p)$ be the linear space of $k\times k$ matrices over $\Z_p$. Consider a linear subspace of $M(k, p)$, $V=\{\alpha\in M(k, p)\mid \forall q\in Q, \alpha\theta_1(q)=\theta_2(q)\alpha, \text{ and } \exists a\in\Z_p, [\alpha f_1]=[af_2]\}$.  Also consider $U=\{\gamma\in M(k, p)\mid \forall q\in Q, \gamma\theta_2(q)=\theta_2(q)\gamma, \text{ and } \exists a\in \Z_p, [\gamma f_2]=[af_2]\}$. It can be verified that $U$ is an associative algebra over $\Z_p$ with identity. Note that we need to allow the possible scalar $a$ in the definition of $U$ in order for $U$ to be closed under addition. This does not hurt the multiplicative condition, since $\alpha \theta_2(q) = \theta_2(q) \alpha$ is not changed if we multiply both sides by a scalar. Once the scalar is allowed in the definition of $U$, we must also allow it in the definition of $V$, in order for $V$ to be a $U$-module (see next paragraph). However, if $[\gamma f_2] = [a f_2]$, then $a^{-1}\gamma \in U$ and $a^{-1} \gamma) [f_2] = [f_2]$, so $U$ essentially consists of endomorphisms of $G$ ``up to scale.''

Then $V$ is a left $U$-module: for $\alpha\in V$, $\gamma\in U$ and $q\in Q$, $\gamma\alpha\theta_1(q)=\gamma\theta_2(q)\alpha=\theta_2(q)\gamma\alpha$. To show that $[\gamma\alpha f_1]=[af_2]$ is a little subtle, and for this we need to recall the fact that, if $\gamma\theta_2(q)=\theta_2(q)\gamma$ for every $q\in Q$, then $\gamma$ preserves $B^2(Q, A, \theta_2)$. 
That is, $\alpha f_1=af_2+g$ for some $a\in\Z_p$ and $g\in B^2(Q, A, \theta_2)$, 
and $\gamma\alpha f_1=\gamma (af_2+g)=a\gamma f_2+\gamma g=a'f_2+g'+g''$ where 
$a'\in\Z_p$, $a\gamma f_2=a' f_2+g'$ and $\gamma g=g''$. Now we claim that if $V$ 
contains invertible elements, then (1) it is cyclic, and (2) every generator is 
invertible. To show (1), let $\alpha'\in V$ be invertible, and form $\phi:U\to V$ 
by sending $\gamma \to \gamma\alpha'$. Then $\phi$ is an $U$-module isomorphism 
between $U$ and $V$, whose inverse is $V\to U$ by $\alpha \to \alpha\alpha'^{-1}$; 
$\alpha\alpha'^{-1}\in U$ again follows from that $\alpha$ and $\alpha'$ can be 
shown to send $B^2(Q, A, \theta_2)$ to $B^2(Q, A, \theta_1)$ as a consequence of 
$\alpha\theta_1(q)=\theta_2(q)\alpha$. For (2), if $\alpha''$ generates $V$, then 
$\alpha''\alpha'^{-1}$ generates $U$ as a left $U$-module, and thus 
$\alpha''\alpha'^{-1}$ is invertible, showing that $\alpha''$ is invertible. 
Finally we note that if some invertible $\alpha'\in V$ sends $[f_1]$ to $[af_2]$ 
for some $a\in \Z_p$, then $a^{-1}\alpha'\in V$ is also invertible and sends 
$[f_1]$ to $[f_2]$.

Given the above reduction, here is an algorithm for the general case: we still represent 2-cocycles by $k\times |Q|^2$ matrices over $\Z_p$. We first compute a $\Z_p$-basis of $B^2(Q, A, \theta_2)$ as in Proposition~\ref{prop:basis}. Using these 2-cocycles we can represent $V$ and $U$ as solution spaces of homogeneous linear equations. Finally we apply the polynomial-time \modcyc algorithm \cite{CIK97}, either to get that $V$ is not cyclic, and thus does not contain invertible elements, or to get a generator $\alpha'\in V$. In the latter case we conclude based on whether $\alpha'$ is invertible or not. 

%When $A$ is abelian, but not necessarily elementary, the above proof goes through with the following changes. Let $e$ be the exponent of $A$; then instead of $A$ being a $\Z_p$-vector space, it is a $\Z_{p^e}$-module. In place of $M(k,p)$ we use $\End(A) = \Hom(A, A)$. We define $V = \{\alpha \in \End(A) | \forall q \in Q, \alpha \theta_1(q) = \theta_2(q) \alpha \text{ and } \exists a \in \Z_{p^e}, [\alpha f_1] = [a f_2]\}$ and $U = \{\gamma \in \End(A) | \forall q in Q, \gamma \theta_2(q) = \theta_2(q) \gamma \text{ and } \exists a \in \Z_{p^e}, [\gamma f_2] = [a f_2]\}$. $U$ is still an associative algebra with identity (albeit, not over a field; rather it is a $\Z_{p^e}$-module with an associative multiplication $U \otimes_{\Z_{p^e}} U \to U$), and $V$ is a $U$-module. The rest of the proof proceeds as before, only now we must use the more general test for \modcyc that works for arbitrary finite modules of finite algebras \cite[Theorem~1.3]{CT15}.

\paragraph{Computing the coset of isomorphisms.} Any single invertible element of $V$ yields an isomorphism. As in Theorem~\ref{thm:quotient_list}, a generating set of automorphisms is determined by (1) its image in $\aut(A) \times \aut(Q)$, and (2) the kernel of the natural map $\aut(G) \to \aut(A) \times \aut(Q)$. In this case, (1) is easily determined by a generating set of the group of units of $U$, and (2) is a generalization of a homomorphism from $Q \to A$ called a ``derivation,'' which we describe below (see, \eg, \cite[Section~3.3]{smith}). In the elementary abelian case, the unit group of $U$ can be computed from a linear basis for $U$ by Theorem~\ref{thm:unit}. All that remains is to compute the derivations $Q \to A$.

Recall that, for a section $s\colon Q \to G$, every element of $G$ is uniquely represented as $a s(q)$ for some $a \in A, q \in Q$. Then we have that $(a s(q))(a' s(q')) = a (s(q) a' s(q)^{-1}) f_s(q,q') s(q q')$. An automorphism that fixes both $Q$ and $A$ pointwise is fully determined by a map $\delta\colon Q \to A$, as in Theorem~\ref{thm:quotient_list}, but now satisfying the more general condition:
\[
\delta(q) (s(q) \delta(q') s(q)^{1}) = \delta(qq').
\]
Such maps $\delta$ are called \emph{derivations}. As conjugation of $Q$ on $A$ is an action by automorphisms, this amounts to $|Q|^2$ homogeneous linear equations in $|Q|\dim A$ variables over the field $\Z_p$; such linear equations can be solved in polynomial time, yielding a generating set for the kernel of the map $\aut(G) \to \aut(A) \times \aut(Q)$. Putting this all together yields a generating set for $\aut(G)$.
\end{proof}

% !TEX root = main.tex
\section{When \texorpdfstring{$\aut(Q)$}{Aut(Q)} is too big} \label{sec:centrad_bb}
\label{sec:centrad} % leave in this second label. We should go through and change all of the labels to be consistent, but for now this is needed.
In this section we present polynomial-time algorithms for certain central-radical groups even when $\aut(Q)$ cannot be enumerated in polynomial time (unconditionally, simply because its size is super-polynomial). In particular, we present two fixed-parameter polynomial-time algorithms for central radical groups with $G/\rad(G)$ a direct product nonabelian simple groups, or a direct product of small perfect groups:

\begin{theorem}\label{thm:ecentradwsoc}
Isomorphism between two groups $G_1, G_2$ with central, elementary abelian radicals can be decided, and the coset of isomorphisms found, in $\poly(|G_i|)$ time if either:
\begin{enumerate}
\item\label{thm:item:simple} $G_1/\rad(G_1)$ is a direct product of simple groups%
; or
\item\label{thm:item:perfect} $G_1/\rad(G_1)$ is a direct product of perfect groups, each of order $O(1)$.
\end{enumerate}
\end{theorem}

Theorem~\ref{thm:ecentradwsoc} yields polynomial-time algorithms for the 
following concrete cases, for example: (\ref{thm:item:simple}) covers groups with $\rad(G) = Z(G) = \Z_p^k$ and $G/Z(G) \cong A_m^k$; (\ref{thm:item:perfect}) covers the case when $\rad(G) = Z(G) = \Z_p^k$ and $G/Z(G) \cong (A_5 \wr \text{PSL}_{20}(\F_7))^{k}$, where the wreath product is taken with respect to any permutation representation of $\text{PSL}_{20}(\F_7)$.

\begin{remark} \label{rmk:tradeoff}
In Theorem~\ref{thm:ecentradwsoc}(\ref{thm:item:perfect}) we can relax the $O(1)$ size bound to $O(f(n))$ to get an algorithm that runs in time $n^{O(f(n)^2 \log f(n))}$. For $f(n) = (\log \log n)^{1/2 - \varepsilon}$ for some fixed $\varepsilon > 0$, Theorem~\ref{thm:ecentradwsoc}(\ref{thm:item:perfect}) becomes $n^{o(\log \log n)}$-time (really-very-nearly-polynomial time). %Note that if \emph{all} the direct factors of $G/\rad(G)$ have order $\Omega(\log \log n)$, then there are at most $O(\log n / \log \log n)$ such factors and this case can be solved in polynomial time by Corollary~\ref{cor:quotient_list2} (although we don't know how to find the coset of isomorphisms so efficiently in this case). The only remaining ``hard'' cases---where our $n^{O(\log \log n)}$-time (Corollary~\ref{cor:quotient_list}) is the best we currently know---seem to be when either the direct factors of $G/\rad(G)$ are of intermediate size, say $(\log \log n)^{3/4}$, or when some of the direct factors of $G/\rad(G)$ are large, say $\log \log n$, but most of them are small so that there are still nearly $\log n$ factors in total.
\end{remark}

%\jnote{Added: } 
We remark that while it's possible that some of the algorithms in this paper could 
yield practically efficient implementations by using suitable sub-routines (with 
or without worst-case guarantees), 
Theorem~\ref{thm:ecentradwsoc}(\ref{thm:item:perfect}) is admittedly not 
practical: the square of the $O(1)$ bound appears in the exponent of the runtime 
of the algorithm, and the smallest centerless, indecomposable perfect group that 
isn't simple---since simple groups are handled by the much more efficient 
algorithm from the first part of the theorem---has order $960$ (the group is a 
semi-direct product of $A_5$ acting on $\F_2^4$; see \cite[Section~5.3, Group 
(4,1)]{HP}).

%For Theorem~\ref{thm:ecentradwsoc} (\ref{thm:item:perfect}), we only give the proof for the case when $A = Z(G) = \rad(G)$ is the elementary abelian $p$-group $\Z_p^k$; the general case of $Z(G)=\rad(G)$ can be obtained following the idea in \Sec{elem}.

For Theorem~\ref{thm:ecentradwsoc}, currently we can only work with elementary abelian groups; an open problem posed in \cite[Section 7.7]{BCGQ}, namely the group code equivalence problem over cyclic groups, seems to be the only current obstacle for part (\ref{thm:item:perfect}); see Observation~\ref{obs:groupcode}. For part (\ref{thm:item:simple}), we show in Lemma~\ref{lem:GKKL} that the cohomological bound \emph{can} be extended to the general abelian case; however, essentially the same obstacle as in Section~\ref{sec:coset} still needs to be overcome in this approach.

As remarked before, for Theorem~\ref{thm:ecentradwsoc} we need a more detailed understanding of central extensions in this special group 
class, including the deep Theorem~\ref{thm:GKKL} \cite{GKKL}. We also remind the reader that singly exponential algorithms for \lincode 
and \cosetint play an important role in Theorem~\ref{thm:ecentradwsoc} (\ref{thm:item:simple}).

\subsection{Preparations from cohomology} \label{sec:prep_coho}
Let $A$ be an abelian group, and $T_1, \dots, T_\ell$ be perfect, centerless, directly indecomposable groups (this includes the case in which the $T_i$ are nonabelian simple; in fact, the reader unfamiliar with perfect groups will not lose too much by just considering the nonabelian simple case instead). For an extension $\extension{A}{G}{Q}$ with $Q=\prod_{i\in[\ell]}T_i$, let $U_i$ be the inverse image of $T_i$ in $G$ under the natural projection from $G$ to $Q$. The following proposition adapted from Suzuki \cite{Suzuki2} is crucial (see Appendix~\ref{app:lem} for a proof).

\begin{proposition}[{Cf. \cite[Chapter 6, Proposition 6.5]{Suzuki2}}]\label{prop:central}
Let notations be as above. For $i, j\in[\ell]$, $i\neq j$, $[U_i, U_j]=\id$. That is, $\forall x\in U_i$, $\forall y\in U_j$, $xy=yx$.
\end{proposition}

We now consider the $U_i$ not just as subgroups of $G$, but as extensions $\extension{A}{U_i}{T_i}$. As Proposition~\ref{prop:central} shows that $[U_i, U_j]=\id$ for $i\neq j$, these extensions determine the extension $\extension{A}{G}{Q}$ as follows:

\begin{lemma}\label{lem:prod}
Given two central extensions $\extension{A}{G_j}{Q}$ ($j=1,2$) with $A = Z(G_j)$ and $Q = \prod_{i=1}^{\ell} T_i$, with the $T_i$ perfect, centerless, and indecomposable, let $U_{j,i}$ be the inverse image of $T_i$ under the natural projection $G_j \to G_j / A$. The extensions $\extension{A}{G_j}{Q}$ ($j=1, 2$) are equivalent if and only if for each $i\in[\ell]$, the extensions $\extension{A}{U_{j,i}}{T_i}$ ($j=1, 2$) are equivalent.
\end{lemma}

\begin{proof}
The only if direction is trivial. For the other direction, for $j=1, 2$ and $i\in[\ell]$, let $f_{j, i}$ be the 2-cocycle of the extension $\extension{A}{U_{j, i}}{T_i}$ induced by some section $s_{j, i}:T_i\to U_i$. By standard cohomology (see, \eg, Appendix~\ref{app:gentle}), as $U_{1,i}$ and $U_{2,i}$ are equivalent extensions for each $i$, $f_{1,i} - f_{2,i}$ is some 2-coboundary $b_i \in B^2(T_i, A)$. To show the equivalence of $\extension{A}{G_j}{\prod_{i\in[\ell]}T_i}$, we only need to exhibit two 2-cocycles $f_j$ for $\extension{A}{G_j}{Q}$ that differ by a 2-coboundary.

As $Q$ is decomposed uniquely as $\prod_i T_i$ (that is, uniquely as subsets of $Q$, not just up to isomorphism), we can identify elements in $Q$ as from $\prod_i T_i$ without ambiguity. Let $(p_1, \dots, p_\ell)$ and $(q_1, \dots, q_\ell)$ be two elements in $Q$, $p_i, q_i\in T_i$ for $i\in[\ell]$. Then define $b:Q\times Q\to A$ as
\begin{equation}\label{eq:coboundary_resp_prod}
b((p_1, \dots, p_\ell), (q_1, \dots, q_\ell))=\sum_{i\in[\ell]} b_i(p_i, q_i).
\end{equation}
Using the $\Z$-linearity of the coboundary condition, one can verify that $b$ is a 2-coboundary in $B^2(Q, A)$.

Recall that the 2-cocycle $f_{j, i}$ is induced by the section $s_{j, i}:T_i\to U_i$. We define a section $s_j:Q\to G_j$, as $s_j((p_1, \dots, p_\ell))=s_{j, 1}(p_1)\dots s_{j, \ell}(p_\ell)$. Let $f_j$ be the 2-cocycle induced by $s_j$, then---noting that for $i_1\neq i_2$, $s_{j, i_1}(p_{i_1})$ and $s_{j, i_2}(p_{i_2})$ commute by Proposition~\ref{prop:central}---it can be verified that
\begin{equation}\label{eq:cocycle_resp_prod}
f_j((p_1, \dots, p_\ell), (q_1, \dots, q_\ell))=\sum_{i\in[\ell]}f_{j, i}(p_i, q_i).
\end{equation}
Thus $f_1-f_2=b\in B^2(Q, A)$, finishing the proof.
\end{proof}

For convenience, in the following we shall call $U_{j, i}$ the \emph{restriction of $G_j$ to $T_i$} and use $G_j|_{T_i}$ to denote it. The next lemma concerns the direct product structure of the normal part; its proof is included in Appendix~\ref{app:lem} for completeness.

\begin{lemma}\label{lem:direct_factor}
Let $A' \times A'' \hookrightarrow G \twoheadrightarrow Q$ be a central extension of $A' \times A''$ by $Q$. Let $p_{A'}\colon A' \times A'' \to A'$ be the projection onto $A'$ along $A''$. If there is a 2-cocycle $f\colon Q \times Q \to A' \times A''$ such that $p_{A'} \circ f \colon Q \times Q \to A'$ is a 2-coboundary, then $G$ is isomorphic (even equivalent) to the direct product $A' \times (G / A')$.

Furthermore, $A'$ can be computed in $\poly(|G|)$ time using linear algebra over abelian groups.
\end{lemma}

Using general algorithms for decomposing direct products \cite{KN09,Wil10}, we could compute $A'$ in polynomial time without the ``furthermore.'' However, in the setting of Lemma~\ref{lem:direct_factor}, we give a much simpler algorithm to compute $A'$ using linear algebra over abelian groups.

\subsection{Warm-up result}
The full proof of Theorem~\ref{thm:ecentradwsoc} requires several ideas. To highlight these ideas separately, we start by proving a warm-up result that only needs some of these ideas. The following result is the same as Theorem~\ref{thm:ecentradwsoc}, except that instead of requiring that the central radical be elementary abelian, it requires that $|\aut(\rad(G))| \leq \poly(|G|)$. (Formally, these two conditions are incomparable, but the condition used here ``feels'' more stringent.)

\begin{proposition} \label{prop:warmup}
Isomorphism between two groups $G_1, G_2$ with central radicals can be decided, and the coset of isomorphisms found, in $\poly(|G_1|)$ time if $|\aut(\rad(G_1))| \leq \poly(|G_1|)$ and either:
\begin{enumerate}
\item\label{prop:item:simple} $G_1/\rad(G_1)$ is a direct product of simple groups%
; or
\item\label{prop:item:perfect} $G_1/\rad(G_1)$ is a direct product of perfect groups, each of order $O(1)$.
\end{enumerate}
%Isomorphism of groups $G_1, G_2$ with central radicals can be decided, and the coset of isomorphisms found, in $\poly(|G_1|)$ time when $|\aut(\rad(G_1))| \leq \poly(|G_1|)$ and $G_1/\rad(G_1)$ is a direct product of perfect groups, each of order $O(1)$.
\end{proposition}

As a few examples of when the condition on $\aut(\rad(G))$ holds, note that $\aut(R) \leq |R|^{\log |R|}$ for all groups $R$, and $|\aut(\Z_{p^e}^\ell)| \sim p^{e \ell^2}$, so the condition holds if $|\rad(G)| \leq 2^{O(\sqrt{\log |G|})}$ or if $\rad(G) = \Z_{p^e}^\ell$ with $e \ell^2 \log p \leq O(\log |G|)$.

\begin{proof}
%For simplicity of exposition, suppose that $A = \rad(G) = Z(G)$ not only satisfies $|\rad(A)| \leq \poly(|G|)$, but is also elementary abelian.

Decompose each $G_j$ ($j=1,2$) as an extension of $A$ by $Q=\prod_{i\in[\ell]}T_i$, where each $T_i$ is perfect and directly indecomposable; each $T_i$ is necessarily centerless because $A = \rad(G)$. Regardless of whether we are in case (\ref{prop:item:simple}) or (\ref{prop:item:perfect}), this decomposition is algorithmically straightforward by \cite[Proposition 2.1]{BCGQ}; although it was stated there only for case (\ref{prop:item:simple}), the same algorithm works in both cases, because the direct factors are uniquely determined not only up to isomorphism, but as subsets of $G_j$, even in the perfect case.

Now we group the $T_i$ according to their isomorphism types, identifying $Q=\prod_{i\in[r]}Q_i^{\ell_i}$, where $r$ is the number of isomorphism types among the $T_i$, each $Q_i$ is isomorphic to some $T_i$, and the $Q_i$'s are pairwise nonisomorphic. In case (\ref{prop:item:simple}), this can be done as each $T_i$ is generated by 2 elements, and in case (\ref{prop:item:perfect}) it can be done as each $T_i$ has size $O(1)$ by assumption. 

Then $\aut(Q)\cong \prod_{i\in[r]}\aut(Q_i)\wr S_{\ell_i}\cong \prod_{i\in[r]}(\aut(Q_i)^{\ell_i}\rtimes S_{\ell_i})$. A \emph{diagonal} of $\aut(Q)$ is an element in $\prod_{i\in[r]}\aut(Q_i)^{\ell_i}$. All diagonals are enumerable in polynomial time by the $n^{\text{\# generators}}$ technique: In case (\ref{prop:item:simple}), by Fact~\ref{fact:genby2}, $|\prod_{i\in[r]}\aut(Q_i)^{\ell_i}|\leq (\prod_{i\in[r]}|Q_i|^2)^{\ell_i})\leq |G_j|^2$, and in case (\ref{prop:item:perfect}), $|\prod_{i\in[r]}\aut(Q_i)^{\ell_i}|\leq (\prod_{i\in[r]}|Q_i|^{\log |Q_i|})^{\ell_i})\leq O(1)^{O(\log |G_j|)} = \poly(|G_j|)$.

By Lemma~\ref{lem:main}, $G_1$ and $G_2$ are isomorphic if and only if they are pseudo-congruent extensions of $A$ by $Q$. The extensions are pseudo-congruent if and only if there is an element of $\Aut(A) \times \Aut(Q)$ such that, after twisting by this element, the resulting extensions are equivalent. Once an element of $\aut(A)\times \aut(Q)$ is fixed, by Lemma~\ref{lem:prod}, the latter problem is reduced to determining the equivalence of $G_1|_{T_i}$ and $G_2|_{T_i}$ for each $i\in[\ell]$.

Note that the equivalence type of the extension $G_j|_{T_i}$ can be computed in polynomial time by Theorem~\ref{thm:quotient_list}, as $\Aut(T_i)$ can be listed in polynomial time: In case (\ref{prop:item:simple}) each $T_i$ is generated by $2$ elements, and in case (\ref{prop:item:perfect}) each $T_i$ has size $O(1)$.

Having decomposed $G_j$ as $\extension{A}{G_j}{\prod_{i \in [r]} Q_i^{\ell_i}}$, the algorithm then proceeds as follows.
For every $\alpha\in\aut(A)$, and every diagonal $\prod_{i\in[\ell]}\delta_i$ of $\prod_{i\in[\ell]} \aut(T_i)$, do the following. Apply $\alpha^{-1}$ and $\delta_i$ to each restricted extension $G_2|_{T_i}$. Now compute the equivalence types of $G_2|_{T_i}$. If the multiset of equivalence types coming from $G_1|_{T_i}$ is equal to the multiset of equivalence types from the $(\alpha, \prod_i\delta_i)$-twisted $G_2|_{T_i}$, then $G_1$ and $G_2$ are pseudo-congruent as extensions, and the algorithm reports ``isomorphic.'' On the other hand, if the equivalence of multisets is not detected for any $(\alpha, \prod_i\delta_i)$ then the algorithm returns ``not isomorphic.''

It is obvious that the above procedure runs in polynomial time in $|G_1|$ and $|\aut(A)|$. We remark that if $T_{i_1}\not\cong T_{i_2}$ then $G_1|_{T_{i_1}}$ and the $(\alpha, \delta_{i_2})$-twisted $G_2|_{T_{i_2}}$ cannot be equivalent. Thus the multiset of equivalence types distinguishes the isomorphism types of $T_i$'s automatically. Finally, it is enough to compare the multisets because we have full symmetric groups $S_{\ell_i}$ acting on the isomorphic factors.

To find the coset of isomorphisms, note that the preceding algorithm enumerates over every element of $\aut(A)$ and every diagonal of $\prod_{i \in [\ell]} \aut(T_i)$. For each such pair, the algorithm determines a multiset of equivalence types, so the relevant subgroup of $S_{\ell_i}$ is just a set transporter, which is easily calculated. Compute the derivations and apply Observation~\ref{obs:genSubProd} twice, as in the proof of Theorem~\ref{thm:quotient_list_general}.
\end{proof}

\subsection{Proof of Theorem~C} %~\ref{thm:ecentradwsoc}} 
\label{sec:proof_ecentradwsoc}
Unlike the warm-up result from the previous section, for Theorem~\ref{thm:ecentradwsoc} we can no longer afford to enumerate $\aut(\rad(G))$. Avoiding this enumeration leads us to two distinct, more advanced methods, each of which may have further applications.

\begin{proof}[Proof of Theorem~\ref{thm:ecentradwsoc}]
We start as in the proof of Proposition~\ref{prop:warmup}. For $G_j$, $j=1, 2$, decompose it as an extension of $A=\Z_p^k$ by $Q=\prod_{i\in[\ell]}T_i$. Classify $T_i$'s according to their isomorphism types and group them together, identifying $Q=\prod_{i\in[r]}Q_i^{\ell_i}$. Then $\aut(Q)\cong \prod_{i\in[r]}\aut(Q_i)\wr S_{\ell_i}\cong \prod_{i\in[r]}(\aut(Q_i)^{\ell_i}\rtimes S_{\ell_i})$. A \emph{diagonal} of $\aut(Q)$ is an element in $\prod_{i\in[r]}\aut(Q_i)^{\ell_i}$. All diagonals are enumerable in polynomial time by the $n^{\text{\# generators}}$ technique. By Lemma~\ref{lem:main}, $G_1 \cong G_2$ if and only if they are pseudo-congruent extensions of $A$ by $Q$. The extensions are pseudo-congruent if and only if there is an element of $\Aut(A) \times \Aut(Q)$ such that, after twisting by this element, the resulting extensions are equivalent. Once an element of $\aut(A)\times \aut(Q)$ is fixed, by Lemma~\ref{lem:prod}, the latter problem is reduced to determining the equivalence of $G_1|_{T_i}$ and $G_2|_{T_i}$ for each $i\in[\ell]$.

We will also need an additional consequence of Lemma~\ref{lem:prod}. For $G_j$, $j=1, 2$, we say a 2-cocycle $f_j:Q\times Q\to A$ respects the direct factors if there exist $f_{j, i}:T_i\times T_i\to A$, $i\in[\ell]$ such that Equation~\ref{eq:cocycle_resp_prod} holds. Let $\cocycleprod$ denote the set of 2-cocycles respecting the direct factors. The proof of Lemma~\ref{lem:prod} shows that $\cocycleprod\neq\emptyset$. Similarly we can define 2-coboundaries that respect the direct factors $\coboundaryprod$ using Equation~\ref{eq:coboundary_resp_prod}. The difference of two cohomologous 2-cocycles in $\cocycleprod$ is automatically in $\coboundaryprod$. We define $\cochainprod$ as the 2-cochains respecting the direct factors. We may view elements of $\cochainprod$ as $k \times \sum_{i \in [\ell]} |T_i|^2$ matrices, whose rows are indexed by $[k]$ and whose columns are indexed by triples $(i; p,q)$ with $p,q \in T_i$. That is, $\cochainprod = \bigoplus_{i \in [\ell]} C^2(T_i, A)$. 

\paragraph{(\ref{thm:item:simple})} Note that the equivalence type of the extension $G_j|_{T_i}$ can be computed in polynomial time by Theorem~\ref{thm:quotient_list}, as each $T_i$ is generated by $2$ elements, and thus $\Aut(T_i)$ can be listed in polynomial time. The algorithm then proceeds as follows.

As every $\beta\in\aut(Q)$ can be represented as $(\delta, \sigma)\in (\prod_i\aut(Q_i)^{\ell_i})\rtimes (\prod_i S_{\ell_i})$, it follows that $\cocycleprod$ (resp., $\coboundaryprod$) is an invariant subset in $Z^2(Q, A)$ (resp., $B^2(Q,A)$) under the actions of both $\aut(A)$ and $\aut(Q)$. Using Proposition~\ref{prop:cob}, for each $i \in [\ell]$ we get projection $\pi_i\colon C^2(T_i, A) \to W_i$ for some $\aut(A)$-invariant $W_i \leq C^2(T_i, A)$ such that $C^2(T_i, A) = B^2(T_i, A) \oplus W_i$ and $\pi_i$ commutes with the action of $\aut(A)$.

Theorem~\ref{thm:GKKL} tells us that $\dim W_i = \dim H^2(T_i, A) \leq 17.5 \dim A$, and the ``alternative viewpoint'' at the end of \Sec{coset} implies that $W_i$ is a direct sum of copies of $A$. By a judicious choice of basis, we may thus write elements of $W_i$ as $k \times 17$ matrices, in a way that is still $\aut(A)$-equivariant. (This does not imply that the 17.5 in Theorem~\ref{thm:GKKL} can be replaced by 17 in general, but rather only in the case when $Q$ acts trivially on $A$.) From now on, let $\pi_i$ denote the composition of the previous $\pi_i$, followed by this mapping onto $k \times 17$ matrices. 

For each choice of diagonal $\delta \in \prod_{i \in [\ell]} \aut(T_i)$, we may choose the complements $W_i$ so that whenever $i,j$ are such that $T_i \cong T_j$, that $W_i$ and $W_j$ are identified by $\delta$. With this choice, we may direct sum these $\pi_i$ together to get a homomorphism $\pi = \pi_\delta\colon \cochainprod \to W \leq M(k, 17 \ell)$ for some $\aut(A)$-invariant $W \leq \cochainprod$ such that $\cochainprod = \coboundaryprod \oplus W$ and $\pi$ commutes with the action of $\aut(A) \times \prod_{i} S_{\ell_i}$. %Commuting with $\sigma \in \prod_i S_{\ell_i} \leq \aut(Q)$ is ensured by choosing the same complement for $B^2(T_i, A)$ and $B^2(T_j, A)$ (in $C^2(T_i, A)$ and $C^2(T_j, A)$, respectively) whenever $T_i \cong T_j$. Here by ``the same'' we mean ``under a chosen isomorphism identifying $T_i$ with $T_j$,'' \ie, a choice of diagonal. %; below we will iterate over all such choices. 
For each diagonal $\delta$, the question then is to decide the existence of 
$(\alpha, \sigma)\in\aut(A)\times (\prod_i S_{\ell_i})$ such that 
$\pi(f_1)=\pi(f_2^{(\id, \delta, \id)})^{(\alpha, \id, \sigma)}$. 
%\jnote{Added:} 
(Note that although $(\id, \id, \sigma)$ does not commute with $(\id, \delta, 
\id)$, we do have that $(\alpha, \delta, \sigma) = (\id, \delta, \id) (\alpha, 
\id, \sigma)$, since $\sigma$ is already ``on the right'' of $\delta$, and 
$\alpha$ and $\delta$ do commute.)

Let $\tb{1}=\tb{\pi(f_1)}$. By Lemma~\ref{lem:direct_factor}, without loss of generality we may assume $\tb{1}$ is of rank $k$. Otherwise Lemma~\ref{lem:direct_factor} splits a direct factor out of the center as $A'\times G_j/A'$ (although Lemma~\ref{lem:direct_factor} is concerned with $Z^2(Q, A)$ and $B^2(Q, A)$, it is readily adapted to $\cocycleprod$ and $\coboundaryprod$). By the Remak--Krull--Schmidt theorem, we then reduce to testing isomorphism between $G_1/A'$ and $G_2/A'$, where the desired rank condition holds. Lemma~\ref{lem:direct_factor} allows us to compute such $A' \leq Z(G_j)$.

The algorithm thus proceeds as follows. For every diagonal $\delta = \prod_{j\in[\ell]}\delta_j$ of $\prod_j T_j$, compute $\pi(f_1)$ and $\pi(f_2^{(\id, \delta, \id)})$, and let $\tb{1}=\tb{\pi(f_1)}$ and $\tb{2}=\tb{\pi(f_2^{(\id, \delta, \id)})}$. Thus $\tb{j}$ is the matrix of size $k\times (17 \ell)$ corresponding to $f_j$. 

As the action of $\aut(A)$ on the $M_i$ is by left multiplication, and the action of $\prod_{i \in [\ell]} S_{\ell_i}$ is by permuting blocks of columns, we treat the $M_i$ as generator matrices of two $\Z_p$-linear codes of dimension $k$ and length $17 \ell$. Compute the coset of equivalences $\CodeEq(\tb{1}, \tb{2})\subseteq S_{17 \ell}$ using  Theorem~\ref{thm:code_eq}. On the other hand $\prod_{i}S_{\ell_i}$ induces an action on $[17 \ell]$, permuting blocks of size 17, and these are the only permutations we want. Thus we intersect $\CodeEq(\tb{1}, \tb{2})$ with $\prod_iS_{\ell_i}$. If the intersection is nonempty, the algorithm reports ``isomorphic.'' On the other hand, if for all diagonals we get empty intersection, then the algorithm returns ``not isomorphic.''

To analyze the running time, the outer loop depending on the diagonals is polynomially related to $n$. Both the applications of the \lincode algorithm (Theorem~\ref{thm:code_eq}), and the singly-exponential time algorithm for \cosetint (\cite{Luk99}), take time $c^{17 \ell}$ for some absolute constant $c$. Since $\ell \leq \log |G|$, this is $\poly(|G|)$ time.

\paragraph{(\ref{thm:item:perfect})} This case proceeds as in the previous case, with minor modifications. Now, to decide isomorphism of the centerless, indecomposable, perfect groups $T_i$ we rely on the assumption that they have size $O(1)$ (since we do not, otherwise, know good bounds on testing them for isomorphism). Also, in the absence of an analog of Theorem~\ref{thm:GKKL} for this class of groups, rather than $17 \ell$, we use the naive bound that $\dim H^2(T_i, A) \leq \dim C^2(T_i, A) = |T_i|^2 \dim A$. Let $D \leq O(1)$ be the maximum of the $|T_i|$; then the singly-exponential algorithms for \lincode and \cosetint run in time $c^{\ell D^2} \leq c^{\ell \cdot O(1)} \leq \poly(|G|)$. 

\paragraph{Computing the coset of isomorphisms.}  The algorithm is essentially the same in both cases (\ref{thm:item:simple}) and (\ref{thm:item:perfect}). To find a single isomorphism, whenever $\delta$ is a diagonal such that the coset intersection $\CodeEq(\tb{1}, \tb{2}) \cap \prod_iS_{\ell_i}$ is nonempty, we pick a single element $\sigma$ of this coset intersection. Apply $(\delta,\sigma) \in \aut(Q)$ to the columns of $\tb{2}$, and then compute using linear algebra a matrix in $\aut(A) \cong \GL(k,p)$ that makes $\tb{1}$ equal to the column-permuted $\tb{2}$. 

To find a generating set for the automorphisms of $G_1$, we proceed as follows. As in Theorems~\ref{thm:quotient_list} and \ref{thm:quotient_list_general}, we can compute the homomorphisms $Q \to A$ using linear algebra; these constitute the kernel of the map $\aut(G_1) \to \aut(A) \times \aut(Q)$. Next, we compute generators of the image of $\aut(G_1)$ in $\aut(A) \times \aut(Q)$. We apply Observation~\ref{obs:genSubProd} again, viewing $\aut(G_1)$ as an extension of $\aut(A) \times \prod_i \aut(Q_i)^{\ell_i}$ by $\prod_i S_{\ell_i}$. For each diagonal $\delta \in \prod_i \aut(Q_i)^{\ell_i}$, we compute the coset intersection $\CodeEq(\tb{1}, \tb{1}) \cap \prod_i S_{\ell_i}$. The union of these coset intersections generates the image of $\aut(G_1)$ in $\prod_i S_{\ell_i}$. For each of these generators $\sigma$, we know which diagonal $\delta$ it came from, and we can compute an $\alpha \in \aut(A)$ such that $(\alpha,\delta,\sigma) \in \aut(G_1)$, as in the preceding paragraph. This gives us lifts of all the generators of the image of $\aut(G_1)$ in $\prod_i S_{\ell_i}$. Finally, we need to compute generators of the kernel of the map $\aut(G_1) \to \prod_i S_{\ell_i}$, as a subgroup of $\aut(A) \times \prod_i \aut(Q_i)^{\ell_i}$. 

For this, we will apply Observation~\ref{obs:genSubProd} one last time. Let $\aut_0(G_1)$ denote the subgroup of $\aut(G_1)$ whose projection to $\prod_i S_{\ell_i}$ is trivial. We can determine the projection of $\aut_0(G_1)$ to $\prod_i \aut(Q_i)^{\ell_i}$ by enumerating over all diagonals, and including only those diagonals $\delta$ for which the matrix associated to $\pi(f_1)$ and the matrix associated to $\pi(f_1^{(\id,\delta,\id)})$ have the same rowspan. For each one, we lift it to $\aut_0(G_1)$ by using linear algebra to find an element of $\aut(A) \cong \GL(k,p)$ that makes the two matrices equal. Finally, we need to compute generators of the subgroup of $\aut(A)$ that fix $\pi(f_1)$. This proceeds exactly as in Theorem~\ref{thm:quotient_list}.
\end{proof}

%Even using the recent quasi-polynomial time algorithm for \cosetint \cite{Bab16}, we seem to only be able to improve Theorem~\ref{thm:ecentradwsoc}(\ref{thm:item:perfect}) to allow $|T_i|$ up to  $O((\log n)^{1/c})$, where $c > 2$ is such that the algorithm of \cite{Bab16} runs in time $2^{O(\log n)^c}$, but at the cost that the resulting algorithm takes time $n^{O(\log \log n)}$. Of course, we could already achieve that running time for a much larger class of groups by the simpler Theorem~\ref{thm:quotient_list_general}, albeit perhaps with a worse constant in the exponent. This, of course, is only a very naive application of \cite{Bab16}; more sophisticated applications may yield better results. A bottleneck here is the fact that the algorithm for \lincode from \cite{BCGQ} has not yet been improved.

\subsection{Towards the general abelian case} \label{sec:towards_genab}
In this section we give two results headed in the direction of extending Theorem~\ref{thm:ecentradwsoc} from elementary abelian central radicals to general central radicals. First, we give an alternative proof of Theorem~\ref{thm:ecentradwsoc} (\ref{thm:item:perfect}), without using Proposition~\ref{prop:cob}. Avoiding Proposition~\ref{prop:cob} is key, because it may not have an analogue for the general abelian case (see the remarks at the end of \Sec{coset}). This alternative proof also has the advantage that it makes clear how to extend to the general abelian case by using group code equivalence over cyclic $p$-groups (Observation~\ref{obs:groupcode}). Then we observe that Theorem~\ref{thm:GKKL} can be extended from modules over fields to general finite $\Z G$-modules, which we expect to be useful for proving the general abelian case of Theorem~\ref{thm:ecentradwsoc}(\ref{thm:item:simple}). 

\begin{proof}[Alternative proof of Theorem~\ref{thm:ecentradwsoc}(\ref{thm:item:perfect})]
As in \Sec{autq}, let $\tb{f_j}$ denote the matrix representation of $f_j$ with row index set $[k]$ and column index set $Q\times Q$. As $f_j$ is completely determined by the direct factors, we can focus on $\tb{f_j}$ with column indices from $\cup_{i\in[\ell]} T_i\times T_i$. Thus for $f_j\in\cocycleprod$ the size of $\tb{f_j}$ becomes $k\times (\sum_{i\in[\ell]}|T_i|^2)$. 

We will need the following analogue of Proposition~\ref{prop:span}. Let $\widetilde{\tb{f_j}}$ denote the matrix with $\sum_{i\in[\ell]}|T_i|^2$ columns, and whose first rows are just $\tb{f_j}$. The remaining rows will be the union of bases for $B^2(T_i, \Z_p)$ for each $i$ (see Proposition~\ref{prop:basis_1}). 

\begin{proposition} \label{prop:basisCode}
Suppose that $\widetilde{\tb{f_j}}$ has full rank for $j=1,2$. Fix some diagonal $\delta \in \prod_i \aut(T_i)$ and let $\widetilde{\tb{f_1}}' = \widetilde{\tb{f_1}}^{(\id, \delta, \id)}$. Then the intersection $\CodeEq(\widetilde{\tb{f_1}}', \widetilde{\tb{f_2}}) \cap \prod_i S_{\ell_i}$ is non-empty if and only if there exists $(\alpha, \sigma) \in \aut(A) \times \prod_{i} S_{\ell_i}$ such that $(\alpha, \delta, \sigma)$ is an isomorphism of $f_1$ and $f_2$ as cohomology classes.
\end{proposition}

\begin{proof}
Let $N$ be the number of rows of $\widetilde{\tb{f_j}}$, and let $b = N-k$ be the number of rows of $\widetilde{\tb{f_j}}$ that were added to $\widetilde{\tb{f_j}}$ compared to $\tb{f_j}$. 

($\Rightarrow$) Suppose $\CodeEq(\widetilde{\tb{f_1}}', \widetilde{\tb{f_2}}) \cap \prod_i S_{\ell_i}$ is non-empty. Then there is some permutation $\sigma \in \prod_i S_{\ell_i}$ and some $\Lambda \in \GL(N, p)$ such that $\Lambda \widetilde{\tb{f_1}}^{(\delta, \sigma)} = \widetilde{\tb{f_2}}$. As the last $b$ rows of $\widetilde{\tb{f_1}}$ and $\widetilde{\tb{f_2}}$ have the same rowspan (namely, $B^2_{prod}(Q,A)$) and this rowspan is preserved by all permutations of the columns (\ie, elements of $\prod_i S_{\ell_i}$), we may assume without loss of generality that $\Lambda$ has the following block form:
\[
\Lambda = \left(
\begin{array}{cc}
\alpha & \gamma \\
 0 & \eta
\end{array}
\right).
\]
In other words, to make the row spans of the bottom $b$ rows equal, it is never necessary to add any multiples of the top $k$ rows to the bottom $b$ rows, as the bottom $b$ rows already have equal row spans that are preserved by all relevant permutations, in particular, by $\sigma$.

The sub-matrix $\gamma$ contributes by adding elements of $B^2_{prod}(Q, A)$ to $\tb{f_1}$, so that $\left(\begin{array}{cc}\id & \gamma \\ 0 & \eta \end{array} \right) \tb{f_1}$ corresponds to a cocycle that is cohomologous to $f_1$. Finally, the contribution of the sub-matrix $\alpha$ is to send $f_1$ to a pseudo-congruent cocycle, since $\alpha \in \aut(A)$. Therefore we have shown that $(\alpha, \delta, \sigma)$ is an isomorphism of the cohomology classes $f_1, f_2$. 

($\Leftarrow$) Suppose that $f_1^{(\alpha, \delta, \sigma)}$ is cohomologous to $f_2$ for some $\alpha \in \aut(A)$ and $\sigma \in \prod S_{\ell_i}$. Then some matrix $\Lambda' = \left(\begin{array}{cc} \alpha & \gamma \\ 0 & \id \end{array}\right)$ will make the first $k$ rows of $\Lambda' \widetilde{\tb{f_1}}^{(\delta, \sigma)}$ equal to the first $k$ rows of $\widetilde{\tb{f_2}}$. The last $b$ rows of $\Lambda' \widetilde{\tb{f_1}}^{(\delta, \sigma)}$ and $\widetilde{\tb{f_2}}$ have the same row span, so there is some $\eta$ such that $\Lambda = \left(\begin{array}{cc} \alpha & \gamma \\ 0 & \eta \end{array}\right)$ makes the two matrices equal. (In fact, $\eta$ will be a block-permutation matrix, which permutes the blocks in the same way that $\delta$ permutes the factors.) In particular, this shows that $\sigma$ is a code equivalence, and hence that $\CodeEq(\widetilde{\tb{f_1}}', \widetilde{\tb{f_2}}) \cap \prod_i S_{\ell_i}$ is nonempty.
\end{proof}

To finish the proof of Theorem~\ref{thm:ecentradwsoc} (\ref{thm:item:perfect}), we proceed as follows. First, find a direct decomposition of $G_j$ \cite{KN09}. If any direct factor is contained in $Z(G_j)$, set this direct factor aside. Then by Lemma~\ref{lem:direct_factor}, the matrices $\tb{f_j}$ will be full rank, \emph{even after modding out by coboundaries}. The latter fact implies that the rank of $\widetilde{\tb{f_j}}$ is then the rank of $\tb{f_j}$ plus the dimension of $B^2_{prod}(Q, A)$, in other words, the matrices $\widetilde{\tb{f_j}}$ are also of full rank.

Finally, we compute the coset of equivalences $\CodeEq(\widetilde{\tb{1}}, \widetilde{\tb{2}})\subseteq S_m$ using  Theorem~\ref{thm:code_eq}. On the other hand $\prod_{i}S_{\ell_i}$ induces an action on $[m]$, which contains the permutations we want. Thus we need to intersect $\CodeEq(\tb{1}, \tb{2})$ with $\prod_iS_{\ell_i}$. By Proposition~\ref{prop:basisCode}, if this intersection is nonempty, the algorithm returns ``isomorphic.'' If the intersection is empty for all diagonals $\delta$, the algorithm returns ``non-isomorphic.''

To analyze the running time, the outer loop depending on the diagonals is polynomially related to $n$. Both the applications of the \lincode algorithm (Theorem~\ref{thm:code_eq}), and the singly-exponential time algorithm for \cosetint (\cite{Luk99}), take time $c^m \leq c^{\ell D^2}$ for some absolute constant $c$, where $D$ is the maximum size of any of the $T_i$.

\paragraph{Computing the coset of isomorphisms} As in the other proof, the key is to compute the image of $\aut(G_1)$ in $\aut(A) \times \prod_{i \in [r]} \left(\aut(Q_i)^{\ell_i} \rtimes S_{\ell_i} \right)$. %We will apply Observation~\ref{obs:genSubProd} with $H = \aut(A) \times \prod_{i \in [r]} \aut(Q_i)^{\ell_i}$ and $K = \prod_{i \in [r]} S_{\ell_i}$. 
For each diagonal $\delta \in \prod_{i \in [\ell]} \aut(T_i)$, we compute the coset intersection of Proposition~\ref{prop:basisCode}, which is a subcoset of $\prod_{i \in [r]} S_{\ell_i}$. This subcoset is given by a single element of $\prod_{i \in [r]} S_{\ell_i}$ together with a generating set for the image of $\aut(G_1)$ in $\prod_{i \in [r]} S_{\ell_i}$.  Given the diagonal $\delta$ and one of the generators $\sigma$ of this subcoset, we compute an element $\alpha \in \aut(A)$ such that $(\alpha, \delta, \sigma) \in \aut(G_1)$ following the first half of the proof of Proposition~\ref{prop:basisCode}. The only part unspecified is how to find $\Lambda$ (in the notation of that proof); given $\delta$ and $\sigma$, this $\Lambda$ is just an arbitrary element of a linear transporter, which is easily computed.

Finally, we need to compute a generating set for the $\aut_0(G_1)$ (notation as in the first proof of this theorem). We compute such a generating set by enumerating all diagonals of $\prod \aut(T_i)$ and computing generators for the linear stabilizers of the rowspan of the matrix $\widetilde{\tb{f_1}}^{(\id, \delta, \id)}$.
\end{proof}

Given a group $A$, a \emph{group code} of length $n$ over $A$ is a subgroup of $A^n$. Note that when $A = \Z_p$ is a prime field, group codes are the same as linear codes.

\begin{observation} \label{obs:groupcode}
If the coset of group code equivalences of length $n$ over cyclic $p$-groups $\Z_{p^k}$ can be found in time $2^{O(n)}$ with arithmetic operations in $\Z_{p^k}$ at unit cost, then Theorem~\ref{thm:ecentradwsoc} (\ref{thm:item:perfect}) also holds when $\rad(G) = Z(G)$ is an arbitrary abelian group. 
\end{observation}

\begin{proof}
Following the definition of $R^{(\mu)}$ in \Sec{elem}, for a given $d \times n$ matrix $M$ representing a 2-cochain over an abelian group $A = \prod_{i=1}^{k} \prod_{j=1}^{\ell_i} \Z_{p_{i}^{\mu_{i,j}}}$, let $M^{(p,\mu)}$ denote the subgroup of $\Z_{p^{\mu}}^{n}$ that is generated by the corresponding (possibly scaled/quotiented) rows $R^{(p,\mu)}$. That is, $M^{(p,\mu)}$ is the subgroup that is generated by the following: a) the rows of $M$ that correspond to $\Z_{p^{\mu}}$ factors of $A$, b) $p^{\mu - \mu'}$ times the rows of $M$ that correspond to $\Z_{p^{\mu'}}$ with $\mu' < \mu$; c) the rows of $M$ that correspond to $\Z_{p^{\mu'}}$ with $\mu' > \mu$, taken modulo $p^{\mu}$. (Rows corresponding to primes $q \neq p$ do not contribute to $M^{(p,\mu)}$.) In extending the argument of Theorem~\ref{thm:ecentradwsoc} (\ref{thm:item:perfect}) to arbitrary central radicals, it is not difficult to show that the corresponding cocycles are isomorphic if and only if there is a product-respecting automorphism of $Q$ that is simultaneously a code equivalence for each $M^{(p,\mu)}$. Thus, to solve this we compute, for each $p^{\mu}$ that appears as a factor of $\rad(G)$, the coset of code equivalences of the $M^{(p,\mu)}$. Intersecting these cosets then yields the desired permutations, and the rest of the proof proceeds as in the theorem.
\end{proof}

Finally, we observe that the cohomological bound of Theorem~\ref{thm:ecentradwsoc} extends from elementary abelian $G$-modules to general abelian $G$-modules. Although we do not yet know how to prove Theorem~\ref{thm:ecentradwsoc} in the general abelian case, because we do not have the general abelian analogue of Proposition~\ref{prop:cob}, we expect the following to be useful.

There are two possible ways to generalize the notion of ``dimension'' from 
$G$-modules over $\F_p$ to general finite $G$-modules, both of which are 
potentially useful for complexity analysis of algorithms for \GpI; our result will 
cover both simultaneously. %\jnote{Updated to log:} 
The first is the log of the cardinality of a finite 
$G$-module $M$, which we denote $\ell(M) := \log |M|$ (with the base of the logarithm unspecified, as the result will hold in any base); this generalizes dimension in that for any $\F_p G$-module $M$, 
$|M|=p^{\dim M}$. 
The second measure is the minimal number of generators, $d(M)$ 
(which equals dimension for $\F_p$-modules), which is also equal to the smallest 
$d$ such that $M$ is a quotient of the free abelian group $\Z^d$. 
%\jnote{Fixed: it used to say that this measure was equal to the number of direct 
%factors in the primary decomposition, but that's not true, e.g. $\Z_2 \times 
%\Z_3$ 
%can be generated by 1 element.} 
The measure $d(M)$ is also closely related to the number of direct factors 
appearing in the primary decomposition of $M$; more precisely, $d(M)$ is at most 
the number of direct factors in the primary decomposition of $M$, and at least the 
number of $p$-primary direct factors in the primary decomposition of $M$, for any 
prime $p$. The ratio $d(H^2(G, M)) / d(M)$ governs the number of columns that 
appear in the matrices representing 2-cohomology classes as in the first proof of 
Theorem~\ref{thm:ecentradwsoc}(1).

\begin{lemma} \label{lem:GKKL}
For any quasi-simple finite group $G$, and any finite $\Z G$-module $M$, $\ell(H^2(G, M)) \leq 17.5 \ell(M)$ and $d(H^2(G,M)) \leq 17.5 d(M)$.
\end{lemma}

If the value $17.5$ in \cite[Theorem~B]{GKKL} gets improved to a smaller constant $c$, then the $17.5$ in this lemma also immediately can be replaced by $c$. (The authors of \cite{GKKL} state that they believe that with additional work the $17.5$ could be replaced by $2$.)

\begin{proof}
%\jnote{Added sentence:} 
We give the proof for $\ell(\cdot)$; the same proof works \emph{mutatis mutandis} for $d(\cdot)$.

$M$ is a finite abelian group, so the Sylow $p$-subgroups of $M$ are characteristic subgroups, and therefore are $G$-invariant. Let $M_p$ denote the Sylow $p$-subgroup of $M$. Then $M \cong \bigoplus_{p | |M|} M_p$ is a direct sum decomposition of $\Z G$-modules. As $H^2$ is additive in its second factor, we have $H^2(G, M) \cong \bigoplus_{p | |M|} H^2(G, M_p)$. Thus, if we can show the result when $M$ is an abelian $p$-group, we are done, for then we have $\ell(H^2(G, M)) = \sum_{p | |M|} \ell(H^2(G, M_p)) \leq \sum_p 17.5 \ell(M_p) = 17.5 \ell(M)$.

So we now assume that $M$ is a finite abelian $p$-group acted on by $G$. In any abelian $p$-group, the $p$-th powers form a characteristic subgroup, which we denote $pM$ since we are using additive notation. So $pM$ is a $G$-submodule, albeit not in general a direct summand. Since $H^2$ is a left-exact covariant functor in its second argument, from the exact sequence $0 \to pM \hookrightarrow M \twoheadrightarrow M/pM \to 0$, we get a left exact sequence
\[
0 \to H^2(G, pM) \hookrightarrow H^2(G, M) \to H^2(G, M/pM),
\]
in which the final map need not be a surjection. However, even if it is not a surjection, by bounding the measures of the first and last term of this sequence, we still get a bound on the measure of the middle term $H^2(G, M)$.

Now, let $e$ be the exponent of $M$ (that is, $p^e M = 0$ but $p^{e-1} M \neq 0$). We proceed by induction on $e$. For $e = 1$, $M$ is in fact an $\F_p G$-module, and we have that $\ell(H^2(G, M)) \leq 17.5 \ell(M)$ by \cite[Theorem~B]{GKKL}.

For $e > 1$, we have that $\ell(H^2(G, M/pM)) \leq 17.5\ell(M/pM)$ by \cite[Theorem~B]{GKKL}, since $M/pM$ is an $\F_p G$-module. Since the exponent of $pM$ is strictly less than the exponent of $M$, we have, by induction, that $\ell(H^2(G, pM)) \leq 17.5 \ell(pM)$. Thus we have that $\ell(H^2(G, M)) \leq \ell(H^2(G, pM)) + \ell(H^2(G, M/pM)) \leq 17.5 (\ell(pM) + \ell(M/pM)) = 17.5 \ell(M)$.
\end{proof}

\section{Future directions} \label{sec:future}
%In this paper we made significant progress on group isomorphism for groups with central radicals, extending the results of \cite{BCQ} and beginning to resolve an open problem from \cite{BCGQ}. We achieved an $n^{O(\log\log n)}$ algorithm for this class of groups, and polynomial-time algorithms for several prominent subclasses. The difficult cases seem to be when the radical $\rad(G)$ and the semisimple quotient $G/\rad(G)$ are roughly of the same size---say both are of order $\sqrt{n}$---and $G/\rad(G)$ is sufficiently complicated (without this last condition, we handle such groups in Theorem~\ref{thm:ecentradwsoc}; see Remark~\ref{rmk:tradeoff}). Although the general case of central radicals remains open, we propose three directions for extending our work which we believe may now be within reach.

In this paper we developed an $n^{O(\log\log n)}$-time algorithm to test isomorphism of groups with central radicals, extending the results of \cite{BCQ} and beginning to resolve an open problem from \cite{BCGQ}. We also developed an $n^{O(\log\log n)}$-time algorithm for groups with elementary abelian radical (not necessarily central), and polynomial-time algorithms for several prominent subclasses of central radical groups. The ``difficult'' cases---those where we do not yet know how to improve beyond $n^{O(\log \log n)}$---seem to be when the radical $\rad(G)$ and the semisimple quotient $G/\rad(G)$ are roughly of the same size---say both are of order $\sqrt{n}$---and $G/\rad(G)$ is complicated (without this last condition, we handle such groups in Theorem~\ref{thm:ecentradwsoc}; see Remark~\ref{rmk:tradeoff}). Although a polynomial-time algorithm for the general case of central radicals remains open, we propose several directions for extending our work which we believe may now be within reach.

\subsection{{Abelian} radical} \label{sec:future:abelian}
A nearby next step is to extend our results to groups with general abelian radicals (not necessarily central nor elementary abelian):

\begin{openproblem} \label{prob:abelrad}
Extend Theorem~\ref{thm:quotient_list_general} to groups whose solvable radicals $\rad(G)$ are general abelian. Ultimately, decide isomorphism of groups with abelian radicals in polynomial time.
\end{openproblem}

We note that even with the recent algorithm for \modcyc for arbitrary finite modules (not necessarily over a field) \cite{CT15}, it is not immediately clear how to generalize the proof of Theorem~\ref{thm:quotient_list_general} to the general abelian case. The issue is in the definition of the algebra $U$: in order to be closed under addition, we must allow the cocycles to get scaled arbitrarily. However, if an element of $\End(A)$ scales a cocycle by a non-unit in $\Z_{p^e}$, it is unclear how to proceed (since then we cannot argue that by rescaling the endomorphism by the inverse of the scalar, we get an automorphism, since the scalar has no inverse). 

In both Theorems~\ref{thm:quotient_list} and \ref{thm:quotient_list_general}, we are currently only able to compute the coset of isomorphisms in the elementary abelian case. This could be generalized to the abelian case of Theorem~\ref{thm:quotient_list} by resolving the following:

\begin{openproblem} \label{open:unit}
Extend Theorem~\ref{thm:unit} (\cite{BrooksbankOBrien, RonyaiUnit}) from matrix algebras over a field to finite $\Z$-algebras. That is, given a $\Z$-linear spanning set of a finite algebra over $\Z$, compute a generating set of the group of units in polynomial time.
\end{openproblem}

\subsection{The Babai--Beals filtration} \label{sec:babaibeals}
The Babai--Beals filtration was defined and used in the context of algorithms for matrix groups \cite{BB99,BBS09}---where the groups are given by a generating set of matrices, and the goal is algorithms which run in time polynomial in the input size, which can be polylogarithmic in $|G|$. In the context of \GpI, it has also been used successfully in the polynomial-time algorithm for semisimple groups \cite{BCGQ,BCQ}.

The Babai--Beals filtration is the following chain of characteristic subgroups:
\begin{equation}\label{eq:bb_chain}
1\leq \rad(G) \leq \soc^*(G) \leq \pker(G) \leq G,
\end{equation}
where $\rad(G)$ is the solvable radical of $G$ and $\soc^*(G)$ is the subgroup such that $\soc^*(G)/\rad(G)=\soc(G/\rad(G))$. Note that the socle of the semisimple group $G/\rad(G)$ is a direct product of non-abelian simple groups. $G$ then acts on this direct product by, amongst other things, permuting the factors. The final subgroup in the Babai--Beals filtration, $\pker(G)$, consists of those $g \in G$ which do not permute the direct factors of $\soc^*(G)/\rad(G)$.

In Theorem~\ref{thm:ecentradwsoc} we make progress on the case of groups $G$ with central radical which further satisfy $G = \soc^*(G)$. It is then natural to consider groups with the next step of the Babai--Beals filtration, $G = \pker(G)$. As a polynomial-time algorithm for isomorphism of \emph{semisimple} groups $G$ satisfying $G = \pker(G)$ \cite{BCGQ} was significantly simpler than the polynomial-time algorithm for general semisimple groups \cite{BCQ}, we have hope that the following is achievable:

\begin{openproblem}
Extend Theorem~\ref{thm:ecentradwsoc} to groups with central radical which satisfy $G = \pker(G)$.
\end{openproblem}

\subsection{The Cannon--Holt strategy} \label{sec:cannonHolt}
Cannon and Holt \cite{CH03} suggest the following strategy for computing $\aut(G)$ 
for a finite group $G$, as well as for isomorphism testing 
%\jnote{Moved into main text:} 
(see also \cite{howden}, which contains improvements and new ideas related to this 
strategy). 
They consider the following chain of characteristic subgroups:
\begin{equation}\label{eq:ch_chain}
1=N_r\unlhd N_{r-1}\unlhd \dots \unlhd N_1=\rad(G)\unlhd G,
\end{equation}
where the $N_i$ refine the derived series of $\rad(G)$ and each $N_i / N_{i+1}$ is elementary abelian. More specifically, let $R^{(i)}$ be the terms of the derived series of $\rad(G)$; so $R^{(0)} = \rad(G)$, $R^{(1)} = [\rad(G), \rad(G)]$, and so on. Then the quotients $R^{(i)} / R^{(i+1)}$ are abelian. From an abelian group we get a characteristic series as follows: first break the abelian group into abelian $p$-groups for each $p$, taking the primes in ascending order. Within an abelian $p$-group, we take the subgroup of $p$-th powers as the largest characteristic subgroup in the series, and then recurse. These groups are then lifted from the abelian quotients $R^{(i)} / R^{(i+1)}$ back to subgroups of $\rad(G)$ in the natural way.
The algorithm proceeds by first computing $\aut(G/N_1) = \aut(G/\rad(G))$, and then iteratively computing $\aut(G/N_{i+1})$ from $\aut(G/N_i)$.

This chain is convenient for describing known results in the Cayley table model: the case when $\rad(G)=1$ (equivalently $r=1$) corresponds to the semisimple case, which can be solved in polynomial time \cite{BCQ}. When $G = \rad(G)$ and $r=2$, the case of $|N_2|$ and $|N_1 / N_2|$ being coprime can be solved in polynomial time \cite{BQ}. When $|N_2|$ and $|N_1 / N_2|$ are not coprime, this includes the notorious case of $p$-groups of class 2. Finally, the present work considers a special case of $r=2$, namely when $\rad(G)=Z(G)$.

In light of \cite{BCQ}, in the Cayley table model the second step in the Cannon--Holt strategy---to compute $\aut(G/N_2)$ from $\aut(G/\rad(G))$---is equivalent to the special case of Problem~\ref{prob:abelrad} in which $\rad(G)$ is elementary abelian, which we have solved in $n^{O(\log \log n)}$ time in Theorem~\ref{thm:quotient_list_general}. 

However, even before getting polynomial-time algorithms for groups with general abelian radicals (Problem~\ref{prob:abelrad}), it may be possible to give a reduction from the third step of the Cannon--Holt strategy to listing isomorphisms of two-step solvable groups. This is headed in the direction of a formal reduction from general group isomorphism to the solvable case. In a related vein, in Proposition~\ref{prop:action_triv} we showed how isomorphism of groups whose outer action on $\rad(G)$ is trivial reduces to isomorphism of central-radical groups and isomorphism of solvable groups.

\begin{openproblem}
Extend Theorems~\ref{thm:quotient_list} and \ref{thm:ecentradwsoc} to groups whose radicals are two-step solvable, allowing access to an oracle for listing $\Aut(\rad(G))$.
\end{openproblem}

%\subsection{Perfect groups}
%Recall that a group is perfect if $G = [G, G]$, where $[G, G]$ is the commutator subgroup, by definition generated by elements of the form $g h g^{-1} h^{-1}$. By considering the derived series of any finite group, $G^{(0)} = G, G^{(1)} = [G, G], G^{(2)} = [G^{(1)}, G^{(1)}], \dotsc$, one eventually arrives at a group $G^{(k)} = G^{(k+1)}$, which is necessarily perfect. Denote this group $G^{(\infty)}$. The maps $G \mapsto G^{(i)}$ and $G \mapsto G^{(\infty)}$ are characteristic subgroup functions, so our general Main Lemma~\ref{lem:main_general} applies. In particular, every group is an extension of a perfect group by a solvable group, namely $G^{(\infty)} \hookrightarrow G \twoheadrightarrow G/G^{(\infty)}$. Furthermore, if $G^{(\infty)}$ happens to be centerless, then one may take advantage of Theorem~\ref{thm:centerless}, the same result that was implicitly used for \GpI for semisimple groups \cite{BCQ, BCGQ}. Together with the suggestion in Remark~\ref{rmk:perfect} of how to use the methods of this paper to extend isomorphism algorithms for centerless perfect groups to isomorphism algorithms for general perfect groups, we arrive at several questions of interest regarding perfect groups:

% !TEX root = main.tex

\section*{Acknowledgments}
The authors thank Nikki Pfarr for producing Figure~\ref{fig:groups}, and Vipul 
Naik for Example~\ref{ex:isonotpc}, as well as several useful discussions early 
on. We thank James Wilson and several anonymous referees for numerous helpful 
comments and suggestions. We also thank I. Cioc\v{a}nea-Teodorescu for sharing and 
discussing the preprint \cite{CT15}. J.A.G. was supported by A. Borodin's NSERC 
Grant \# 482671. Y. Q. was supported by the Australian Research Council 
DECRA DE150100720. %\jnote{Added:} 
Both authors were supported by NSF grant DMS-1620484 during the preparation of the 
manuscript.

\nocite{GZ}
\bibliographystyle{alpha}
\bibliography{ref}

\appendix
\addtocontents{toc}{\setcounter{tocdepth}{1}}
%\addtocontents{toc}{\protect\setcounter{tocdepth}{1}}

% !TEX root = main.tex
\section{A gentle introduction to group extensions and cohomology} 
\label{app:gentle}

In this appendix we present a gentle introduction to group extensions and 
cohomology from the viewpoint of isomorphism testing, for readers encountering
such notions for the first time. While this content exists in many textbooks, the typical textbook treatments often require reading several chapters first; our goal here is to cover the needed material in a self-contained, leisurely and expository manner.

\subsection{An intuitive starting point}

Suppose we have a group $G$ with a normal subgroup $N$, whose quotient $G/N$ we call $Q$. We write this as $\extension{N}{G}{Q}$, and refer to $G$ as an \emph{extension} of $N$ by $Q$ (some authors use the opposite convention, and call this as an extension of $Q$ by $N$). Given $N$ and $Q$, what additional data determines $G$ up to isomorphism? 

If $G$ is the direct product $N \times Q$, then we essentially need no additional information. 

A slightly more complicated kind of extension is captured by the semi-direct 
product: In a direct product, $G$ contains subgroups $N$ and $\overline{Q} \cong 
Q$ such that $G = N\overline{Q} = \{nq : n \in N, q \in \overline{Q}\}$ and $N 
\cap \overline{Q}=1$, and both $N$ and $\overline{Q}$ are normal. As a 
consequence, every element of $n$ commutes with every element of $\overline{Q}$. A 
semi-direct product is defined similarly, except that $\overline{Q}$ need no 
longer be normal, and consequently the elements of $\overline{Q}$ need no longer 
commute with the elements of $N$. However, as $N$ is normal, given any $q \in 
\overline{Q}$, $q N q^{-1} = N$, so the map $\varphi\colon Q \to \Aut(N)$ given by 
sending $q$ to the map $(n \mapsto q n q^{-1})$ gives an action of $\overline{Q}$ 
on $N$. We denote this semi-direct product by $N \rtimes_\varphi Q$. (When the 
action is trivial, we recover the direct product.) So to specify $G$ from $N$ and 
$Q$, we also need to specify something about this action. 
%\jnote{Removed the comment about $N$ non-abelian, since at this point for a 
%``fresh'' reader I think it will just be confusing.}
%When $N$ is abelian, we 
%will need to specify exactly this action. When $N$ is not abelian, for extensions 
%that aren't semi-direct products, we will only specify this action up to 
%conjugation by elements of $N$, that is $Q \to \Out(N) = \Aut(N) / \Inn(N)$, 
%rather than $Q \to \Aut(N)$. Since here our goal is to focus on cohomology, not 
%actions, we leave the details of this to Section~\ref{sec:framework}, which could 
%be read immediately following this appendix.

However, consider the two groups $\Z_2 \times \Z_2$ and $\Z_4$, both as extensions of $\Z_2$ by $\Z_2$. (For $\Z_4$, written additively, we have $\{0,2\} \cong \Z_2 \unlhd \Z_4$, with the quotient $\Z_4 / \Z_2$ again being a copy of $\Z_2$. The quotient map is the parity map.) However, these two extensions are not distinguished by any action of the quotient on the normal subgroup: In $\Z_2 \times \Z_2$ this action is trivial, as always happens for direct products; in $\Z_4$, since $\Z_4$ is abelian, conjugation by any element of $\Z_4$ is trivial, so the action of $Q \cong \Z_2$ on $N \cong \Z_2$ is also trivial. If we think of the group elements in both cases as specified by length-2 binary strings, the difference between the groups is that when adding in $\Z_4$, one must perform \emph{carries}, whereas in $\Z_2 \times \Z_2$ one does not. Carrying is a simple example of nontrivial cohomology.

How can we capture this difference formally? Note that in the case of $\Z_4$, there is no choice of coset representatives of $\Z_2$ such that the coset representatives themselves form a subgroup; that is, there is no subgroup $\overline{Q}$ as in the discussion of semi-direct product above. It turns out that measuring the \emph{failure} of the existence of such a subgroup is precisely what we need to capture the difference between $\Z_4$ and $\Z_2 \times \Z_2$; as we will see, in this case, this ``measure of failure'' will turn out to be described exactly by carrying.

Let $N = \Z_2$, $G = \Z_4$, and $Q = G/N = \Z_2$, and let $\pi\colon G \to Q$ 
denote the natural quotient map. Let $s\colon Q \to G$ denote any choice of coset 
representatives; that is, $\pi(s(q)) = q$ for all $q \in Q$. If there \emph{were} 
a subgroup $\overline{Q}$ as above, then there would be some choice of $s$ so that 
$s$ is a homomorphism (and, in fact, it is not hard to see that such an $s$ would 
be an isomorphism $Q \stackrel{\cong}{\to} \overline{Q}$). So we will measure the 
failure of $s$ to be a homomorphism, and then consider this measure over all 
possible $s$.

For $s$ to be a homomorphism is to say precisely that $s(p)s(q) = s(pq)$ for all $p,q \in Q$, or equivalently that $s(p)s(q) s(pq)^{-1} = 1$ for all $p,q \in Q$. We will use this expression to measure the failure of $s$ to be a homomorphism, and say that the failure of $s$ to be a homomorphism is exactly captured by the function $f_s(p,q) \defeq s(p)s(q) s(pq)^{-1}$. This function is identically $1$ if and only if $s$ is a homomorphism. Any $f_s$ of this form is called a \emph{2-cocycle}, and all such 2-cocycles arise in this fashion. We then define:

\begin{definition}[Provisional definition of a cohomology class]
The 2-cohomology class of an extension $\extension{N}{G}{Q}$ is the collection $[f_s] \defeq \{ f_s | s\colon Q \to G \text{ is a choice of coset representatives}\}$. 
\end{definition}

Before understanding what this collection is, let us return to our example.

\begin{example}
The map $\pi \colon \Z_4 \to \Z_2$ maps $\{0,2\} \mapsto 0$ and $\{1,3\} \mapsto 1$. So we have four choices for $s\colon \Z_2 \to \Z_4$: 
\[
s_0 \begin{cases}
0 & \mapsto 0 \\
1 & \mapsto 1
\end{cases}
\qquad
s_1 \begin{cases}
0 & \mapsto 0 \\
1 & \mapsto 3
\end{cases}
\qquad
s_2 \begin{cases}
0 & \mapsto 2 \\
1 & \mapsto 1
\end{cases}
\qquad
s_3 \begin{cases}
0 & \mapsto 2 \\
1 & \mapsto 3
\end{cases}
\]
We calculate $f_{s_0}$ for exposition, and leave the calculation of the remaining $f_{s_i}$ to the reader. Since we are writing $\Z_4$ and $\Z_2$ additively, our expression for $f_s$ becomes $s(p) + s(q) - s(pq)$.
\[
\begin{array}{lclclcr}
f_{s_0}(0,0) & = & s_0(0) + s_0(0) - s_0(0 + 0) & = & s_0(0) & = & 0 \\
f_{s_0}(0,1) & = & s_0(0) + s_0(1) - s_0(0 + 1) & = & s_0(0) & = & 0 \\
f_{s_0}(1,0) & = & s_0(1) + s_0(0) - s_0(1 + 0) & = & s_0(0) & = & 0 \\
f_{s_0}(1,1) & = & s_0(1) + s_0(1) - s_0(1 + 1) & = & 1 + 1 - 0 & = & 2 \\
\end{array}
\]
There are two things to note. First, every $f_{s_0}(p,q)$ actually landed in the normal subgroup $N = \{0,2\}$; as we will see in a moment, this was no accident. Second, $f_{s_0}$ is exactly the carry function: If we write the values in $\Z_4$ in binary, then $f_{s_0}(p,q)$ was 00 except when $p=q=1$, when it was 10. This should help give some intuition for what a cohomology class ``really is.'' It is easily verified that none of the $f_{s_i}$ are trivial (with image $0$), so this extension is not a semi-direct product, as we already knew.
\end{example}

Note that the cohomology class depends not only on $G$, but also on the choice of $N$ (see Example~\ref{ex:isonotpc}), \emph{and} on the choice of map $\pi\colon G \to Q$ (see Example~\ref{ex:isonotcong}). We will give examples of both of these in the next section, after we've developed some more structure on cohomology classes, which will make the examples clearer. (At first, the reader may find the dependence on $\pi$ surprising, since given $N$, we think of $Q$ as being the same as $G/N$, for which there is the natural map $g \mapsto gN$. However, note that for any given map $\pi\colon G \to Q$ with kernel $N$, and any automorphism $\beta \in \Aut(Q)$, the map $\beta \circ \pi \colon G \to Q$ is also a homomorphism $G \to Q$ with kernel $N$, and it is this choice that is relevant.)

\subsection{The group structure on cohomology}
Defining the cohomology class as $\{f_s | s \colon Q \to G\}$ is a perfectly fine definition, but what is this collection? It turns out that, when $N$ is abelian, this collection is a coset of a subgroup of some group, as we now explain. For the remainder of this appendix, $N$ will be abelian, so we will denote $N$ by $A$. In $A$ we will use additive notation; although $A \leq G$ and $G$ need not abelian, this will turn out not to cause much confusion, and using this notation helps simplify things greatly.

Let $C^2(Q, A)$ denote the group of all functions (not necessarily homomorphisms) $Q \times Q \to A$; $C^2(Q, A)$ is an abelian group under pointwise addition: $(f+g)(p,q) \defeq f(p,q) + g(p,q)$. Elements of $C^2(Q,A)$ are called ``2-cochains of $Q$ with coefficients in $A$.'' First, we show that any $f_s$ as above is in fact a 2-cochain. To show this, we need merely show that $f_s(p,q) \in A$ for all $p,q \in Q$. Recalling that $A = \ker(\pi)$, we apply $\pi\colon G \to Q$ to $f_s(p,q)$: 
\[
\pi(f_s(p,q)) = \pi(s(p) s(q) s(pq)^{-1}) = \pi(s(p)) \pi(s(q)) \pi(s(pq)^{-1}) = p q (pq)^{-1} = \id_Q,
\]
so $f_s(p,q)$ is indeed in $\ker(\pi) = A$. So $f_s$ is an element of the abelian group $C^2(Q, A)$.

However, not every 2-cochain arises in this manner; it turns out that those that do are a subgroup of $C^2(Q, A)$. To determine this subgroup, we need to determine the constraint that all $f_s$ as above satisfy. One constraint comes from the fact that multiplication in $G$ is associative; it will turn out that this one constraint suffices (that is, any 2-cochain satisfying this associativity-like constraint will be of the form $f_s$ for some choice of coset representatives $s\colon Q \to G$).

To derive this associativity-like constraint, let us write the elements of $G$ in terms of those of $A$ and $Q$. Let $s\colon Q \to G$ be a choice of coset representatives; then every element of $G$ can be written uniquely as $a s(q)$ for some $a \in A, q \in Q$. Let's determine what the product in $G$ looks like when we represent elements of $G$ in this form:
\begin{eqnarray*}
(a s(p)) (b s(q)) & = & a s(p) b (s(p)^{-1} s(p)) s(q) \\
 & = & a (s(p) b s(p)^{-1}) s(p) s(q) \\
 & = & a \theta_p(b) s(p) s(q) \\
 & = & a \theta_p(b) s(p) s(q) s(pq)^{-1} s(pq)\\
 & = & a \theta_p(b) f_s(p,q) s(pq) \\
 \end{eqnarray*}
In other words, if we use the notation $(a,q)_s$ to mean $a s(q)$, then we have
\begin{equation} \label{app:eq:mult}
(a, p)_s \cdot (b,q)_s = (a + \theta_p(b) + f_s(p,q), pq)_s.
\end{equation}

\begin{exercise*}
Using Equation (\ref{app:eq:mult}), expand out the two expressions $((a,p)_s \cdot (b,q)_s) \cdot (c,r)_s$ and $(a,p)_s \cdot ( (b,q)_s \cdot (c,r)_s)$ (note that the order of parentheses matters for the syntactic form of the final expressions; don't just assume associativity!).
\end{exercise*}

Since multiplication in $G$ is associative, the two expressions in the preceding exercise must be equal. Expanding them out and setting them equal, we find that $f_s$ must satisfy the condition:
\[
f_s(p,q) + f_s(pq, r) = \theta_p(f_s(q,r)) + f_s(p,qr) \qquad \text{(the 2-cocycle identity)}.
\]
Reversing the above reasoning, we find that given an action $\theta \colon Q \to \Aut(A)$, any 2-cochain $f\colon Q \times Q \to A$ which satisfies the 2-cocycle identity arises as $f_s$ for some extension $\extension{A}{G}{Q}$ and some choice of coset representatives $s\colon Q \to G$. Any 2-cochain $f$ satisfying the 2-cocycle identity is called a \emph{2-cocycle} (with respect to the action $\theta$). Since the 2-cocycle identity is easily seen to be $\Z$-linear---if $f$ and $g$ are both 2-cocycles, then so is their pointwise sum $f + g$---the 2-cocycles form a subgroup of the 2-cochains, which we denote $Z^2(Q, A, \theta)$ (unlike 2-cochains, the 2-cocycle identity depends on the action $\theta$). Thus each cohomology class (with respect to $\theta$) is a subset of $Z^2(Q, A, \theta)$

Finally, let us determine what kind of subset a cohomology class is. Towards this end, suppose $s,t\colon Q \to G$ are two choices of coset representatives of $Q$ in 
$G$. What is the difference between $f_s$ and $f_t$? As $s(q)$ and $t(q)$ lie in 
the same coset of $A$, there is a function $u\colon Q \to A$ such that $s(q) = 
u(q)t(q)$ for all $q \in Q$. Then $f_s(p,q) = f_{t}(p,q) + \left(u(p) + 
\theta_p(u(q)) - u(pq)\right)$, where $\theta_p \colon A \to A$ is the automorphism given by conjugation by $p \in Q$. A \emph{2-coboundary} is a function of the form 
$b_u(p,q) \defeq u(p) + \theta_p(u(q)) - u(pq)$ for any function $u\colon Q \to 
A$. Hence, if two 2-cocycles come from the same extension, they differ by a 
2-coboundary.

Now, for a 2-coboundary of the form $f_s - f_t$, it is clear that it lies in $Z^2(Q, A, \theta)$, since both $f_s$ and $f_t$ are elements of this group. More generally, we have:

\begin{exercise*}
Show that any 2-coboundary satisfies the 2-cocycle identity. Show that if $f_s$ is a 2-cocycle corresponding to the extension $\extension{A}{G}{Q}$, and $b_u$ is any 2-coboundary, then $f_s + b_u$ is another 2-cocycle corresponding to the same extension.
\end{exercise*}

Thus, the 2-coboundaries are elements in $Z^2(Q, A, \theta)$. Finally, note that for any two functions $u,v\colon Q \to A$, we have $b_{u+v} = b_u + b_v$, so the 2-coboundaries in fact form a sub\emph{group} of $Z^2(Q, A, \theta)$, which we denote $B^2(Q, A, \theta)$. Thus, we have finally arrived at:

\begin{definition}
A 2-cohomology class is an element of the abelian group $Z^2(Q, A, \theta) / B^2(Q, A, \theta)$. We denote this quotient group $H^2(Q, A, \theta)$, called the second cohomology group of $Q$ with coefficients in $A$, relative to the action $\theta$.
\end{definition}

We thus arrive at one of the central notions in this paper:

\begin{definition}\label{def:ext_data}
For $A$ an abelian group and $Q$ any group, a pair $(\theta, f)$ of an action 
$\theta\colon Q \to \aut(A)$ and a 2-cocycle $f \colon Q \times Q \to A$, $f \in 
Z^2(Q, A, \theta)$ is \emph{extension data}. Two extension data for the pair $(Q, 
A)$ are \emph{equivalent} if they have the exact same action and if the two 
2-cocycles are cohomologous (differ by a coboundary).
\end{definition}

Given an extension $A \hookrightarrow G \twoheadrightarrow Q$, the extension data 
for this particular extension are the action $\theta$ as defined above, and any 
2-cocycle $f_s$ for any section $s\colon Q \to G$.
Note that extension data are non-unique, as we may choose any representative of 
the corresponding 2-cohomology class. 

Two important special cases of extension data $(\theta, f)$ are as follows.
\begin{description}
\item[$f$ is trivial (as 2-cohomology class).] This implies that there exists $P\leq G$ such that $AP=G$ and $P\cap A=\id$, \ie, that $G$ is the semi-direct product $A \rtimes P$. Such $P$ is called a \emph{complement} of $A$ in $G$, and the extension is called a \emph{split} extension. In this case, an isomorphism test need only focus on one of the two aspects of \GpI: \edpc simplifies to \actcomp.
\item[$\theta$ is trivial.] This implies that $A\leq Z(G)$, and the extension is called central. In this case, an isomorphism test need only focus on the other aspect of \GpI: \edpc simplifies to \cohiso.
\end{description}

\begin{remark}
It is not difficult to test whether an input satisfies one of the above 
conditions: it is trivial to test whether an extension is central. We leave it as 
an exercise for the interested reader to devise an algorithm to test whether an 
extension is split when the normal subgroup is abelian. See \cite[Section 
7.6.2]{HEO05} for a practical algorithm for the latter problem.
\end{remark}

\subsection{Equivalence of extensions and extension data}
Now that we see there is a group structure on $H^2(Q, A, \theta)$, we return to give examples of the dependence on $A$ and $\pi\colon G \to Q$. For this, we introduce a standard, slightly different viewpoint on the notion of ``equivalence.''

\begin{definition} \label{def:ext_equiv}
Two extensions of $A$ by $Q$, $G_1$ and $G_2$, are \emph{equivalent} if there 
exists an isomorphism $\gamma:G_1\to G_2$ such that, in the following diagram 
%\jnote{Removed footnote and said explicitly what we want from this diagram, 
%instead of just saying ``commutes''} 
% the following diagram commutes:\footnote{Such a diagram \emph{commutes} if for 
%any two directed paths in the diagram from one group to another, the 
%corresponding 
%compositions are equal as maps.}
\begin{equation*}
\xymatrix@C=50pt{
A \ar@{^{(}->}[r]^{\iota_1} \ar@{=}[d] & G_1 \ar@{->>}[r]^{\pi_1} \ar@{->}[d]^{\gamma}_{\cong} & Q  \ar@{=}[d]\\
A \ar@{^{(}->}[r]^{\iota_2} & G_2 \ar@{->>}[r]^{\pi_2} & Q
}
\end{equation*}
we have $\gamma \iota_1 = \iota_2$ and $\pi_1 = \pi_2 \gamma$, in which case we say that the diagram \emph{commutes}. 
The vertical double-lines are stretched out, rotated equality signs; that is, they denote the identity map.
\end{definition}

%\jnote{Updated:} 
It is classical, going back to O. H\"{o}lder and O. Schreier, that two extensions 
have equivalent extension data (Definition~\ref{def:ext_data}) if and only if the 
extensions are equivalent according to Definition~\ref{def:ext_equiv}. The 
explicit connection with cohomology was developed by Eilenberg and Mac Lane 
\cite{em2}.

\begin{theorem}[{See \cite[Chapter~11]{Rob} and Eilenberg--Mac Lane \cite{em2}}]
There is a bijection between equivalence classes of extensions of $A$ by $Q$ with action $\theta$, and elements of the group $H^2(Q, A, \theta)$.
\end{theorem}

Now we give the promised example of the dependence of the cohomology class on the choice of map $\pi\colon G \to Q$. This example also serves a secondary purpose, as an example of two extensions where $G_1 \cong G_2$ but the extensions are not equivalent according to the above definition(s).

\begin{example}[Isomorphic groups from non-equivalent extensions] \label{ex:isonotcong}
$\Z_9$ can be viewed as an extension of $\Z_3$ by $\Z_3$ in two ways. First, $\Z_3\hookrightarrow \Z_9$ by sending $1$ to $3$. Then define $\pi_i(1)=i$ for $i\in[2]$. To see that $\pi_1$ and $\pi_2$ yield non-equivalent extensions, consider any $\phi\in\aut(\Z_9)$. Let $k = \phi(1)$; since $\phi$ is an automorphism, $k$ is one of $\{1, 2, 4, 5, 7, 8\}$. In order for $\phi$ to be an equivalence of extensions, it must induce the identity on $\langle 3\rangle$, implying that $3=\phi(3)=3\phi(1)=3k\mod 9$, so $k \mod 3$ must be $1\mod 3$. On the other hand, to be an equivalence of extensions, we must also have $1=\pi_1(1)=\pi_2(\phi(1))=\pi_2(k)=2k\mod 3$, that is $k=2\mod 3$, giving a contradiction.
\end{example}

Given the preceding example, the question then becomes how to leverage cohomology for isomorphism testing (since what we'll ultimately care about is the isomorphism class of $G$, and not merely the equivalence class of any given extension with $G$ as its total group). This is what we tackle next.

\subsection{Main lemma for abelian characteristic subgroups}\label{subsec:main}

Recall that a characteristic subgroup is a subgroup invariant under all automorphisms. The analogous notion for isomorphisms (rather than automorphisms) is a function $\mathcal{S}$ that assigns to each group $G$ a subgroup $\mathcal{S}(G) \leq G$ such that any isomorphism $\varphi\colon G_1 \to G_2$ restricts to an isomorphism $\varphi|_{\mathcal{S}(G_1)}\colon \mathcal{S}(G_1) \to \mathcal{S}(G_2)$. We call such a function a \emph{characteristic subgroup function}.
%In line with other works in group theory, we call such a function a \emph{characteristic subgroup function}. 
Note that if $G_1=G_2$, this says that $\mathcal{S}(G_1)$ is sent to itself by every automorphism of $G_1$, that is, $\mathcal{S}(G_1)$ is a characteristic subgroup of $G_1$. Most natural characteristic subgroups encountered are characteristic subgroup functions, for example the center $Z(G)$, the commutator subgroup $[G, G]$, or the radical $\rad(G)$.

Let $\charfn$ denote a fixed characteristic subgroup function, and suppose we are 
given two groups $G_1, G_2$ such that $\charfn(G_1)$ and $\charfn(G_2)$ are both 
abelian. To determine how to use cohomology in isomorphism testing, we first examine the consequences of an isomorphism $G_1\cong G_2$, as an exercise in reverse engineering. Let 
$\gamma:G_1\to G_2$ be an isomorphism. By the definition of characteristic 
subgroup function, $\gamma(\charfn(G_1))=\charfn(G_2)$, so 
$\gamma$ induces $\gamma_1: \charfn(G_1)\to \charfn(G_2)$ and $\gamma_2: 
G_1/\charfn(G_1)\to G_2/\charfn(G_2)$. Using these isomorphisms, we identify 
$A=\charfn(G_2)=\charfn(G_2)^{\gamma_1}$ and $Q=G_2/\charfn(G_2)=(G_2/\charfn(G_2))^{\gamma_2}$.
%thus $\charfn(G_1)\cong \charfn(G_2)$ (identified as $A$) and 
%$G_1/\charfn(G_1)\cong G_2/\charfn(G_2)$ (identified as $Q$). 
Let $(\theta_i, f_i)$ be the extension data of $A \hookrightarrow G_i 
\twoheadrightarrow Q$, where $\theta_i:Q\to\aut(A)$ and $f_i\in Z^2(Q, A, 
\theta_i)$. As we have identified $A=\charfn(G_1)=\charfn(G_2)$ and 
$Q=G_1/\charfn(G_1)=G_2/\charfn(G_2)$, $\gamma$ induces some $\alpha\in\aut(A)$ 
and $\beta\in\aut(Q)$. We write $\theta_{i, q}$ as the shorthand for $\theta_i(q)$ 
for $i=1, 2$ and $q\in Q$. It can then be verified that for $q \in Q$ and $a \in 
A$,
\begin{equation}\label{eqn:action_pc}
\theta_{1,q}(a) = \alpha^{-1}(\theta_{2, \beta(q)}(\alpha(a))) =: \theta_2^{(\alpha, \beta)}(q)(a),
\end{equation}
and we record this as $\theta_1 = \theta_2^{(\alpha, \beta)}$, where $\theta_2^{(\alpha, \beta)}$ is defined as above.

It can be similarly verified that $[f_1] = [f_2^{(\alpha, \beta)}]$ as cohomology classes in $H^2(Q, A, \theta_1)$, where $f_2^{(\alpha, \beta)}(p,q) \defeq \alpha^{-1}(f_2(\beta(p), \beta(q)))$ for all $p, q \in Q$. In other words, we have:
\begin{equation}\label{eqn:cohom_pc}
f_1(p, q) = \alpha^{-1}(f_2(\beta(p), \beta(q))) + b_u(p,q)
\end{equation}
for some 2-coboundary $b_u \in B^2(Q, A, \theta_1)$. Note that Equation~\ref{eqn:action_pc} ensures $f_2^{(\alpha, \beta)}$ is indeed a 2-cocycle relative to $\theta_1$, \ie, in $Z^2(Q, A, \theta_1)$. This discussion leads to the following definition:

\begin{definition}\label{def:pc}
Let $A$ be an abelian group and $Q$ any group, and let $(\theta_1, f_1)$ and $(\theta_2, f_2)$ be two extension data for $A$-by-$Q$. Then the extension data are \emph{pseudo-congruent}\footnote{See Footnote~\ref{fn:pc} on page~\pageref{fn:pc}.} if there exists $(\alpha, \beta)\in \aut(A)\times\aut(Q)$, such that $\theta_1 = \theta_2^{(\alpha, \beta)}$ and $[f_1] = [f_2^{(\alpha, \beta)}]$, that is, Equations~(\ref{eqn:action_pc}) and (\ref{eqn:cohom_pc}) hold. In this case we write $(\theta_1, f_1)\pc (\theta_2, f_2)$.
\end{definition}

\begin{lemma}[Main Lemma, abelian case]%\label{lem:main}
Let $\charfn$ be a characteristic subgroup function. Given two finite groups $G_1$ and $G_2$, suppose $\charfn(G_1)$ and $\charfn(G_2)$ are abelian. Then $G_1\cong G_2$ if and only if both of the following conditions hold:
\begin{enumerate}
\item $\charfn(G_1) \cong \charfn(G_2)$ (which we denote by $A$) and $G_1 / \charfn(G_1) \cong G_2 / \charfn(G_2)$ (which we denote by $Q$); 
\item $(\theta_1, f_1)\pc (\theta_2, f_2)$, where $(\theta_i, f_i)$ is the extension data of the extensions $A \hookrightarrow G_i \twoheadrightarrow Q$. 
\end{enumerate}
\end{lemma}

\begin{proof}
The above discussion shows the only if direction. For the other direction, suppose we are given an abelian group $A$, a group $Q$, an action $\theta:Q\to\aut(A)$, and a 2-cocyle $f:Q\times Q\to A$, $f\in Z^2(Q, A, \theta)$. We shall need the following procedure of Eilenberg and Mac Lane \cite{em2} that takes $A$, $Q$, $\theta$ and $f$ as input, and outputs a group $H$ as an extension of $A$ by $Q$ with extension data $(\theta, f)$. We refer to this as the \emph{standard reconstruction procedure}. The set of group elements of $H$ is $A\times Q$. For $(a, p), (b, q)\in A\times Q$, the group operation $\circ_H$ is defined as
\[
(a, p)\circ_H (b, q)=(a+\theta_p(b)+f(p, q), pq).
\]
A simple but tedious calculation verifies that $A \hookrightarrow H \twoheadrightarrow Q$ is an extension with extension data $(\theta, f)$.

Getting back to our problem, from $(\theta_1, f_1)\pc (\theta_2, f_2)$, we can choose appropriate sections $s_i \colon Q \to G_i$ such that the corresponding 2-cocycles satisfy $f_1=f_2^{(\alpha, \beta)}$ in $Z^2(Q, A, \theta_1)$. Note that as $\theta_1=\theta_2^{(\alpha, \beta)}$, $f_2^{(\alpha, \beta)}\in Z^2(Q, A, \theta_1)$. Now apply the standard reconstruction procedure to $(\theta_i, f_i)$ to get $H_i \cong G_i$ (isomorphism follows from Eilenberg and Mac Lane \cite{em1}. It is then straightforward to verify that the bijection $\gamma\colon H_1 \to H_2$ defined by $\gamma((a, p))=(\alpha(a), \beta(p))$ is in fact an isomorphism.
\end{proof}

\subsubsection{Pseudo-congruence of extensions and extension data}\label{subsec:pc}
As with the the standard concept of equivalence, the standard concept of pseudo-congruence applies directly to group extensions themselves, rather than extension data as in our definitions. We use our definitions because the standard definitions seem to presuppose that the total groups are isomorphic, whereas in our setting the whole goal is to determine whether this is the case. However, we show in this section that the definitions are in fact equivalent (which is closely related to the Main Lemma~\ref{lem:main}). We present the standard definition here as it has more intuitive appeal and we believe it makes some discussions in the paper clearer, for example the proof of Theorem~\ref{thm:ecentradwsoc}.

Throughout this section, $A$ denotes an abelian (normal sub-)group. Non-abelian normal subgroups are handled in \Sec{main_lemma_nonab}.

\begin{definition} \label{def:pc_grp}
Two extensions $A \hookrightarrow G_i \twoheadrightarrow Q$ ($i=1,2$) of $A$ by $Q$ are \emph{pseudo-congruent}  if there is an isomorphism $\gamma\colon G_1 \to G_2$ such that $\gamma(A) = A$. In particular, $\gamma$ induces automorphisms $\alpha \in \aut(A)$ and $\beta \in \aut(Q)$.
\end{definition}

Pictorially, $G_1$ and $G_2$ are pseudo-congruent as extensions if there exist $\alpha\in\aut(A)$, $\beta\in\aut(Q)$ and $\gamma\in\Iso(G_1, G_2)$ such that the following diagram commutes:
\begin{equation*}
\xymatrix@C=50pt{
A \ar@{^{(}->}[r]^{\iota_1} \ar@{->}[d]^{\alpha}_{\cong} & G_1 \ar@{->>}[r]^{\pi_1} \ar@{->}[d]^{\gamma}_{\cong} & Q \ar@{->}[d]^{\beta}_{\cong} \\
A \ar@{^{(}->}[r]^{\iota_2} & G_2 \ar@{->>}[r]^{\pi_2} & Q
}
\end{equation*}
where $\iota_i$ is the injective homomorphism from $A$ to $G_i$ and $\pi_i$ is the surjective homomorphism from $G_i$ to $Q$ with $\ker(\pi_i)=\im(\iota_i)$. 

It is possible for the total groups $G_1$ and $G_2$ to be isomorphic without the extensions being pseudo-congruent. The following example was provided by Vipul Naik \cite{naik2}. This, finally, also serves as an example of the dependence on the choice of $A$.

\begin{example}[Isomorphic groups from non-pseudo-congruent extensions] \label{ex:isonotpc}
$A = \Z_{p^2} \times \Z_p \times \Z_p$, $Q = \Z_{p^2} \times \Z_p$, $G = \Z_{p^3} \times \Z_{p^2} \times \Z_p \times \Z_p$. In one extension, $\iota_1(a,b,c) = (pa, 0, b, c)$ and in the other $\iota_2(a,b,c) = (pa, pb, a\pmod{p}, c)$. To see that there is no automorphism of $G$ sending $\im \iota_1$ to $\im \iota_2$---and hence that the two extensions are not pseudo-congruent---note that $\im \iota_1$ contains elements that are $p$ times an element of order $p^3$ in $G$, but $\im \iota_2$ contains no such elements.
\end{example}

Despite the fact that the usual Definition~\ref{def:pc_grp} presupposes that the 
total groups are isomorphic, in fact it is equivalent to our 
Definition~\ref{def:pc}. The isomorphism of the total groups follows for free from 
pseudo-congruence of the extension data:

\begin{lemma}\label{lem:eq_pc}
Definitions~\ref{def:pc} and \ref{def:pc_grp} are equivalent. In detail: let $A \hookrightarrow G_i \twoheadrightarrow Q$ ($i=1,2)$ be extensions of $A$ by $Q$, and let $(\theta_i, f_i)$ be the corresponding extension data. Then $G_1$ and $G_2$ are pseudo-congruent as extensions of $A$ by $Q$ if and only if $(\theta_1, f_1)\pc (\theta_2, f_2)$.
\end{lemma}

\begin{proof}
Suppose that the extensions are pseudo-congruent (Definition~\ref{def:pc_grp}), and let $\gamma \in \Iso(G_1, G_2)$, $\alpha \in \aut(A)$, $\beta \in \aut(Q)$ be as in Definition~\ref{def:pc_grp}. It is readily verified that $\theta_1 = \theta_2^{(\alpha, \beta)}$ and $[f_1] = [f_2^{(\alpha, \beta)}]$, that is, that the extension data are pseudo-congruent under Definition~\ref{def:pc}.

Conversely, suppose the extension data are pseudo-congruent (Definition~\ref{def:pc}). Then the isomorphism $\gamma$ constructed in the proof of the Main Lemma~\ref{lem:main} satisfies the conditions of Definition~\ref{def:pc_grp}.
\end{proof}

\subsubsection{Some algorithmic problems arising from special cases of pseudo-congruence}\label{app:ext}
\iffalse
\jnote{Removed due to duplication with above}
Recall that two extensions $A\hookrightarrow G_j\twoheadrightarrow Q$, $j=1, 2$ are pseudo-congruent if there exist $\alpha\in\aut(A)$, $\beta\in\aut(Q)$ and $\gamma\in\Iso(G_1, G_2)$ such that the following diagram commutes:
\begin{equation*}
\xymatrix@C=50pt{
A \ar@{^{(}->}[r]^{\iota_1} \ar@{->}[d]^{\alpha}_{\cong} & G_1 \ar@{->>}[r]^{\pi_1} \ar@{->}[d]^{\gamma}_{\cong} & Q \ar@{->}[d]^{\beta}_{\cong} \\
A \ar@{^{(}->}[r]^{\iota_2} & G_2 \ar@{->>}[r]^{\pi_2} & Q
}
\end{equation*}
That is, if there exists an isomorphism $\gamma:G_1\to G_2$ such that $\gamma(A)=A$ and it induces $\alpha$ on $A$ and $\beta$ on $Q$. %Also recall that if $\gamma$ induces the identity maps on $A$ and $Q$ the extensions are called equivalent.
\fi

We describe two special cases of pseudo-congruence of extensions, explain the algorithmic problems corresponding to them, and indicate some of the solutions. The first one was discussed in \cite{rob_paper}. Consider the case when a pseudo-congruence $\gamma$ induces the identity map on $Q$ as follows:
\begin{equation*}
\xymatrix@C=50pt{
A \ar@{^{(}->}[r]^{\iota_1} \ar@{->}[d]^{\alpha}_{\cong} & G_1 \ar@{->>}[r]^{\pi_1} \ar@{->}[d]^{\gamma}_{\cong} & Q\ar@{=}[d] \\
A \ar@{^{(}->}[r]^{\iota_2} & G_2 \ar@{->>}[r]^{\pi_2} & Q
}
\end{equation*}
This corresponds to the algorithmic setting when enumerating $\aut(Q)$ is allowed (\Sec{autq}), as after fixing some $\beta\in\aut(Q)$ we are reduced to looking for $\alpha$ such that $G_1$ and $G_2$ are pseudo-congruent via $(\alpha, \beta)$. If the extension is split and $A\cong \Z_p^k$ is elementary abelian, this problem reduces to \modiso: the action of each $q\in Q$ can be expressed as a nonsingular matrix in $\GL(k, p)$. So suppose $Q=\{q_1, \dots, q_s\}$, and in $G_j$ the conjugation action of $Q$ is written as $\{M(j, i)\mid M(j, i)\in \GL(k, p)\}$ where $M(j, i)$ denotes the action of $q_i$ on $A$ in $G_j$. The problem of \actcomp then reduces to determining whether there exists $T\in\GL(k, p)$ such that $TM(1, i)=M(2, i)T$ for every $i\in[\ell]$. This is a special case of \modiso, which admits deterministic polynomial-time algorithms \cite{CIK97,BL08,IKS10}. At the other extreme, when the extension is central, Theorem~\ref{thm:quotient_list} solves \cohiso. At present it is not clear to us how combine these two procedures to solve \edpc as a whole, beyond the elementary abelian case (Theorem~\ref{thm:quotient_list_general}); see \Sec{future:abelian} for a discussion of the difficulties involved.

On the other hand, consider the case in which a pseudo-congruence $\gamma$ induces the identity map on $A$, but not necessarily on $Q$:
\begin{equation*}
\xymatrix@C=50pt{
A \ar@{^{(}->}[r]^{\iota_1} \ar@{=}[d] & G_1 \ar@{->>}[r]^{\pi_1} \ar@{->}[d]^{\gamma}_{\cong} & Q  \ar@{->}[d]^{\beta}_{\cong}\\
A \ar@{^{(}->}[r]^{\iota_2} & G_2 \ar@{->>}[r]^{\pi_2} & Q
}
\end{equation*}
The isomorphism problem here corresponds to the algorithmic setting when $\aut(A)$ 
is enumerable, and our goal is to find $\beta$ such that $G_1$ and $G_2$ are 
pseudo-congruent. Proposition~\ref{prop:warmup} and 
Remark~\ref{remark:action_triv} fall into this setting. 
Recently, G. Ivanyos and the second author \cite{IQ16} 
presented a randomized efficient algorithm to decide whether two alternating 
bilinear maps 
over a finite field of odd characteristic are isometric or not; given the 
connections between $p$-groups of 
class 2 and exponent $p$ and alternating bilinear maps (cf. \cite[Section 
3.4]{Wil09a}), this amounts to solving the problem of finding $\beta$ as above. 
That result, and the techniques therein, are inspired by the work of Brooksbank 
and Wilson \cite{BW12} who presented an efficient algorithm to compute the 
isometry group %\footnote{While being called an ``isometry group,'' it is actually the \emph{automorphism} version of the problem that is solved in \cite{BW12}.} 
 of a single 
alternating bilinear map (the automorphism version of the above problem). 
%\jnote{Added parenthetical remark and removed footnote}. 
These two works \cite{BW12, IQ16} together suggest that we can compute the coset 
of isometries 
between two alternating bilinear maps.%\ynote{Modified explanations regarding \cite{IQ16} and \cite{BW12}.}
%solved this problem 
%for $p$-group of class $2$ and exponent $p$, $p$ odd, based on 
%techniques from 
%\cite{BW12}. 
%The work of Brooksbank and 
%Wilson \cite{BW12} also falls into this setting: they presented an efficient 
%algorithm to compute the isometry group of a Hermitian bilinear map in a model 
%stronger than Cayley table model; given the connections between $p$-groups of 
%class 2 and exponent $p$ and alternating bilinear maps (cf. \cite[Section 
%3.4]{Wil09a}), this amounts to solving the problem of finding $\beta$ as above. 
%It would be very interesting to use the approach in \cite{BW12} to shed more 
%light 
%on $p$-groups of class $2$. %In fact, 
%\jnote{Can you be more specific? The way it's written now it sounds like you just 
%solve the same problem as BW12} 

\subsection{Some cohomological lemmas}\label{app:lem}
Recall that a group $G$ is perfect if $[G,G]=G$; as $[G,G]$ is always a normal subgroup, all nonabelian simple groups are perfect.

\paragraph{Proposition~\ref{prop:central}, restated.}
Suppose $T_1, \dots, T_\ell$ are perfect groups. Let $A$ be an abelian group, and $Q=\prod_{i\in[\ell]} T_i$. Let $G$ be a group with $Z(G)=A$ and $G/Z(G)=Q$. Let $U_i=AT_i$ be the full preimage of $T_i$ under the natural map $G \to G/Z(G)$. Then for $i, j\in[\ell]$, $i\neq j$, $[U_i, U_j]=1$.

\begin{proof}
Let $\pi:G\to G/Z(G)$ be the natural projection. Note that $Q=G/Z(G)$ and $T_i$'s are direct factors of $Q$. For $i\in[\ell]$, define $V_i$ to be the smallest normal subgroup of $G$ such that $\pi(V_i)=T_i$. Then $U_i=V_iZ(G)$, and for $i\neq j$, $[U_i, U_j]=\id$ if and only if $[V_i, V_j]=\id$.

As $T_i$ is perfect, $\pi([V_i, V_i])=\pi(V_i)$. Because of minimality of $V_i$, $V_i=[V_i, V_i]$. For $i\neq j$, $T_i\cap T_j=\id$, thus $[\pi(V_i), \pi(V_j)]=[T_i, T_j]=\id$ in $Q$, which implies that $[V_i, V_j]\subseteq Z(G)$. Now we have: (1) $[[V_i, V_j], V_j]\subseteq [Z(G), V_j]=\id$; (2) $[[V_j, V_i], V_j]=\id$ as $[V_i, V_j]=[V_j, V_i]$. Then Hall's three subgroup lemma \cite[Chapter 4, Proposition 1.9]{Suzuki2} gives that $[[V_j, V_j], V_i]=\id$. Finally noting that $V_j=[V_j, V_j]$ we have $[V_j, V_i]=\id$.
\end{proof}

\newtheorem*{direct_factor_lem}{Lemma~\ref{lem:direct_factor}, restated}
\begin{direct_factor_lem}
Let $A' \times A'' \hookrightarrow G \twoheadrightarrow Q$ be a central extension of $A' \times A''$ by $Q$. Let $p_{A'}\colon A' \times A'' \to A'$ be the projection onto $A'$ along $A''$. If there is a 2-cocycle $f\colon Q \times Q \to A' \times A''$ such that $p_{A'} \circ f \colon Q \times Q \to A'$ is a 2-coboundary, then $G$ is isomorphic (even equivalent as an extension of $A' \times A''$ by $Q$) to the direct product $A' \times (G / A')$.

Furthermore, given the Cayley table of $G$, $A'$ can be computed in polynomial time using linear algebra over abelian groups.
\end{direct_factor_lem}

\begin{proof}
We prove directly that $A' \unlhd G$, exhibit a complement of $A'$ in
$G$ and show that this complement is normal. At the end we show how to compute $A'$ using linear algebra.

We may assume without loss of generality that the image of $f$ lies entirely within $A''$. For if not, then we may add the 2-coboundary $p_{A'} \circ f\colon Q \times Q \to A' \hookrightarrow A' \times A''$ to $f$ to get an equivalent 2-cocycle satisfying this condition. Similarly, we may assume that $f$ is normalized so that $f(1,q) = f(q,1) = 0$ for all $q \in Q$.

We construct an extension equivalent to $G$ from the cocycle $f$ in the
usual way: the elements are $A' \times A'' \times Q$ as a set, with
multiplication given by (writing $A'$ and $A''$ additively):
\[
(a_1, a'_1, q_1)(a_2, a'_2, q_2) = (a_1 + a_2, a'_1 + a'_2 + f(q_1,
q_2), q_1 q_2)
\]
since the image of $f$ lies entirely in $A''$.  We also have
$(a,a',q)^{-1} = (-a, -a' - f(q,q^{-1}), q^{-1})$.

$A'$ is normal:
\begin{eqnarray*}
(a, a', q)^{-1} (a_0, 1, 1) (a,a', q) & = & (-a, -a' - f(q,q^{-1}),
q^{-1})(a_0 + a, a', q)  \quad \text{(since $f(1,q)=0$)} \\
 & = & (-a + a_0 + a, -a' - f(q,q^{-1}) + a' + f(q,q^{-1}, qq^{-1}) \\
 & = & (a_0, 0, \id_{Q}).
\end{eqnarray*}

$A'$ has a normal complement: as the image of $f$ lies entirely in
$A''$, it is readily verified that elements of the form $(0, a', q)$
are closed under product, hence form a subgroup of $G$ which is
isomorphic to $G/A'$ and intersects $A'$ only in the identity. Moreover,
this subgroup is normal. For consider conjugating one of its elements
by an arbitrary element of $G$: $(-a, -a' - f(q,q^{-1}), q^{-1}) (0,
a'_0, q_0) (a, a', q)$. From the multiplication rule above, it is
clear that the first coordinate of this product is just the sum of the
first coordinates of the three factors---namely, zero---whatever the
second and third coordinates are.

Finally, we show how to compute $A'$ from the Cayley table for $G$ using linear algebra over abelian groups. We give the proof in the case that $Z(G) = \Z_p^k$ is elementary abelian; the general case uses the same ideas as in \Sec{elem}. First compute $Z(G)$ (which is $A' \times A''$, but we do not yet know this decomposition of $Z(G)$, we are only promised it exists) and $Q = G/Z(G)$. Choose any set-theoretic section $s\colon Q \to G$ and compute the corresponding cocycle $f := f_s$. Let $M_f$ be the $k \times |Q|^2$ $\Z_p$-matrix corresponding to $f$ as in \Sec{autq}. We may view $M_f$ as a $\Z_p$-linear map from $Z(G) = \Z_p^k$ to $\Z_p^{Q \times Q}$. As in Proposition~\ref{prop:basis}, we may compute a basis of $B^2(Q, Z(G))$ that is a direct sum of bases for $B^2(Q, \Z_p)$, one copy for each row of $M_f$. The maximal $A'$ satisfying the conditions of the theorem is then the inverse image of $B^2(Q, \Z_p)$ under this map. Computing the inverse image of $B^2(Q, \Z_p)$ under the map $M_f^{T}\colon Z(G) \to \Z_p^{Q \times Q}$ is then just linear algebra over $\Z_p$.
\end{proof}

\section{Generalized Fitting subgroups of groups with central radical}\label{app:fitting}

\begin{definition}
A group $G$ is quasisimple if $G=[G, G]$ and $G/Z(G)$ is a nonabelian simple group. $G$ is m-quasisimple if $G=[G, G]$ and  $G/Z(G)$ is a direct product of nonabelian simple groups.
\end{definition}
Our m-quasisimple groups are Suzuki's ``semisimple groups'' \cite[Page~446]{Suzuki2}; we cannot use Suzuki's terminology as we have used ``semisimple'' for groups with no abelian normal subgroups. Although central-radical groups $G$ with $G/Z(G)$ a direct product of nonabelian simple groups need not be m-quasisimple groups (as $G$ need not be perfect), the difference is not much:

\begin{proposition}[{\cite[Ch.~6, corollary to Theorem 6.4]{Suzuki2}}]
Let $G$ be a group such that $G/Z(G)$ is a direct product of nonabelian simple groups. Then $G=Z(G)[G, G]$, and $[G, G]$ is an m-quasisimple group.
\end{proposition}

M-quasisimple groups are crucial for defining the generalized Fitting subgroups.
\begin{proposition}[\cite{Suzuki2}]
Let $H$ and $K$ be m-quasisimple normal subgroups of $G$, then $HK$ is m-quasisimple.
\end{proposition}
This motivates the following definition. Recall that the Fitting subgroup $F(G)$ of $G$ is the maximal nilpotent normal subgroup of $G$.
\begin{definition}
Let $G$ be a group. The \emph{layer} $E(G)$ of a group $G$, is the maximal m-quasisimple normal subgroup of $G$.
The \emph{generalized Fitting subgroup} $F^*(G)$ of $G$ is $E(G)F(G)$.
\end{definition}

\begin{proposition}
For any group $G$, if $\rad(G) = Z(G)$, then $\soc^*(G)=F^*(G)$.
\end{proposition}
\begin{proof}
As $\rad(G)=Z(G)$, $F(G)=Z(G)$. Let $D=[\soc^*(G), \soc^*(G)]$. So $D$ is m-quasisimple and $\soc^*(G)=Z(G)D=F(G)D$ (\cite[Ch.~6, corollary to Theorem 6.4]{Suzuki2}). Thus $D\subseteq E(G)$ and $\soc^*(G)\subseteq F^*(G)$.

To show $\soc^*(G)\supseteq F^*(G)$, for the purpose of contradiction, suppose $Z(G)D=\soc^*(G)\subsetneq F^*(G)=Z(G)E(G)$. Consider the decomposition of $E(G)$ into quasisimple groups (\cite[Ch.~6, Definition 6.8]{Suzuki2}) as $Q_1\cdot\ldots\cdot Q_d$, where $\cdot$ denotes central product, and $Q_i$ is subnormal in $G$. Without loss of generality assume $Z(G)Q_1\not\subseteq Z(G)D$. As $Q_1$ is subnormal in $G$, $G/Z(G)$ necessarily has $Q_1Z(G)/Z(G)$ as a subnormal group, contained in some minimal normal group $N/Z(G)\lhd G/Z(G)$ (\cite[Lemma 9.17]{Isaacs}). By assumption, $Q_1Z(G)/Z(G)$ is not contained in $DZ(G)/Z(G)=\soc^*(G)/Z(G)=\soc(G/Z(G))$, so $N/Z(G)$ is a minimal normal subgroup not contained in $\soc(G/Z(G))$, contradicting the definition of the socle.
\end{proof}

\section{Relationship with results on practical algorithms}\label{sec:compare}

%\ynote{Rewrite: It is not surprising that the complexity-theoretic results and 
%practical results from the computational group theory (CGT) community often 
%leverage the same underlying structure of the relevant group classes.}

It is not surprising that, when it comes to the group isomorphism problem, 
theoretical computer scientists and computational group theorists often leverage 
the same underlying structure of the relevant group classes, 
%\jnote{Rewritten:} 
even when they consider groups as being input by different kinds of data 
structures. %though their data structures for models of groups can be different. 
Here we discuss the 
relationship between these two sets of results. A general reference for CGT is the 
handbook \cite{HEO05}; algorithms in CGT are often implemented in {\sf 
Magma}~\cite{magma} and/or {\sf GAP}~\cite{GAP}.

%\ynote{Rewrite: ``The goal in CGT is to get practical algorithms, so the groups 
%are typically given as input by generating sets of permutations or matrices, or by poly-cyclic 
%presentations (for solvable groups).'' By saying that it is one major achievement, 
%hopefully we leave the room to say that CGT have achieved many other things.}
One major achievement in CGT, as reflected in the {\sf Magma} and {\sf GAP} as 
well as in \cite{HEO05}, is to implement a huge collection of practical routines 
to work with groups, when the groups are represented by generating sets of 
permutations or matrices, or by poly-cyclic 
presentations (for solvable groups).
In general, these encodings are of size 
poly-logarithmic in $|G|$, so in that context even a provable worst-case guaranteed 
running time of $O(|G|)$ is usually impractical. However, in order to achieve 
better running times, practically fast 
%\jnote{Changed from ``heuristic,'' since their methods are guaranteed to be 
%correct, they just don't have guaranteed good running times} 
methods are often used, even without guarantees on their running time. 
%\jnote{Added/changed the 
%rest of this paragraph:} 
In the general setting, the best-known worst-case 
guarantee on the running time of these practical algorithms is $|G|^{O(\log 
|G|)}$, even when the group is given succinctly \cite{wilsonConf}. %\footnote{See  Footnote~\ref{fn:CGT} on page~\pageref{fn:CGT}.}
Thus, improving the state of the 
art on large group classes to anything less than $|G|^{\Theta(\log |G|)}$ represents not 
only an asymptotic, worst-case improvement in the complexity, but potentially a 
practical improvement in many cases as well. Note that, although we sometimes 
phrase results in terms of the Cayley table model, essentially none of our results 
depend on this; it just means that running times that depend on $|G|$ are then 
also being expressed in terms of their dependence on the input size. 

%In contrast, from the complexity viewpoint, we are interested in provable worst-case guarantees, but our input is the entire Cayley table and we allow ourselves running time polynomial in $|G|$.

Regarding isomorphism testing algorithms in CGT, besides \cite{CH03} mentioned in Section~\ref{sec:cannonHolt}, some notable works include \cite{OBrien,BE99,LW12}, and many of the algorithms in this area are summarized in the theses \cite{smith, howden}. Very often isomorphism testing arises as a subroutine in the construction of all finite groups up to a certain order (up to isomorphism), as in \cite{Tau55,BE99,BEO02}. Recently, Wilson \etal have produced several results related to isomorphism of $p$-groups (sometimes reformulated in the context of Hermitian bilinear maps) in \cite{Wil09a,Wil09b,LW12,BW12, BMW15}, including worst-case guarantees (and sometimes even worst-case guarantees that are polynomial in the size of the succinct input!). The structure they are uncovering in $p$-groups is thus also notable from the complexity perspective, and there is likely more to be discovered in this direction.

Although the two communities often leverage the same structure that is present in various classes of groups, the worst-case guarantees often require further structural results on the group classes considered, in addition to further algorithmic results. For example, Besche and Eick had already considered the group classes in Le Gall's work \cite{Gal09} (the algorithm in \cite[Figure~4]{BE99}), but Le Gall's work was the first to prove a polynomial-time upper bound on an isomorphism algorithm for groups of the form $A \rtimes \Z_p$ with $A$ abelian and $p \nmid |A|$. A necessary ingredient in Le Gall's work is a detailed understanding of automorphism groups of abelian groups traced back to Ranum \cite{Ran07} (note: the citation here is to \emph{19}07), which was not needed in the practical setting of Besche and Eick. Another example is the polynomial-time algorithm for semisimple groups \cite{BCQ}, where a similar situation is described at the end of that paper, comparing it with the practical work of Cannon and Holt \cite{CH03}. For example, the algorithm in \cite{BCQ} required bounds on the orders of the transitive permutation groups other than $S_n$ and $A_n$.

\paragraph{Relations to the present work.} As mentioned above, our choice to focus 
on groups with central radicals is partially motivated by the strategy of Cannon 
and Holt \cite{CH03}. Another work of particular relevance is \cite{BE99}. There 
Besche and Eick considered construction of finite groups, and proposed three 
heuristics. To support one heuristic, they proposed the concept of ``strong 
isomorphism'' of groups, which can be viewed as a special case of our Main 
Lemma~\ref{lem:main_general} in their setting. We use the same structural results 
to support the approach, but as our goal is worst-case running time upper bounds 
of $\poly(|G|)$, we have more freedom to handle the 2-cohomology classes directly, 
as in Theorem~\ref{thm:quotient_list}; we also need Lemma~\ref{lem:prod} and 
Theorem~\ref{thm:GKKL} from \cite{GKKL}, which in turn allows us to apply 
algorithms for \lincode and \cosetint as in Theorem~\ref{thm:ecentradwsoc}. These 
ingredients are not present in \cite{CH03} nor \cite{BE99}, %\jnote{Changed:} 
as they were not needed to get practical algorithms. However, it is nonetheless 
possible that taking advantage of these ingredients, in combination with the 
practically-fast approaches from \cite{CH03, BE99}, could lead to further 
practical improvements.
%, which is natural as in their setting these ingredients would not be practical.

\end{document}